%% file: main.tex
\global\def\draftcontrol{0}
\pdfoutput=1
   \def\versionno{ Adinkras and Pure Spinors }

\documentclass[11pt]{article}
\input{packages}
\input{hacks}
\input{structure}

\input{macros}

\begin{document}

\title{Adinkras and Pure Spinors}

\pubnum{
arXiv:2404.07167
}

\date{April 2024}

\author{
Richard Eager$^{a}$,
Simone Noja$^{b}$,
Raphael Senghaas$^{c}$,
Johannes Walcher$^{c}$
\\[0.3cm]
\it $^a$ Kishine Koen, Yokohama 222-0034, Japan  \\
\it McLean, VA 22101, United States \\
\it Department of Physics, University of Maryland, College Park, MD 20742, USA \\ [0.1cm]
\it $^b$ Institute for Mathematics\\
\it Ruprecht-Karls-Universit\"at Heidelberg, 69120 Heidelberg, Germany \\[0.1cm]
\it $^c$ Institute for Mathematics \rm and \it Institute for Theoretical Physics\\
\it Ruprecht-Karls-Universit\"at Heidelberg, 69120 Heidelberg, Germany
}

\email{eager@mathi.uni-heidelberg.de\\
noja@mathi.uni-heidelberg.de\\
raphael.senghaas@gmx.de\\
walcher@uni-heidelberg.de}

\Abstract{The nilpotence variety for extended supersymmetric quantum mechanics is a cone over a
quadric in projective space. The pure spinor correspondence, which relates the description of
off-shell supermultiplets to the classification of modules over the corresponding hypersurface ring,
reduces to a classical problem of linear algebra. Spinor bundles, which correspond to maximal
Cohen--Macaulay modules, serve as basic building blocks. Koszul duality appears as a deformed version
of the Bernstein--Gel'fand--Gel'fand correspondence that we make fully concrete. We illustrate in
numerous examples the close relationship between these connections and the powerful graphical
technology of Adinkras, which appear as a decategorification of special
complexes on quadrics. We emphasize the role of R-symmetry for recovering higher-dimensional gauge
and gravity multiplets.
\begin{center}
{\it To Sylvester James Gates, with gratitude and admiration.}
\end{center}}

\makepapertitle

\pagestyle{paper}

\newpage
\setcounter{tocdepth}{2}
\tableofcontents
\setcounter{footnote}{0}

\newpage

\hfill\parbox[b]{9.9cm}{{\it
Par ma foi! Il y a plus de quarante ans que je dis de  \\
la prose sans que j’en susse rien, et je vous suis le plus  \\
oblig\'e du monde de m’avoir appris cela.} (Moli\`ere)}

\input{S1Introduction.tex}
\input{S2representations.tex}
\input{S3adinkras.tex}

\input{S4mcm.tex}

\input{S5examples.tex}

\input{S6BeyondAdinkras.tex}

\section*{Acknowledgements}
\addcontentsline{toc}{section}{Acknowledgments}
We would like to thank Fabian Hahner, Ingmar Saberi, Steffen Schmidt, for valuable discussions and
collaboration on related projects, and Daniele Faenzi for correspondence.\\
Special thanks to Sylvester James Gates, Chuck
Doran and Tristan H\"ubsch for discussions and continuous interest in our
endeavor.%
The authors also thank Owen Gwilliam, Kevin Iga and Stefan M\'endez-Diez for helpful comments.
R.E.\ thanks S.J. Gates for encouragement.
R.E. and R.S. thank Bard College, Annandale-on-Hudson, for hospitality and the
opportunity to present this work during the workshop \enquote{MathScape 2024: The Mathematics of
Supersymmetry}.
J.W.\ thanks the International Centre for Mathematical Sciences, Edinburgh, for support and hospitality during the ICMS Visiting Fellows programme where this work was completed. This work was supported by EPSRC grant EP/V521905/1.

This work is funded by the Deutsche Forschungsgemeinschaft (DFG, German Research Foundation) under
project number 517493862 (Homological Algebra of Supersymmetry: Locality, Unitary, Duality).
This work is funded by the Deutsche Forschungsgemeinschaft (DFG, German Research Foundation) under
Germany’s Excellence Strategy EXC 2181/1 — 390900948 (the Heidelberg STRUCTURES Excellence Cluster)

\appendix

\input{AppendixA}

\input{AppendixB}

\input{AppendixC}

\input{AppendixD}

\printbibliography[heading=bibintoc]

\end{document}

%% file: packages.tex
\usepackage{pgfplots}\pgfplotsset{compat=1.17} 
\usepackage{amsmath,amssymb,amsthm,mathtools,array,booktabs}
\usepackage{tikz}
\usetikzlibrary{shapes.geometric}
\usetikzlibrary{snakes,arrows,shapes}
\usepackage{tikz-cd}
\usepackage{xcolor}
\usepackage{graphicx}
\usepackage[cmtip,all]{xy}
\usepackage{enumitem}
\usepackage{csquotes}
\usepackage{comment}
\usepackage{float}
\usepackage{mathdots}
\ifnum\draftcontrol=1
\usepackage[inline]{showlabels}
\usepackage{showkeys}
\fi

\usepackage[citestyle=numeric,bibstyle=numeric,
giveninits=true,doi=false,url=false,isbn=false,
parentracker=false,
backend=biber,
maxnames=10,
date=year,
sorting=nty]{biblatex} 

\usepackage{hyperref}
\hypersetup{
  unicode,
  bookmarksnumbered,
  linktoc = all,
  hidelinks
}

%% file: structure.tex
\makeatletter
\renewcommand\section{\@startsection {section}{1}{\z@}%
                                   {-3.5ex \@plus -1ex \@minus -.2ex}%
                                   {2.3ex \@plus.2ex}%
                                   {\normalfont\large\bfseries}} 
\renewcommand\subsection{\@startsection{subsection}{2}{\z@}%
                                   {-3.25ex\@plus -1ex \@minus -.2ex}%
                                   {1.5ex \@plus .2ex}%
                                   {\normalfont\normalsize\bfseries}} 
\renewcommand\subsubsection{\@startsection{subsubsection}{3}{\z@}%
                                   {-3.25ex\@plus -1ex \@minus -.2ex}%
                                   {1.5ex \@plus .2ex}%
                                   {\normalfont\normalsize\it}} 
\renewcommand\paragraph{\@startsection{paragraph}{4}{\z@}%
                                   {-3.25ex\@plus -1ex \@minus -.2ex}%
                                   {1.5ex \@plus .2ex}%
                                   {\normalfont\normalsize\bf}} 
\renewcommand\subparagraph{\@startsection{subparagraph}{5}{\z@}%
                                   {-1.25ex\@plus -1ex \@minus -.2ex}%
                                   {0ex \@plus .2ex}%
                                   {\normalfont\normalsize\it}} 
\makeatother

\newtheorem{theorem}{Theorem}[section]
\newtheorem{proposition}[theorem]{Proposition}
\newtheorem{corollary}[theorem]{Corollary}

\newtheorem{construction}{Construction}[section]

\theoremstyle{definition}
\newtheorem{definition}[theorem]{Definition}
\newtheorem{lemma}[theorem]{Lemma}

\theoremstyle{remark}

\newtheorem*{remark}{Remark}

\numberwithin{equation}{section}

%% file: macros.tex
\def\lm{\lambda}
\renewcommand{\AA}{\mathbb{A}}
\newcommand{\CC}{\mathbb{C}}
\newcommand{\HHH}{\mathbb{H}}
\newcommand{\fp}{\mathbb{P}}
\newcommand{\RR}{\mathbb{R}}
\newcommand{\ZZ}{\mathbb{Z}}
\newcommand*{\defeq}{\mathrel{\vcenter{\baselineskip0.5ex \lineskiplimit0pt
                     \hbox{\scriptsize.}\hbox{\scriptsize.}}}%
                     =}

\newcommand\scalemath[2]{\scalebox{#1}{\mbox{\ensuremath{\displaystyle #2}}}}

\def\fd{\mathfrak{d}}
\def\fg{\mathfrak{g}}

\def\fp{\mathfrak{p}}
\def\fr{\mathfrak{r}}
\def\ft{\mathfrak{t}}
\def\so{\mathfrak{so}}
\def\susy{\mathfrak{susy}}

\def\kk{\mathbf{k}}
\def\QQ{\mathcal{Q}}
\def\DD{\mathcal{D}}
\def\HH{\mathcal{H}}

\def\calA{\mathcal{A}}
\def\calC{\mathcal{C}}

\def\calN{\mathcal{N}}

\def\calV{\mathcal{V}}

\def\pzA{\mathpzc{A}}
\def\pzE{\mathpzc{E}}
\def\pzV{\mathpzc{V}}
\def\pzb{\mathpzc{b}}
\def\pzc{\mathpzc{c}}
\def\pzh{\mathpzc{h}}
\def\pzp{\mathpzc{p}}
\def\pzs{\mathpzc{s}}
\def\pzt{\mathpzc{t}}
\def\cliff{C\ell}

\def\del{\partial}

\newcommand\topwedge[1]{\mbox{\text{\Large$\wedge$}}^{\hskip -1pt #1}}
                     
\DeclareMathOperator{\Ext}{Ext}

\DeclareMathOperator{\Sym}{Sym}

\DeclareMathOperator{\coker}{coker}

\DeclareMathOperator{\RHom}{RHom}

\DeclareMathOperator{\dg}{deg}
\DeclareMathOperator{\Mat}{Mat}

\DeclareMathOperator{\Spec}{Spec}

\DeclareMathAlphabet{\mathpzc}{OT1}{pzc}{m}{it}

\def\be{\begin{equation}}
\def\ee{\end{equation}}

\def\bc{\begin{center}}
\def\ec{\end{center}}

\def\bcd{\begin{tikzcd}}
\def\ecd{\end{tikzcd}}

\def\ba{\begin{align*}}
\def\ea{\end{align*}}

\def\bmat{\begin{pmatrix}}
\def\emat{\end{pmatrix}}

\definecolor{darkgreen}{RGB}{20,120,10}
\definecolor{gold}{RGB}{240,190,0}

\newtheorem{example}{Example}[theorem]

%% file: S1introduction.tex
\section{Introduction}

Supersymmetric dynamics is governed by a super Lie algebra usually defined in terms of some number,
$N$, of anti-commuting supercharges, $Q_i$, that square to the Hamiltonian, $H$, in other words, they satisfy the relations
\begin{equation}
\label{commH}
\{Q_i , Q_j \} = 2 \delta_{ij} H \,,\qquad i,j = 1,\ldots N,
\end{equation}
where $\delta_{ij}=1$ when $i=j$, and $0$ otherwise. If, as in ordinary quantum mechanics, symmetry
generators are represented by Hermitian operators, the relations \eqref{commH} should be interpreted
over the real numbers. Irreducible representations are finite-dimensional, and, when $H>0$, can be
constructed in a basis-independent way as modules over the real Clifford algebra
$\cliff(\RR^N)\cong\Mat(2^n,\AA)$, where, famously, $\AA=\RR, \RR\oplus\RR, \RR, \CC,
\HHH,\HHH\oplus\HHH,\HHH,\CC$ for $N=0,\ldots, 7\bmod 8$, respectively, and $n$ is an integer such
that $\dim_\RR \Mat(2^n,\AA)=2^N$.\footnote{In actuality, this is not the most useful approach. For
starters, when $N$ is odd the simple Clifford modules are not $\ZZ_2$-graded (because the
non-trivial central element of the Clifford algebra is odd). Requiring a well-defined Fermion number
$(-1)^F$ leads after appropriate rescaling to an enhancement to the next even number of real
supercharges. (Specifically, by $Q_{N+1} = (-1)^{N(N+1)/4} H^{(1-N)/2}\cdot Q_1\cdots Q_N \cdot
(-1)^{F}$.) Second, one needs to worry about the existence of a suitable invariant Hermitian inner
product after complexification. This is most easily taken care of by the tensor product of $N/2$
copies of the standard and obvious $2$-dimensional representation of $\{Q,Q^\dagger\}=2H$. The
important point for now is that the unitary representations of \eqref{commH} are essentially unique,
of complex dimension $\lceil N/2 \rceil$. In fact, even the Hermiticity of the supercharges cannot
really be justified without appeal to higher-dimensional Lorentz, R-, or gauge symmetry.} When
$H=0$, the only unitary representation of \eqref{commH} is the trivial one, while $H<0$ is
impossible. In other words, unitary representations of \eqref{commH} are essentially unique.

When \eqref{commH} arises by ``dimensional reduction'' of the (possibly extended) $d$-dimensional
super Poincar\'e algebra, $N$ is the (real!) dimension of the underlying spinorial representation,
and usually an even number. More precisely, in the presence of either central charges or
non-trivial momentum (in particular, for massless representations), the $d$-dimensional bracket
might be degenerate, and the representations factor through its kernel. This leads to a somewhat
intricate, but well-known and easily parameterizable pattern of ``BPS representations'' of various
sizes.

In quantum field theory, one is interested in realizations of supersymmetry that are in addition
``local'' in the sense that they are supported by observables satisfying a specific set of intricate
conditions with respect to the standard topology, and usually also a metric structure, of the
underlying spacetime manifold. These conditions, from which we shall refrain to spell out here, are
formally satisfied in the representation as infinite-dimensional integrals of local functionals
(pointwise polynomial functions of fields and their derivatives) over a suitable ``field space,''
provided that field space (and the measure on it, defined with the help of the action functional)
carries a manifest ``local, off-shell'' representation of the super Poincar\'e algebra. Referring to
\cite{Eager22} for a thorough recent description of these notions (which are important background
for our present paper), and to appendix \ref{Appendix1} for some details on our conventions, we
record the intuitive idea that fields should be sections of super vector bundles on which the
algebra act by derivations, preserving the Lagrangian density, and thence close, without appeal to
the equations of motion. In general, this is a rather difficult problem to capture mathematically,
even before one comes to the definition of the integral.

Interestingly, as spelled out most clearly by Gates et al.\ \cite{GatesChallenge}, the challenge of
reconstructing off-shell representations of the super Poincar\'e algebra in arbitrary dimension has
a shadow in its dimensional reduction to $d=1$ (i.e., quantum mechanics, eq.\ \eqref{commH}) that is
already far more complicated than one might expect based on the simple results recorded above. This
can be thought of as a consequence of replacing on-shell unitarity with the grading (or rather,
filtering, see below) of fields and observables by ``engineering dimension,'' and leads in turn to
the beautiful graphical method of ``Adinkras'' \cite{Faux05}, further developed and studied, for
example, in \cite{Doran06,Doran07}.

More recently, a rather different approach to the problem that works uniformly across dimensions has
been developed as a generalization of the so-called pure spinor formalism \cite{BerkovitsParticle,
Cederwall:2001dx, CederwallReview}. The idea, as codified in \cite{Eager22,ElliottDerived}, is that
the ubiquitous differential constraints that preclude the application of ordinary superfields for
large number of supercharges can be ``resolved'', in the sense of homological algebra, over an
auxiliary geometric space known as the ``nilpotence variety'' $Y=\{ Q | Q^2=0 \}$ that parameterizes
\cite{Eager18} square-zero supercharges in the {\it complexified} super Poincar\'e
algebra.\footnote{Indeed, as emphasized above, real supercharges square to zero on-shell only on the
actual kernel of the super Lie bracket, and must then vanish identically. Off-shell, much more is
possible.} This construction plays particularly well with the recent interest in understanding, as
toy models and possibly for exact solution, all topologically and holomorphically twisted cousins of
various supersymmetric field theories \cite{SaberiSCA,6Dmultiplets}, and especially supergravity
\cite{Eager11d,Raghavendran11d,Hahner11d}.

The purpose of the present paper is to explain the relationship between off-shell multiplets, as
witnessed by Adinkras, and the pure spinor formalism, restricted to $d=1$. Specifically, after
reviewing the basics of Adinkras, in section \ref{Adinkras} we show that every Adinkra is associated
with a graded linear complex of modules over the nilpotence variety of the $d=1$ supersymmetry
algebra, which can be interpreted -- in a suitable sense -- as the input of the
pure spinor functor. In one dimension the nilpotence variety is particularly nice: it is defined by a quadric hypersurface,
\begin{equation}
    q_N = \sum\limits_{i=1}^N \lambda_i^2 = 0.
\end{equation}
The tools developed in appendix \ref{CplxAd} as the mathematical foundation for
the correspondence between Adinkras and complexes on quadrics are interesting in
their own rights. In particular, Adinkras can be understood to represent objects in the
heart of a non-standard t-structure on the bounded derived category of the
quadric and we prove a generalization of Eisenbud's unrolling procedure for the
complexes associated to Adinkras.

This embeds the graphical technology of Adinkras into an incredibly rich
algebro-geometric environment, reviewed in section \ref{mcm}. In fact, in this easy setting, the pure
spinor correspondence takes on a beautiful (and concrete!) meaning. Modules on the nilpotence
variety correspond to modules over the universal enveloping algebra of the supersymmetry algebra,
and vice versa. This equivalence, which is a \enquote{deformation} of the celebrated
Bernstein--Gel'fand--Gel'fand correspondence \cite{BGG} between sheaves on projective spaces and
modules on a related exterior algebra\footnote{Notice that the standard Bernstein--Gel'fand--Gel'fand
correspondence can in fact be thought of as a pure spinor superfield formalism in $d=0$.}, can be
interpreted as a very concrete form of Koszul duality, relating a geometric category to a
representation-theoretic one \cite{Buchweitz2, Buchweitz06}. On the geometry side, a
special class of modules is singled out by this equivalence: these are the so-called spinor bundles
\cite{Ottaviani88}, which correspond to one of the most fundamental types of Adinkras --- the valise Adinkras. On the other hand, starting from a valise Adinkra, one can construct other
Adinkras by raising or lowering its vertices. In section \ref{sec:flipping} we give this procedure a precise geometric meaning, as \enquote{taking
quotients} (in a suitable sense) in a derived category. Our endeavor, though, is not limited to
abstract theoretical constructions and geometric framing. We put to good use the homological methods
developed to construct and study several examples of Adinkra multiplets in section \ref{props}.
\emph{Inter alia}, we show in section \ref{1_7_7_1} that a specific class of Adinkras can be represented by constrained
superfields, as it allows for embeddings into the free superfield. Further, in section \ref{sec:counterexample} we offer evidence that
cohomology provides new useful \enquote{refined} invariants to distinguish Adinkras, and we describe the non-Adinkras examples
considered in \cite{Hubsch13, Doran13} as suitable extension classes of Adinkras in section \ref{ExtAdi}.

In the last section of the paper -- section \ref{beyond} --, we make contact with higher-dimensional supersymmetric gauge and gravity theories, by
recovering multiplets via dimensional reduction and emphasizing the role of R-symmetry -- this is a
promising route for future research endeavors. Finally, in a different direction, we conclude our study by introducing
a generalization of Adinkras to possibly non-positive definite quadratic forms, and we use the
pure spinor formalism in combination with projective geometry to classify families of $d=1$ multiplets and exhibit examples of multiplets that
do not arise from Adinkras.

We hope that our endeavor might serve a two-fold purpose. On the one hand, uncovering the rich
geometry disclosed under the beautiful representation theoretic method of Adinkras, which -- in an unexpected development --
prove themselves a formidable workshop for derived geometric constructions.
On the other hand, offering an easy enough (yet non-trivial!) physically relevant setting to understand in a very
concrete manner the pure spinor correspondence \cite{ElliottDerived} as an instance of Koszul
duality. In view of this, we hope this serves to further illustrate the use of the pure spinor
formalism in higher dimensions.

%% file: S2representations.tex
\section{Off-shell Representations of the Supersymmetry Algebra}
\label{representations}

Let $\kk$ be the field of real or complex numbers, and $N$ a positive integer.

\begin{definition}[Supersymmetry algebra] \label{SUSYalg} The ($N$-extended) supersymmetry algebra,
denoted $\susy$, is the $\ZZ$-graded super Lie algebra over $\kk$ whose underlying vector space is
given by
\begin{equation}
\susy =
\underbrace{\phantom{\langle}\fd\phantom{\rangle}}_{\mbox{\tiny{deg 0}}} \oplus \underbrace{\langle Q_1, \ldots , Q_N \rangle}_{\mbox{\tiny{deg 1}}} \oplus \underbrace{\langle H \rangle}_{\mbox{\tiny{deg 2}}}.
\end{equation}
where $\langle\cdots\rangle$ indicates linear span, $\fd$ is the one-dimensional Lie algebra acting
as the grading, and the fundamental bracket is
\begin{equation}
\tag{\ref{commH}}
\{Q_i, Q_j \} = 2\delta_{ij} H.
\end{equation}
for $i,j=1,\ldots N$. The largest automorphism group preserving the bracket is $Spin(N)$ and is called R-symmetry.
\end{definition}
\begin{remark}
As alluded to in the introduction, and further reviewed in appendix \ref{Appendix1}, the positively
graded part, $\ft = \susy_{>0}$, of the supersymmetry algebra can be viewed as ``reduction'' of a
higher-dimensional super-translation algebra $\ft^{d,\calN}\subset \fp^{d,\calN}$, inside
$d$-dimensional ($\mathcal{N}$-extended) super Poincar\'e. We prefer not to address $\susy$ as ``1d super Poincar\'e''
because ${\rm Spin}(1) \cong \ZZ/2$ really is non-relativistic, and cannot distinguish scalars,
vectors, and spinors. The $\so(N)$ symmetry algebra will play an important role in our story, but it
need not be represented in full, and hence is not included in the algebra.  We will see examples
where only a subgroup of the R-symmetry is preserved when the subgroup arises from the dimensional
reduction of the Lorentz and R-symmetry of a $d$-dimensional supersymmetry algebra. The largest
subalgebra $\mathfrak{r} \subseteq \so(N)$ that can be added to the degree-zero part of the algebra,
such that a multiplet is a representation of the extended algebra serves as an important
invariant. The seemingly unnatural choice of basis in the definition of $\susy$, restriction on the
coefficients, and signature of the quadratic form (when $\kk=\RR$), are dictated by convenience and
physical considerations.
\end{remark}

\subsection{Field multiplets and superfields} \label{FmSf}

To construct physically interesting quantum theories that realize the supersymmetry algebra, we look
for representations of one-dimensional supertranslations, $\ft$, on ``spaces of one-dimensional
fields'', {i.e.}, functions of a time variable $t$ valued in a finite-dimensional super vector space
$V$ over $\kk$, on which $H$ acts as the generator of translations,\footnote{The sign is
conventional \cite{DeligneClassical}.}
\begin{equation}
\label{conventional}
H = -\HH = - \del_t = - \frac{d}{dt}
\end{equation}
Intuitively, $t$ should be real, and at least all differentiable functions allowed. Since our
algebraic considerations will however focus on $V$ and require complexification of $\susy$, it is
more natural to restrict to analytic or even polynomial functions of $t$ as a formal variable, {i.e.},
work with $\calV=V\otimes_\kk\kk[t]$ as the space of fields, while still thinking of $\kk[t]$ as the
ring $\calC^\infty(\RR,\kk)$ of smooth, $\kk$-valued functions of time. At face value, such
representations are not compatible with the grading, i.e., do not lift to representations of
$\susy$, before fields and their derivatives are viewed as distinct observables in a quantum field
theory. Since this is not our primary goal in this paper, we will employ the useful short-cut to
``diagonalize time-translations by a formal Fourier transform'', in other words, work over $\kk[H]$
as universal enveloping algebra of the degree $2$ subalgebra \cite{Doran06}. This maneuver, which
is mediated by the vector space $V$, is accompanied by the identification of the physicist's time
derivative with the action of $H$:
\begin{equation}
\label{accompanied}
\text{For $\phi\in\calV$ as $\kk[t]$-module }
\quad -\dot \phi \;\; \longleftrightarrow \;\; H \phi \quad
\text{ for $\phi\in\calV$ as $\kk[H]$-module}
\end{equation}

\begin{definition}[Off-shell representations]
    \label{inthesense}
An off-shell representation, field multiplet, or simply multiplet, of the supersymmetry algebra is a
$\susy$-supermodule $\calV$ that restricts (after inducing a $U(\ft)$-module) to a finitely generated
free $\kk[H]$-module. The grading of $\calV$ by $\fd\subset\susy$, which we denote by
$\dg$, is assumed to be integral and defines what is known as twice the engineering or mass
dimension.
\end{definition}

Note that the $\ZZ$-grading on an off-shell representation descends to a $\ZZ$-grading on the
underlying super vector space $V\cong \calV/H\calV$. The generators of $V$ can be thought of as
component fields of the multiplet described by $\calV$ and the grading on $V$ captures their mass
dimension.

To see this in action in the simplest example of $N=1$ (with $Q\equiv Q_1$ and $Q^2=H$), note that
if $v \in V$ is a homogeneous element, then $v \otimes 1$ is homogeneous in $V \otimes \kk[H] \cong
\calV$. Acting with $Q^2$ maps $v \otimes 1$ to $v \otimes H$. This leaves two possibilities for the
transformation under $Q$. Either
\\[.2cm]
1. $Q \cdot v \otimes 1 = w \otimes 1$ where $w$ is homogeneous of $\dg(w)=\dg(v) + 1$ and $Q \cdot w \otimes 1 = v \otimes H$; or
\\[.2cm]
2. $Q \cdot v \otimes 1 = w \otimes H$ where $w$ is homogeneous of $\dg(w)=\dg(v) - 1$ and $Q \cdot w \otimes 1
= v \otimes 1$. \\[.2cm]
These two possibilities are isomorphic, unless distinguished by the grading. Using standard physics
notation \eqref{accompanied}, we would write \label{physics}
\\[.2cm]
1. $Qx= \psi$, $Q\psi = -\dot x$ if $v \leftrightarrow x$ is even, and $w\leftrightarrow \psi$ is
odd; or
\\[.2cm]
2. $Q\chi = F$, $Q F = -\dot\chi$ if $v \leftrightarrow \chi$ is odd and $w\leftrightarrow F$ is
even.
\\[.2cm]
Namely, in option 1, the boson is the lowest, while in option 2, the boson is the highest component
of the multiplet. This is reflected in the simplest invariant actions being\footnote{We are
displaying these to ensure consistent sign conventions.}
\begin{equation}
\label{reflected}
S_1 = \int dt \biggl ( \frac 12 \dot x^2 + \frac 12 \psi\dot\psi \biggr)
\,,\qquad
S_2 = \int dt \biggl( \frac 12 \chi\dot\chi + \frac 12 F^2 \biggr)
\end{equation}
for the two options, respectively.

As is well-known, a natural source of field multiplets are ``components of superfields''.
Unconstrained superfields are simply smooth functions on superspace, $\Phi\in
\calC^\infty(\RR^{1|N},\kk)\cong \calC^\infty(\RR,\kk)\otimes \topwedge{\bullet}(\ft_1^{\,\vee})$,
alternatively $\kk[t,\theta_1,\ldots, \theta_N]$, where $\theta_i$ are anti-commuting linear
coordinates, possibly with coefficients twisted to adjust mass dimension and parity. The space of
superfields inherits a representation of $\ft$ from infinitesimal left-translations by the
supervector fields
\begin{equation}
\label{superconfusing}
Q_i = \QQ_i = \del_{\theta_i} + \theta_i \del_t \,,
\qquad
\{\QQ_i,\QQ_j\} = 2\delta_{ij} \HH
\end{equation}
which thereby induces a representation on the component fields (coefficients of $\theta$-homogeneous
pieces of $\Phi$). For example, when $N=1$, the components of the bosonic superfield $X = x + \theta
\psi \in \calC^\infty(\RR^{1|1},\kk)$ transform like those in option 1 above, while the components
of the Fermi superfield $\Upsilon = \chi + \theta F\in \calC^\infty(\RR^{1|1},\kk[1])$ correspond to
option 2. The actions \eqref{reflected} can be written
\begin{equation}
    \label{manifest}
S_1 = - \frac{1}{2} \int dt d\theta \dot X \DD X
\,,\qquad
S_2 = \frac{1}{2} \int dt d\theta \Upsilon \DD \Upsilon
\end{equation}
where $\DD=\del_\theta-\theta \del_t$ is the familiar generator of right-translations that
anti-commutes with $\QQ$ and thereby makes the invariance of \eqref{manifest} manifest.

For general $N$, the unconstrained or free bosonic superfield $X\in \calC^\infty(\RR^{1|N},\kk)$ has
$2^N$ independent components. The associated multiplet corresponds by the above construction to the
real Clifford algebra $V=\cliff(\RR^N)$ as a module over itself -- we will give a further geometric characterization of
this multiplet in the following.

\subsection{The chiral superfield}

For example, when $N=2$, the supersymmetry transformations of the components of the free superfield
$X= x + \theta_1 \psi + \theta_2 \chi + \theta_2\theta_1 F$ induced by $\QQ_1$ and $\QQ_2$ from
\eqref{superconfusing} are explicitly
\begin{equation}
    \label{N2free}
    \begin{array}{rclrclrclrcl}
    Q_1 x &=& \psi \qquad\qquad & Q_1 \psi &=& -\dot x\qquad\qquad
    & Q_1 \chi &=& F \qquad\qquad & Q_1 F & =& -\dot \chi
    \\
    Q_2 x &=& \chi & Q_2 \chi &=& -\dot x
    & Q_2 \psi & =& -F & Q_2 F &=& \dot\psi
    \end{array}
\end{equation}
Under the two $N=1$ subalgebras, this decomposes in an obvious way into a bosonic and Fermi
multiplet with, however, one additional sign to ensure that $Q_1$ and $Q_2$ anti-commute.

But supersymmetry would not be cool if this were the only option. By repeating the steps from the
previous subsection, one finds that also the transformations
\begin{equation}
    \label{N2chiral}
    \begin{array}{rclrclrclrcl}
    Q_1 x &=& \psi \qquad\qquad & Q_1 \psi &=& -\dot x\qquad\qquad
    & Q_1 y &=& \chi \qquad\qquad & Q_1 \chi & =& -\dot y
    \\
    Q_2 x &=& -\chi & Q_2 \chi &=& \dot x
    & Q_2 y & =& \psi & Q_2 \psi &=& -\dot y,
    \end{array}
\end{equation}
which can be understood from a combination of 2 bosonic $Q_1$-multiplets with components $(x,\psi)$
and $(y,\chi)$, satisfy the supersymmetry algebra. For obvious degree reasons, this multiplet
cannot be isomorphic to \eqref{N2free}. However, it can be embedded into a combination of two free
superfields
\begin{equation}
    \label{hasbeenmoved}
    \begin{split}
    X &= x +  \theta_1\psi + \theta_2\psi_2  +  \theta_2 \theta_1 F\\
    Y &= y + \theta_1\chi  + \theta_2\chi_2  +  \theta_2 \theta_1 G
    \end{split}
\end{equation}
subject to the constraints
\begin{equation}
    \label{CR}
\DD_1 X = \DD_2 Y \qquad \DD_2 X = -\DD_1 Y
\end{equation}
where $\DD_i =\del_{\theta_i} - \theta_i \del_t$. Indeed, solving the constraints, we have
\begin{equation}
    \begin{split}
        X&= x + \theta_1\psi - \theta_2 \chi + \theta_2\theta_1 \dot y
        \\
        Y&= y + \theta_1\chi + \theta_2 \psi - \theta_2\theta_1 \dot x
    \end{split}
\end{equation}
and one easily verifies that $\QQ_1$, $\QQ_2$ generate \eqref{N2chiral}. It comes as no surprise
that the constraints \eqref{CR} have the form of Cauchy-Riemann equations, and, assuming complex
coefficients and coordinates
\begin{equation}
    \label{introduced}
\theta = \frac{1}{\sqrt{2}} \bigl( \theta_1 + i \theta_2)
\qquad
\bar \theta = \frac{1}{\sqrt{2}} \bigl( \theta_1 - i \theta_2)
\end{equation}
can be written in the form
\begin{equation}
\overline{\DD} Z = 0
\end{equation}
where
\begin{equation}
\overline{\DD} = \frac{1}{\sqrt{2}} \bigl(\DD_1 + i \DD_2) = \del_{\bar\theta} - \theta \del_t
\end{equation}
on the single superfield
\begin{equation}
Z = X+i Y = z + \theta \xi + \bar\theta \theta \dot z
\end{equation}
where $\xi = \sqrt{2}(\psi+ i\chi)$. This is referred to as the ``chiral multiplet''.
\footnote{Specifically, this multiplet arises by dimensional reduction from the 2d $(0,2)$ chiral
\cite{distler} or 1d $\calN=2$ ``boundary chiral'' multiplet \cite{lerche}.}

\subsection{Adinkras}

Faux and Gates introduced a graphical technology called ``Adinkras'' in \cite{Faux05} to describe
how off-shell representations $\calV=V\otimes \kk[H]$ of the supersymmetry algebra can be assembled
from collections of representations under the $N=1$ subalgebras that come with the Definition
\ref{SUSYalg} \footnote{For earlier work, see \cite{Gates:1995pw, Gates:1995ch,Gates:2002bc}}. The construction begins by associating nodes of a finite graph to basis elements of a
suitable homogeneous basis of $V$, with a black node for a boson and a light gray node for a fermion.  Each node is placed at a
vertical height that corresponds to the grading by engineering dimension.
For every two basis elements that are related under $Q_i$, a colored line is drawn between the two elements with a unique color for each index $i$.  In total, there are $N$ different
colors that distinguish the $N=1$ subalgebras.
The line is solid if the
transformation of the basis elements $v,w$ with $\dg(w) = \dg (v)+1$ is $Q\cdot v\otimes 1 =
w\otimes 1$, $Q\cdot w\otimes 1 = v\otimes H$ as on page \pageref{physics}, and dashed if $Q \cdot
v\otimes 1=-w\otimes 1$, $Q\cdot w \otimes 1 = - v\otimes H$, as played a role in \eqref{N2free} and
\eqref{N2chiral}. It is easy to see that there is some inherent ambiguity in the distribution of
signs --- inverting the sign of any given basis element of $V$ exchanges solid and dashed
lines incident to that node. There are also rather strong conditions on graphs that correspond to
multiplets that we specify momentarily. However, not all multiplets arise from such graphs. For example, when $N=1$ itself, we
have four different adinkras depicted in Fig.\ \ref{n1_koszul}. Two of these correspond to the free
bosonic multiplet, the two others to the Fermi multiplet.
\begin{figure}[ht!]
    \centering
    \includegraphics[height=4cm]{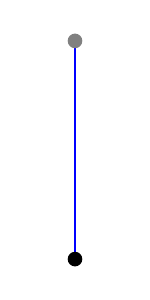}
    \qquad \quad
    \rotatebox[origin=c]{180}{%
    \includegraphics[height=4cm]{graphics/n1_1_1.pdf}}
    \qquad \quad
    \includegraphics[height=4cm]{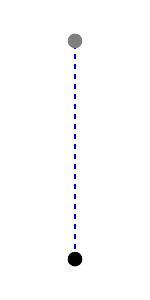}
    \qquad \quad
    \rotatebox[origin=c]{180}{%
    \includegraphics[height=4cm]{graphics/n1_1_1_minus.pdf}}
    \caption{Adinkras for $N = 1$.}
    \label{n1_koszul}
\end{figure}

The free bosonic and chiral multiplet for $N=2$ discussed in the previous subsection are captured by
the adinkras shown in the left and right of Fig.\ \ref{n2_valise}, respectively.
\begin{figure}[ht!]
    \centering
    \includegraphics[height=5cm]{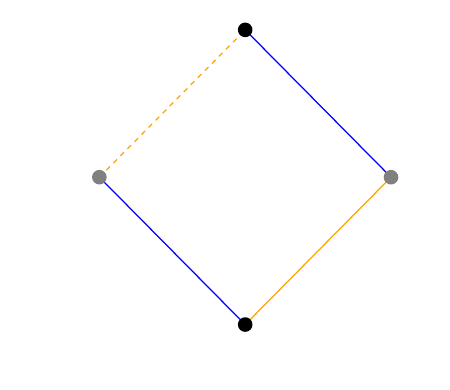}
    \qquad
    \raisebox{.5cm}{
    \includegraphics[height=4cm]{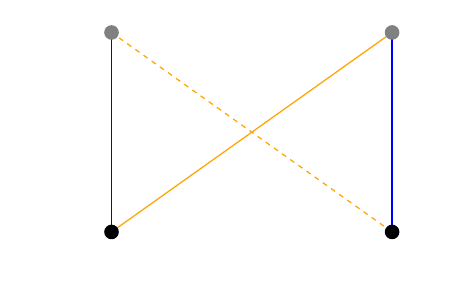} \vspace{.5cm}}
    \caption{Adinkras for free bosonic and chiral superfield, $N=2$.}
    \label{n2_valise}
\end{figure}

More formally, we can give the following definition.
\begin{definition}[Adinkra]
\label{N_Adinkra}
An Adinkra is a finite simple graph $\pzA$ equipped with two functions,
$\vert \cdot \vert : \pzV \rightarrow \ZZ/2$ (parity or ``statistics'') and
$\pzh : \pzV \rightarrow \ZZ$ (degree or ``dimension''),
on the set $\pzV$ of vertices, and two functions
$\pzc : \pzE \rightarrow \{1, \ldots, N \}$ (the ``color'') and
$\pzp : \pzE \rightarrow \ZZ/2$ (the ``dashing'')
on the set $\pzE$ of edges, subject to the following conditions;
\begin{enumerate}[leftmargin=*,topsep=0pt,itemsep=-1ex,partopsep=1ex,parsep=1ex]
\item Edges are odd with respect to $\lvert\cdot\rvert$, i.e., there can be an edge between $v,w\in
\pzV$ only if $|v|\neq |w|$. In particular, $\pzA$ is bipartite. We refer to the vertices in
$\pzV_0\defeq\{v \in \pzV \mid |v|=0\}$ as \emph{bosonic}, and the vertices in $\pzV_1\defeq \{v\in\pzV\mid
|v|=1\}$ as \emph{fermionic}.
\item Edges are unimodular with respect to $\pzh$, i.e., there can be an edge $e\in\pzE$ between
$v,w\in \pzV$ only if $|\pzh(w)-\pzh(v)| = 1$. In this case, assuming $\pzh(w)= \pzh(v)+1$, we
call $v=\pzt(e)$ the target, and $w=\pzs(e)$ the source of $e$.
\item  $\pzA$ is ``$N$-color-regular'', i.e., every vertex has exactly one incident edge of each
color.
\item For every pair $i,j\in\{1,\ldots,N\}$ of distinct colors, the set of edges colored $i$ and $j$
form a disjoint union of $4$-cycles, which we denote by $\coprod C^{(2\pzc)}_4$, that are odd with
respect to the dashing, i.e., for every such four-cycle $C^{(2\pzc)}$, we have
\begin{equation}
\sum_{e \in C_4^{(2\pzc)}} \pzp (e) = 1 \in \ZZ/2
\end{equation}
\end{enumerate}
\end{definition}

We can associate a field multiplet $\mathcal{V} (\mathpzc{A})$ to an Adinkra $\pzA$ in the sense of
definition \ref{inthesense} as follows.

\begin{definition}[Field multiplet $\mathcal{V} (\pzA)$ associated to an Adinkra $\pzA$]
\label{FieldMultiplet}
Let $\pzA$ be an Adinkra. We define the field multiplet $\mathcal{V} (\mathpzc{A})$ associated to $\pzA$ as $\mathcal{V} (\mathpzc{A}) \defeq V(\mathpzc{A}) \otimes \kk[H]$, where the vector space
$V (\mathpzc{A})$ has a basis given by the vertices $v \in \mathpzc{V}$ of $\pzA$. The $\susy$-supermodule structure on $\mathpzc{V} (\pzA)$ is induced by linearly extending the action of the $\{ Q_i\}_{i= 1,\ldots, N}$ on the basis elements of $V (\pzA)$, which (upon identifying a vertex and its corresponding basis element) is defined as follows.

For an edge $e \in \mathpzc{E}$ between a source $w \in \mathpzc{V} $ and target $v \in \mathpzc{V} $ with color $\pzc(e)=i,$ the action of $Q_i$ on the corresponding basis elements is defined as
\begin{equation}
Q_i \cdot \left(v \otimes 1\right) = (-1)^{\pzp(e)} w \otimes 1,  \qquad Q_i \cdot \left(w \otimes 1 \right)= (-1)^{\pzp(e)} v \otimes H.
\end{equation}
\end{definition}
It is immediate to check that that the action of the $Q_i$ above satisfies $\{Q_i, Q_j \} = 2 \delta_{ij} H.$ This follows from conditions 3 and 4 in the definition of an Adinkra for distinct and equal $i$ and $j$, respectively.

\begin{remark}[Graded Components of $\mathpzc{V}(\pzA)$]
The graded components of the $\susy$-supermodule
\be
\mathcal{V} (\mathpzc{A}) = \bigoplus_{i \geq 0} P^{-i} \otimes \kk[H],
\ee
are $\kk$-vector spaces, $P^{-i}$,  with dimension $\dim_\kk P^{-i} = | \lbrace v \in \mathpzc{V}(\pzA): \mathpzc{h}(v) = i \rbrace |$.  Similarly, the super vector space
\be
V (\mathpzc{A}) = \bigoplus_{i \geq 0} P^{-i}.
\ee
Finally, we remark that, in field theory, we would replace $\kk[H]$ by the ring of continuous functions $\mathcal{C}^\infty (\mathbb{R}).$
The module structure is defined analogously.
\end{remark}
Our goal is to assess what type of off-shell
representations one obtains this way and explain the relationship with the pure spinor superfield formalism.  

%% file: S3adinkras.tex
\section{Pure Spinor Superfields and Homological Algebra of Adinkras}
\label{Adinkras}

The pure spinor superfield formalism is another method for constructing field representations of
supersymmetry algebras in general dimensions.  In $d$ dimensions, it constructs field
representations of the super Poincar\'e algebra $\mathfrak{p}^{d,\mathcal{N}}$ starting from
geometric data defined on an algebraic variety --- the nilpotence variety of the supersymmetry algebra.  For more details, see appendix \ref{Appendix1} (in particular, \ref{PSapp}) or the recent literature \cite{Eager22, Eager18}.
We will now focus on the $d=1$ case and introduce the related nilpotence varieties. 
We let
\begin{equation} \label{SP1}
R \defeq \mbox{Sym}^\bullet_{\kk} ( \ft_1^{\,\vee} ) \cong \kk[\lambda_1 \ldots, \lambda_{N}]
\end{equation}
be the ring of polynomial functions on the odd part of the super-translation algebra $\ft \subset
\mathfrak{susy}$, understood as a graded ring with $\deg (\lambda_i) = 1$.  In equation \eqref{SP1}
we have denoted with $\{ \lambda_i \}_{i = 1, \ldots, N}$ the {linear duals} of the $\{ Q_i
\}_{i = 1, \ldots, N}$. For varying $Q\in\ft_1$, the nilpotency equation, $\{ Q, Q\} = 0,$ defines a
quadratic ideal $I$ in the polynomial ring $R$, and the quotient ring determines the ring of
functions of an algebraic variety, the \emph{nilpotence variety} of the $d=1$ supersymmetry
algebra. More concretely, expanding $Q = \sum_i \lambda_i Q_i$ and using the commutation relations in \eqref{commH}
it is immediate to see that the ideal $I$ is generated by the standard quadratic form
\begin{equation}
q_N \defeq \sum_{i =1}^N \lambda_i^2,
\end{equation}
so that $I = \langle q_N \rangle$. Posing $\mathbb{A}^N_\kk \defeq \Spec \kk [\lambda_1, \ldots,
\lambda_N]$, we summarize this in the following definition, which is the pivotal construction
underlying the pure spinor superfield formalism, to be briefly discussed in the next section.

\begin{definition}[Nilpotence Variety $Y_N$ of $\susy$] The nilpotence variety of the $N$-extended
supersymmetry algebra is the affine scheme $(Y_N, \mathcal{O}_{Y_N}) \subset \mathbb{A}^N_\kk$ such
that $Y_N \defeq \mbox{Spec} (R / I)$ and $\mathcal{O}_{Y_N} \defeq R/I,$ with $I = \langle
q_N\rangle.$
\end{definition}
Notice that $I$ is homogeneous, therefore for any $N > 1$ the above nilpotence variety descends to a quadric hypersurface in $\mathbb{P}_\kk^{N-1} \defeq \mbox{Proj}\, \kk[\lambda_1, \ldots, \lambda_N]$, which we will denote again with $Y$ by abusing the notation. These varieties, together with their characteristic bundles -- the \emph{spinor bundles}, which will play a major role in this paper -- are described in more detail in appendix \ref{NVSPapp}.

In the next section we will recover the representations derived in the
previous section from the $N=2$ Adinkras, but now using the pure spinor superfield formalism instead.

\subsection{\texorpdfstring{$N=2$}{N=2} multiplets from pure spinor superfields}

Given an $R/I$-module $M$ for $R/ I \defeq \kk[\lambda_1, \lambda_2] / \langle q_2 \rangle,$ the ring of functions of the affine quadric in 2-dimensions, the pure spinor complex reads
\begin{equation}
( \calA^\bullet (M) , \DD ) \defeq \left ( M \otimes \mathcal{C}^\infty(\RR^{1|2}),\; \sum_{i = 1}^{2} \lambda_i \otimes \DD_i \right ),
\end{equation}
where $ \mathcal{C}^{\infty} (\RR^{1|2}) \defeq \mathcal{C}^\infty (\RR) \otimes \kk[\theta_1, \theta_2]$ is the ring of $\kk$-valued functions of the split supermanifold $\RR^{1|2}$, whose elements are the \emph{free superfields}. The $\lambda_i$ act via the $R/I$-module structure and $\DD_i \defeq \partial_{\theta_i} - \theta_i \partial_{t}$.
\begin{enumerate}[leftmargin=*,topsep=0pt,itemsep=-1ex,partopsep=1ex,parsep=1ex]
\item Let us first take the constant $R/I$-module $M = \kk$, seen as $(R/I)\, \mbox{mod}\, \langle
\lambda_1, \lambda_2 \rangle.$ Then the differential $\mathcal{D}$ acts as zero, and the cohomology
is the pure spinor complex itself, \emph{i.e.}
\begin{equation}
H^\bullet (\mathcal{A}^\bullet (\kk), 0)) = \mathcal{A}^\bullet (\kk) = \kk \otimes \mathcal{C}^\infty (\mathbb{R}^{1|2}) \cong \mathcal{C}^\infty (\mathbb{R}^{1|2}).
\end{equation}
This means that the free superfield representation on elements of $\mathcal{C}^\infty (\mathbb{R}^{1|2})$ of the form
\begin{equation}
X (t, \theta_1, \theta_2) = x (t) + \sum_{i = 1}^2  \theta_i\psi_i (t) + \theta_2\theta_1 F(t)
\end{equation}
corresponds to the constant module $\kk$ via pure spinor superfield formalism---this is true in
general for any dimension and amount of supersymmetry.
\item Let us now consider the following module\footnote{It is not hard to see that this descends to an $R/I$-module.}
\begin{equation} \label{resM}
  \begin{split}
M &=
\xymatrix{
\coker   \biggl( \kk[\lambda_1, \lambda_2] \oplus \kk[\lambda_1, \lambda_2] \biggr.
\ar[rrr]^{\qquad  \scalemath{0.7}{ \begin{pmatrix} \lambda_1 & -\lambda_2 \\  \lambda_2 & \lambda_1
\end{pmatrix}}} & & &  \biggl. \kk[\lambda_1, \lambda_2] \oplus  \kk[\lambda_1, \lambda_2]
\biggr )}
\\
&= \biggl\{ \biggl(\begin{matrix} e_1 \\ e_2 \end{matrix}\biggr) \in \kk[\lambda_1, \lambda_2]
 \oplus \kk[\lambda_1, \lambda_2] \biggr\} \biggl/
\biggl(\begin{matrix} \lambda_1 \\ \lambda_2 \end{matrix}\biggr) \, \kk[\lambda_1,\lambda_2]
+
\biggl(\begin{matrix} -\lambda_2 \\ \lambda_1 \end{matrix}\biggr) \, \kk[\lambda_1,\lambda_2].
 \biggr.
  \end{split}
\end{equation}
Then, if we parameterize a generic element of the pure spinor complex $\mathcal{A}^\bullet (M)$ as the linear combination
\begin{equation}
\Phi =  e_1  X \oplus e_2  Y \in \mathcal{A}^\bullet (M)
\end{equation}
of two free superfields $X, Y \in \mathcal{C}^\infty (\mathbb{R}^{1|2})$, we find
\begin{equation}
\mathcal{D}\Phi =
\biggl(\begin{matrix}
\lambda_1 e_1\DD_1 X + \lambda_2 e_1 \DD_2 X
\\
\lambda_1 e_2\DD_1 Y + \lambda_2 e_2 \DD_2 Y
\end{matrix}
\biggr)
\equiv
\biggl(\begin{matrix}
\lambda_1e_1 \DD_1 X - \lambda_1 e_2 \DD_2 Y
\\
\lambda_1 e_2 \DD_1 Y + \lambda_1 e_1 \DD_2 X
\end{matrix}
\biggr)
\end{equation}
For this to vanish, we need that $e_2=e_1$ up to a constant, and then $\DD_1 X = \DD_2 Y$ and
$\DD_1Y = - \DD_2 X$. All the while, the image of $\DD$ is generated by elements of the form
\begin{equation}
\biggl(\begin{matrix} \lambda_1 \DD_1 \Xi \\ \lambda_1 \DD_2\Xi \end{matrix}
\biggr)\equiv
\biggl(\begin{matrix} \lambda_2 \DD_2 \Xi \\ -\lambda_2 \DD_1 \Xi \end{matrix}
\biggr)
\end{equation}
for arbitrary superfield $\Xi$, and one can show that any $\Phi\in \langle
\lambda_1,\lambda_2\rangle \cap\ker\DD$ is of this form by a version of the Dolbeault lemma for
supermanifolds.\footnote{In complex coordinates introduced in \eqref{introduced}, this lemma says
that any chiral superfield $Z$ can be written as $\overline\DD\Xi$ for a certain superfield $\Xi$,
such as $\Xi =\bar\theta Z$.} Namely, the cohomology of the pure spinor complex is exactly
isomorphic to the $N=2$ chiral superfield of \eqref{CR} in (cohomological) degree $0$.
\begin{equation}
H^\bullet(\calA^\bullet(M),\DD) = \bigl\{ (X, Y ) \in \mathcal{C}^\infty (\mathbb{R}^{1|2}) \times \mathcal{C}^\infty (\mathbb{R}^{1|2}) : \mathcal{D}_1 X - \mathcal{D}_2 Y = 0 , \; \mathcal{D}_2 X + \mathcal{D}_1 Y  = 0\bigr\}.
\end{equation}
\end{enumerate}

\subsection{Complexes from Adinkras} \label{CplAdMain}
Given an Adinkra, we now show how to define, in a very natural fashion, a linear complex of free $R =
\kk[\lambda_1, \ldots, \lambda_N]$-modules. This complex descends to a complex of $R/\langle q_N \rangle$-modules
defined on the nilpotence variety $Y_N$. We will introduce a notion of Laplacian which behaves like the multiplication by the quadratic form
$q_N$ on the complex -- this will play a major role in what follows, and it will
allow to establish a connection between Adinkras, matrix factorizations and the geometry of quadrics hypersurfaces.\\
The following is the fundamental construction of this section.
\begin{construction} \label{constr} Let $\mathpzc{A}$ be an Adinkra. Letting $n_i \defeq \vert \{ v \in \mathpzc{V} (\mathpzc{A}) : \mathpzc{h} (v) = i \} \vert, $
we define

\begin{equation}
    \label{complexA}
C^\bullet (\mathpzc{A}) = \bigoplus_{i \geq 0} C^{-i} (\mathpzc{A}) \defeq
    \bigoplus_{i \geq 0} R^{ n_{i}},
\end{equation}
that is for each vertex of $\mathbb{Z}$-degree $\mathpzc{h} (v) = i$, we have a summand
$R$.

Abusing the notation and denoting with $v$ the basis element of the free $R$-module $C^{-i} (\mathpzc{A})$ corresponding to the vertex $v$, we define maps $d^{i}_{\pzA} : C^{-i} (\mathpzc{A}) \rightarrow C^{-i+1} (\mathpzc{A})$, by setting
\begin{equation} \label{diffA}
d^{i}_{\pzA} (v) \defeq \sum_{\mathpzc{S}(v)} \mathpzc{p}(e) \lambda_{\mathpzc{c} (e)} \mathpzc{t}(e),
\end{equation}
where we have posed $ \mathpzc{S} ({v}) \defeq \{ e \in E (\mathpzc{A}) : \mathpzc{s} (e) = v \}$.
\end{construction}
\noindent For the sake of notation, we will often leave the reference to the specific Adinkra $\pzA$ and the degree $-i$ in the differential understood, \emph{i.e.}\ $d \defeq d^{i}_{\mathpzc{A}}.$
It is easy to see that the pair $(C^\bullet(\mathpzc{A}), d )$ defines a linear complex of free $R$-module, so that the following definition if well-posed.
\begin{definition}[Complex of $\mathpzc{A}$] \label{defCpl} We call
    the linear complex of free $R$-modules $(C^\bullet
    (\mathpzc{A}), d)$ the complex associated to $\mathpzc{A}$.
\end{definition}
\begin{example}[The Koszul Adinkra $\pzA_K$] \label{KAd} As a concrete example, consider the $N=2$ free superfield Adinkra as introduced in the previous sections. For $R = \kk[\lambda_1, \lambda_2]$, it gives rise to a complex of the form
\be
\xymatrix{
( C^\bullet (\mathpzc{A}), d ) \equiv \big ( 0 \ar[r] & R \ar[r]^{d_{2} \quad } & R\oplus R \ar[r]^{\; \quad d_{1}} & R  \ar[r]  & 0 \big ),
}
\ee

where $d_1 = (\lambda_1, \lambda_2)^t$ and $d_2 = (- \lambda_2, \lambda_1)$. The complex $C^\bullet (\mathpzc{A})$ is the Koszul complex, and it resolves the $R/\langle q_2 \rangle$-module $\kk$ in free $R$-modules. As such, the input module of the pure spinor formalism $\kk$ is quasi-isomorphic to the (Koszul) complex $C^\bullet (\mathpzc{A})$. It is easy to see that the same construction is true for any $N$. Accordingly we will call the $N$-Adinkra $\mathpzc{A}$ whose associated complex is the Koszul complex resolving $\kk$ as $\kk[\lambda_1, \ldots , \lambda_N] / \langle q_N \rangle$-module the Koszul Adinkra $\pzA_K$. As we have seen, the field multiplet associated with the Koszul Adinkra is described by the free superfield, which is a section of $ \mathcal{C}^\infty (\mathbb{R}^{1|N})$.
\end{example}

We remark that, as defined above, the complex $C^\bullet (\pzA)$ only belongs to $\mbox{\sffamily{D}}^\flat (R\mbox{-\sffamily{Mod}})$. On the other hand, it is possible to prove that this complex belongs to the full subcategory $\mbox{\sffamily{D}}^\flat (R/ \langle q_N\rangle \mbox{-\sffamily{Mod}})$ of $\mbox{\sffamily{D}}^\flat (R\mbox{-\sffamily{Mod}})$, thus showing that the complex can in fact be viewed as a complex over the nilpotence variety $Y_N$.
More precisely, the restriction of scalars defines a fully faithful embedding
$
\mbox{\sffamily{D}}^\flat(R/I\mbox{-\sffamily{Mod}}) \hookrightarrow \mbox{\sffamily{D}}^\flat(R\mbox{-\sffamily{Mod}}),
$
and the adjunction between restriction of scalars and extension of scalars restricts to an equivalence
of categories on the essential image of the embedding. The interested reader might find the relevant constructions and the proof
of this result in the appendix \ref{CplxAd}: for the sake of exposition, here we content ourselves to state the following.
\begin{theorem}[Embedding]  \label{embeddingThm} Let $\mathpzc{A}$ be an Adinkra. The associated complex $C^{\bullet}(\pzA)$ to $\mathpzc{A}$ is in the
essential image of the embedding $\mbox{\sffamily{\emph{D}}}^\flat(R/I\mbox{-\sffamily{\emph{Mod}}})
\hookrightarrow \mbox{\sffamily{\emph{D}}}^\flat(R\mbox{-\sffamily{\emph{Mod}}})$.
\end{theorem}
A crucial ingredient in the proof of the previous theorem \ref{embeddingThm} that will also play a major role in the following of this paper is the fact the differential \eqref{diffA} of the complex $C^\bullet (\mathpzc{A})$ admits an \emph{adjoint} $d^\dagger : C^{\bullet} (\mathpzc{A}) \rightarrow C^{\bullet -1} (\mathpzc{A})$ mapping in the opposite direction of $d$. On a basis element $v$ given as above, it is defined by
\begin{equation} \label{adj}
d^{\dagger} (v) \defeq \sum_{\mathpzc{T}(v)} \mathpzc{p}(e) \lambda_{\mathpzc{c} (e)} \mathpzc{s} (e),
\end{equation}
where we have posed $\mathpzc{T} (v) \defeq \{e \in E (\mathpzc{A}) : \mathpzc{t} (e) = v \}.$ Reasoning as above, it is not hard to see that also $d^\dagger$ is nilpotent, and one can introduce \emph{Laplacian} $\Delta: C^{-i} (\mathpzc{A}) \rightarrow C^{-i} (\mathpzc{A})$, in the usual fashion via
\begin{equation} \label{Lapl}
\xymatrix@R=1.5pt{
\Delta : C^{-i}(\mathpzc{A}) \ar[r] & C^{-i} (\mathpzc{A}) \\
v \ar@{|->}[r] & \Delta v \defeq (d \circ d^\dagger + d^\dagger \circ d) v.
}
\end{equation}
The action of the Laplacian is characterized by the following lemma, whose proof is given in appendix \ref{CplxAd}.
\begin{lemma} \label{LaplacianLemma} Let $(C^\bullet (\mathpzc{A}), d)$ be the complex associated to
$\pzA$. Let $q_N \defeq \sum_{i = 1}^N \lambda^2_i$ be the standard quadratic form on $R$. Then the
Laplacian acts via multiplication by $q_N $ in $C^\bullet (\mathpzc{A})$, \emph{i.e.}\
\begin{equation}
    \Delta  = q_N \cdot \mbox{id}_{C^\bullet (\mathpzc{A})}.
\end{equation}
In other words, we have $(d + d^\dagger)^2 = q_N \cdot \mbox{id}_{C^\bullet (\pzA)}.$
\end{lemma}
Getting back to the above $N=2$ Koszul Adinkra above one sees that
\be
\xymatrix{
( C^\bullet (\mathpzc{A}), d , d^\dagger ) \equiv \Bigg ( 0 \ar[r] & R \ar@/^1pc/[r]^{d_{2} \quad } & \ar@/^1pc/[l]^{d_{2}^\dagger \quad }  R\oplus R \ar@/^1pc/[r]^{\; \quad d_{1}} &  \ar@/^1pc/[l]^{d_1^\dagger} R  \ar[r]  & 0 \Bigg ),
}
\ee
where $d^\dagger_i$ is simply the transpose of $d_i$ for $i=1,2$. \\
We will see in the next sections how this easy result will allow us to place Adinkras in the beautiful (and much broader) mathematical framework of the geometry of quadric hypersurfaces, which will be reviewed in the next section -- with a view toward Adinkras.

\subsubsection{Chevalley--Eilenberg Complex, Adinkras and Pure Spinors}

In this subsection, we make contact between the construction of the complex
$C^\bullet (\mathpzc{A})$ attached to an Adinkra $\mathpzc{A}$ as given above,
and the Chevalley--Eilenberg complex valued in the field multiplet $\mathcal{V}
(\mathpzc{A})$. This will allow us to place our construction in the context of
pure spinor superfield formalism, as recently proposed in \cite{Eager22,
ElliottDerived}. Namely, we will show that applying the pure spinor functor
$\mathcal{A}^\bullet$ to $C^\bullet(\pzA)$ yields precisely the field multiplet
$\mathcal{V} (\pzA)$ associated to the Adinkra $\pzA$.

The Chevalley--Eilenberg complex of the $N$-extended $d=1$ super translation algebra $\mathfrak{t}$, that was introduced in section \ref{representations}, is defined as
\be
\mbox{CE}^\bullet (\mathfrak{t}) \defeq \mbox{Sym}^\bullet (\mathfrak{t}^\vee[1], \mbox{d}_{\mbox{\scriptsize{CE}}}),
\ee
where the Chevalley--Eilenberg differential $\mbox{d}_{\mbox{\scriptsize{CE}}} $ is induced by the dual of the Lie bracket on $\mathfrak{t}$. Totalizing the degrees, so that $\mathfrak{t}^\vee_1$ and $\mathfrak{t}^\vee_2$ sit in degrees $0$ and $-1$ respectively, one has
\be
\mbox{CE}^{-p} (\mathfrak{t}) = \wedge^{p} \mathfrak{t}^\vee_2 \otimes_\kk R,
\ee
for $R = \mbox{Sym}^\bullet (\mathfrak{t}_1^\vee)$. Concretely, by choosing coordinates as above, so that $R = \kk[\lambda_1, \ldots, \lambda_N]$ and $\mathfrak{t}^\vee_2$ is generated by a single (odd) basis vector $v$, one can write
\be
\xymatrix{
\mbox{CE}^{-p} (\mathfrak{t}) \equiv \big ( 0 \ar[r] & R[v] \ar[r] & R \ar[r] & 0 \big ).
}
\ee
Notice that in terms of the related super Lie group $\mathpzc{T}$ obtained by  ``exponentiating'' $\mathfrak{t}$, one has that $\lambda_\alpha = d\theta_i$ and $v$ is the \emph{einbein} $v = dt + \sum_i \theta_i d\theta_i$, if $(t, \theta_i)$ is a system of coordinates for $\mathpzc{T} = \exp (\mathfrak{t}).$ Using these coordinates, one has that
\be
\mbox{CE}^\bullet (\mathfrak{t}) = (\kk [\lambda_1, \ldots, \lambda_N | v] , \; \mbox{d}_{\mbox{\scriptsize{CE}}} = q_N \frac{\partial}{\partial v}),
\ee
where $q_N = \sum_i \lambda_i^2$ is the quadratic form -- this immediately shows that $H^0 (\mbox{CE}^\bullet (\mathfrak{t})) \cong R/\langle q_N \rangle$.\\
Given a $\mathfrak{t}$-module or a $U_\kk(\mathfrak{t})$-module $\calV$, the
$\calV$-valued Chevalley--Eilenberg complex reads
\be
\mbox{CE}^\bullet (\mathfrak{t}, \mathcal{V} ) \defeq \big ( \mbox{CE}^\bullet
(\mathfrak{t}) \otimes_\kk \calV, \; \mbox{d}_{\mbox{\scriptsize{CE}}}^{\calV} \big ),
\ee
where the differential $\mbox{d}_{\mbox{\scriptsize{CE}}}^{\calV}$ gets a correction
with respect to $\mbox{d}_{\mbox{\scriptsize{CE}}}$ coming from the
$\mathfrak{t}$-action on $\calV$.
We are interested in the case where $\calV$ is a field multiplet $\mathcal{V} (\mathpzc{A})$ associated to an Adinkra $\mathpzc{A}$.  Recall from definition \ref{FieldMultiplet} that
\be
\mathcal{V} (\mathpzc{A}) = \bigoplus_{i \geq 0} P^{-i} \otimes_{\RR} \kk[H],
\ee
for some $\kk$-vector spaces $P^{-i}$ such that $\dim_\kk P^{-i} = | \lbrace v \in \mathpzc{V}(\pzA): \mathpzc{h}(v) = i \rbrace |$.
In the above coordinates, one has
\be
\label{CEdifferential}
\mbox{CE}^\bullet (\mathfrak{t}, \mathcal{V} (\pzA)) \defeq \left ( \mbox{CE}^\bullet (\mathfrak{t}) \otimes_\kk \mathcal{V} (\mathpzc{A}) , \; \mbox{d}_{\mbox{\scriptsize{CE}}}^{\mathcal{V} (\mathpzc{A})} = q_N \partial_v \otimes 1 + \Big( \sum_{i=1}^N \lambda_i \otimes \sigma (Q_i) \Big ) + v \otimes \sigma (H) \right ),
\ee
where $\sigma(Q_i)$ and $\sigma(H)$ are the representations of the generator of $\mathfrak{t}$ on $\mathcal{V} (\pzA)$ as in \ref{FmSf} above. Using ``coordinates'' $(\theta_i, t)$, these reads $\mathcal{Q}_i \defeq \sigma(Q_i) = \partial_{\theta_i} + \theta_i \partial_t$ and $\mathcal{H} \defeq \sigma (H) = \partial_t.$  Following \cite{ElliottDerived}, taking $\mathfrak{t}_2$ invariants of the above complex -- i.e.\ taking cohomology with respect to the last piece of the differential -- yields a quasi-isomorphism of complexes
\be \label{qisoCE}
\mbox{CE}^\bullet (\mathfrak{t}, \mathcal{V} (\pzA)) \cong (\kk[\lambda_1, \ldots, \lambda_N] \otimes_\kk \mathcal{V}_0(\pzA), \; \sum_{i=1}^N \lambda_i \otimes \sigma (Q_i) |_{t=\mbox{\scriptsize{const}}} ),
\ee
where $\mathcal{V}_0 (\pzA)$ is the ``fiber over $t=0$'' of $\mathcal{V} (\pzA)$, i.e.\  $
\mathcal{V}_0 (\mathpzc{A}) = \bigoplus_{i \geq 0} P^{-i},
$
and $\sigma (Q_i) |_{t=\mbox{\scriptsize{const}}}$ is the related restriction of the differential\footnote{Note that taking $\mathfrak{t}_2$-invariants eliminates $v$, hence the action of the first piece of the differential trivializes.}, whose action on tensors in $\mathcal{V}(\mathpzc{A})$ reads
\be
\sigma (Q_i) |_{t=\mbox{\scriptsize{const}}} (v \otimes f) = \mbox{ev}_{t=0} (Q_i (v \otimes 1)).
\ee
Looking at \eqref{qisoCE}, one immediately sees that $\kk[\lambda_1, \ldots, \lambda_N] \otimes \mathcal{V}_{0} (\pzA)$ is isomorphic as a graded vector space to the complex $C^{\bullet} (\pzA)$ attached to the
Adinkra $\pzA$ introduced above in \ref{constr}.
Moreover, let $v \in \mathpzc{V}(\pzA)$ be any vertex: by abuse
    of notation, we denote by the same letter the associated basis vector to $v$ in the field multiplet $\mathcal{V}(\pzA)$ with
    the same letter. Then, by construction, for any edge $e \in \mathpzc{E}$
    ending on $v$,
    $(Q_{\mathpzc{c}(e)})|_{t=\mbox{\scriptsize{const}}} (v) = \pm v'$ if $v =
    \mathpzc{t}(v)$, where $v'$ is associated to some vertex with $h(v') = h(v)
    - 1$. More precisely, $v' = p(e) \mathpzc{s}(e)$.
    However, if $v = \mathpzc{S}(e)$, then we have $\mbox{ev}_{t = 0}
    (Q_{\mathpzc{c}(e)})(v \otimes 1) = 0$. Thus, we can write the differential
    as
    \begin{equation}
        d(v) = \sum\limits_{e \in \mathpzc{T}(v)} p(e) \lm_{\mathpzc{c}(e)}
        \mathpzc{s}(e),
    \end{equation}
    which is precisely the differential appearing in \eqref{diffA}, which implies the following quasi-isomorphism (of complexes of $R$-modules)
    \be \label{qisoCEC}
    (C^\bullet (\pzA) , d) \cong \mbox{CE}^\bullet (\mathfrak{t}, \mathcal{V} (\pzA)).
    \ee
Finally, as a consequence of theorem 4.3 in \cite{ElliottDerived} (see also \ref{derivedPS} in the appendix of this paper) and the previous quasi-isomorphism in equation \eqref{qisoCEC}, one finds that
\be
\mathcal{A}^\bullet \circ C^\bullet (\pzA) \cong \mathcal{V}(\pzA),
\ee
for every $\mathfrak{t}$-modules $\mathcal{V}(\pzA)$, where $\mathcal{A}^\bullet$ is the pure spinor functor as introduced in \cite{ElliottDerived}. We summarize the previous discussion in the proposition below.
\begin{proposition} Let $\pzA$ be an Adinkra. The complex $C^\bullet (\pzA)$ associated to $\pzA$ is isomorphic to the Chevalley--Eilenberg complex $\mbox{\emph{CE}}^\bullet (\mathfrak{t}, \mathcal{V} (\pzA))$ in $\mbox{\sffamily{\emph{D}}}^{\flat}(R\mbox{-\sffamily{\emph{Mod}}})$.
In particular, $\mathcal{A}^{\bullet} \circ C^{\bullet}(\pzA) \cong \mathcal{V}(\pzA)$ as $\mathfrak{t}$-modules.
\end{proposition}
With reference to the above discussion, going back to the example of the Koszul $N$-Adinkra, one sees that that equation \ref{qisoCE} yields exactly the Koszul complex, as indeed $\kk[\lambda_1, \ldots, \lambda_N ] \otimes_\kk \mathcal{V}_0 (\mathpzc{A}) = \kk[\lambda_1, \ldots, \lambda_N ] \otimes_\kk \wedge^\bullet \kk^N $, and the differential reads $d = \lambda_i \otimes \partial_{\theta_i}$ if $\kk^N = \langle \theta_1, \ldots, \theta_N \rangle$. This is a minimal free resolution of the module $\kk$, and, in particular, it matches the $N=2$ Koszul Adinkra example discussed above.

%% file: S4mcm.tex
\section{Geometry of Quadric Hypersurfaces: Koszul Duality and Adinkras}
\label{mcm}

In this section, we will see how to place Adinkras and supersymmetric quantum mechanics into the much broader conceptual framework of the geometry of quadratic complete intersections.   The nilpotence variety of $N$-extended 1-dimensional supersymmetry algebra is cut out by a quadratic polynomial, giving our first connection to quadrics.  The quadratic polynomial
also appears in lemma \ref{LaplacianLemma} as the Laplacian acting on the complex associated to an Adinkra.  These two ingredients will be essential to our construction of matrix factorizations of quadrics from Adinkras.

Without any claim of originality or completeness, we now give a brief account of the algebro-geometric results that will be relevant to our setting. We will begin by introducing matrix factorizations, which naturally appear in our context due to lemma \ref{LaplacianLemma}, and see their relationship with a special class of modules on the related quadric hypersurface --- the maximal Cohen--Macaulay modules. In the second subsection, we will briefly explain how this correspondence can be lifted to a derived equivalence that can be interpreted as a ``deformed'' version of the {Bernstein--Gel'fand--Gel'fand correspondence}. Finally, in the last section, this equivalence of derived categories is recast into a concrete instance of Koszul duality, relating a geometric category -- (complexes of) modules on the nilpotence variety -- to a representation-theoretic one -- complexes of representations of the (universal enveloping algebras of the) supersymmetry algebra. This example illustrates how the pure spinor superfield formalism can be conceptually understood as a particular form of Koszul duality.

\subsection{Matrix factorization and maximal Cohen--Macaulay modules}  \label{mcmmf}

The notion of a matrix factorization is originally due to Eisenbud \cite{Eisenbud80}, who introduced
it as a tool to study resolutions of modules over hypersurfaces. In the following $R$ is a generic
Noetherian commutative ring with unity, for example -- with an eye towards our applications -- $R$ is a
polynomial ring over the complex or the real numbers.
 
\begin{definition}[Matrix Factorization of $x$] \label{MF} Let $x \in R$ and let $M_0$, $M_1$ be two $R$-modules. A matrix factorization of $x$ is an ordered pair of $R$-module homomorphisms $(\psi : M_0 \rightarrow M_1 , \; \varphi : M_1 \rightarrow M_0)$ such that $\psi \circ \varphi = x \cdot \mbox{id}_{M_1}$ and $\varphi \circ \psi = x\cdot \mbox{id}_{M_0}$.
\end{definition} 
The above definition can be conveniently recast into a supersymmetric form taking into account a $\mathbb{Z}/2$-grading.
Indeed, it is easily seen that the above data is the same as a $\mathbb{Z}/2$-graded module $M  =
M_0 \oplus M_1$ together with an odd endomorphism $f : M \rightarrow M$ such that $f \circ f = x
\cdot \mbox{id}_{M}.$ In the block matrix form this reads 
\begin{equation}
f = \left ( \begin{array}{c|c}
0 & \varphi \\
\hline
\psi & 0
\end{array}
\right ).
\end{equation}
It is immediate to observe from the definition that given a matrix factorization $(\psi , \varphi)$
of $x\in R$ then $x \cdot \coker (\varphi ) = 0$ as $\varphi \circ \psi = x \cdot \mbox{id}_{M_0}$.
It follows that $\coker (\varphi)$ is endowed with a structure of $R / \langle x \rangle$-module,
where $\langle x \rangle$ is the principal ideal $I$ generated by $x\in R$.

The concept of matrix factorization is useful to study modules over the quotient ring $R / I
\defeq R / \langle x \rangle $, referred to as the hypersurface ring. In particular, if $R$ is a
regular local ring or a graded ring, the minimal free resolution of every finitely generated
(possibly graded) $R/I$-module is eventually determined by a matrix factorization of the form
$(\psi, \varphi)$ as above. More precisely, in this situation, every resolution becomes periodic,
and, in turn, periodic resolutions correspond to \emph{maximal} Cohen--Macaulay modules\footnote{The
importance of Cohen--Macaulay modules in the context of the pure spinor superfield formalism has been
recently stressed in \cite{Eager22} -- namely, these rather regular modules can be used to describe
a neat geometric relation between a multiplet and its antifield multiplet. See also
\cite{6Dmultiplets}.}.
\begin{definition}[Maximal Cohen--Macaulay Module] Let $R$ be a local or graded ring of Krull dimension $\dim (R) = d$. We say that a finitely generated $R$-module $M$ is a maximal Cohen--Macaulay (MCM) module if $\mbox{depth} (M) = d$. 
\end{definition}
The above definition can be phrased by saying that the depth of a maximal Cohen--Macaulay module is the greatest possible. 
More precisely, the relation between MCM modules and matrix factorizations is made precise by the following result due to Eisenbud \cite{Eisenbud80}. 

\begin{theorem}[Resolutions of Modules over Hypersurfaces] \label{Eisenbud} Let $R$ be a regular
local or graded ring of Krull dimension $d$, let $x \in R$ and let $I \defeq \langle x \rangle$ be
the ideal generated by $x$. Let $M$ be a finitely-generated $R/I$-module whose minimal free
resolution $\mathbf{F} \stackrel{\varepsilon}{\longrightarrow} M \rightarrow 0$ is given by 
\begin{equation}
\mathbf{F} \defeq ( \ldots \rightarrow F_n \rightarrow \ldots  \rightarrow F_1 \rightarrow F_0).
\end{equation}
Then the following are true.
\begin{enumerate}[leftmargin=*]
\item $\mathbf{F}$ becomes periodic of period 2 after at most $d+1$ steps;
\item $\mathbf{F}$ is periodic of period 2 if and only if $M$ is a maximal Cohen--Macaulay module without free summands. In particular, every periodic free resolution is determined by a matrix factorization of $x$ over $R$, \emph{i.e.}\ by a pair of matrices $(\psi, \varphi)$ as in definition \ref{MF}.  
\end{enumerate}
\end{theorem}

Expanding on the second point of the previous theorem, and denoting with a bar the reduction modulo
$x$, one can see that the maps 
\begin{equation} \label{coker}
(\psi, \varphi) \longmapsto \left \{ \begin{array}{lc}
\mathbf{F}_{(\psi,\varphi)} \defeq ( \ldots \stackrel{\bar{\psi}}{\rightarrow} 
\bar{M}_1 \stackrel{\bar{\varphi}}{\rightarrow} \bar{M}_0 \stackrel{\bar{\psi}}{\rightarrow} 
\bar{M}_1 \stackrel{\bar{\varphi}}{\rightarrow} \bar{M}_0) \\
M_{(\psi, \varphi)} \defeq \coker (\varphi)
\end{array}
\right.
\end{equation}
define bijections between the set of (suitably defined equivalence classes of) matrix
factorizations, (isomorphism classes of) 2-step periodic minimal free resolution over $R/I$ and
maximal Cohen--Macaulay $R/I$-modules (modulo free summands) respectively, see \cite{Eisenbud80}. 

In \cite{Buchweitz2}, Buchweitz crucially enhanced the above result to an equivalence of
(triangulated) categories - see also Orlov (\cite{OrlovLG} and theorem 3.9 in \cite{Orlov}), which we now briefly explain. \\
In the following we let $\mbox{\sffamily{{MF}}} (R, x)$ be the category of matrix factorizations\footnote{The objects of $\mbox{\sffamily{{MF}}} (R, x)$ 
are matrix factorizations as in definition \ref{MF}. The morphisms 
in $\mbox{\sffamily{{MF}}} (R, x)$ are $R$-linear maps $\varphi : M \rightarrow
M^\prime$, endowed with a differential $d : \mbox{Hom}_R (M, M^\prime) \rightarrow \mbox{Hom}_R (M,
M^\prime) $ such that $d (\varphi) = \varphi \circ f - (-1)^{|\varphi|} f^\prime \circ \varphi $.
It is not hard to see that $(\mbox{Hom}_R (M, M^\prime), d)$ defines a complex.  } and we denote with
$[\mbox{\sffamily{{MF}}} (R, x)]$ its related homotopy category (obtained by taking
cohomology, analogously to the homotopy category of complexes of modules).
Further, we will denote with $\underline{\mbox{\sffamily{{MCM}}}} (R/I)$ the (stabilized) category of maximal 
Cohen--Macaulay $R/I$-modules\footnote{The objects of $\underline{\mbox{\sffamily{{MCM}}}} (R/I)$ are finitely-generated
maximal Cohen--Macaulay modules, morphisms are elements in $\mbox{Hom}_{R/I} (M, N)$ for $M, N$
two maximal Cohen--Macaulay modules, up to morphisms that factor through a projective or free
$R/I$-module.}. The second part of the above theorem suggests that these two categories are equivalent.\\
On the other hand, if we start from a finitely-generated $R/I$-module (\emph{i.e.}\ a module which
is not necessarily maximal Cohen--Macaulay without free summands), the first part of the above
theorem guarantees that one can still get a periodic resolution by modding out its initial
non-periodic part. This defines a bounded complex of finitely-generated projective modules: these
complexes define a subcategory of the derived categories $\mbox{\sffamily{{D}}}^\flat (R/I)$, deemed
as perfect complexes and denoted with $\mbox{\sffamily{{Perf}}}\, (R/I)$. The Verdier quotient
(triangulated category) $\underline{\mbox{\sffamily{{{D}}}}}^\flat (R/I) \defeq
\mbox{\sffamily{{{D}}}}^\flat (R/I) / \mbox{\sffamily{{{Perf}}}}\, (R/I)$ is known as the stabilized
derived category (or the singularity category\footnote{In the literature, this triangulated quotient
category if often also denoted by $\mbox{\sffamily{{{D}}}}^\flat_{\mathpzc{SG}} (R/I)$.}) of $R/I$.

There is a functor mapping $\underline{\mbox{\sffamily{{MCM}}}} (R/I) \rightarrow
\underline{\mbox{\sffamily{{{D}}}}}^\flat (R/I)$. The main result in \cite{Buchweitz2} is that this
functor defines an equivalence of categories. More precisely, one has the following.  

\begin{corollary}\label{eqc1} The following is an equivalence of triangulated categories
\begin{equation}
[ \mbox{\sffamily{{\emph{MF}}}} (R, x)] \cong \underline{\mbox{\sffamily{{\emph{D}}}}}^\flat (R/I) \cong \underline{\mbox{\sffamily{{\emph{MCM}}}}} (R/I).
\end{equation} 
\end{corollary}

We stress that since we are primarily interested in hypersurfaces in
$\mathbb{P}^N_{\kk}$, we will focus in particular on \emph{graded} matrix factorizations,
\emph{i.e.}\ if the element $x \in R$ is homogeneous of degree $d$, we require the maps and the free
modules to be graded, 
\begin{equation}
\xymatrix{
\ldots \ar[r] & \bigoplus_{i = 1}^n R (m_i) \ar[r]^{\varphi} & \bigoplus_{i = 1}^n R (n_i) \ar[r]^{\psi\quad } & \bigoplus_{i = 1}^n R (m_i + d) \ar[r]^{\varphi (d)} & \bigoplus_{i =1}^n R (n_i + d) \ar[r] & \ldots
}
\end{equation} 
Working over projective spaces, the free modules $R(\ell)$ get substituted by the sheaves
$\mathcal{O}_{\mathbb{P}_\kk} (\ell)$. We will denote the category of \emph{graded} matrix
factorization with the symbol $\mbox{\sffamily{{{MF}}}}_{\mbox{\scriptsize{\sffamily{gr}}}} (R, x)$,
and -- similarly as above -- with $[\mbox{\sffamily{{{MF}}}}_{\mbox{\scriptsize{\sffamily{gr}}}} (R,
x)]$ its derived category. The relevance of graded matrix factorizations for the (derived) geometry of quadric hypersurfaces 
is highlighted in \cite{HoriWalcher,WalcherStab}
and \cite{Orlov}. In turn, taking into account Kapranov's semi-orthogonal decomposition \cite{Kapranov88, Kapranov88b}, one has the following.  
\begin{theorem}[Quadrics \& Semi-Orthogonal Decompositions]
\label{KapranovQuadrics}
Let $X \subset \mathbb{P}_\kk^{N-1}$ be defined by a quadratic form $q$. Then the bounded derived category $\mbox{\sffamily{\emph{D}}}^\flat (X)$ of $X$ has the following semi-orthogonal decompositions
\begin{equation} \label{Orlov}
\mbox{\sffamily{\emph{D}}}^\flat (X) = \big \langle \mathcal{O}_X (-N-2), \ldots, \mathcal{O}_X,  [\mbox{\sffamily{{{\emph{MF}}}}}_{\mbox{\scriptsize{\sffamily{\emph{gr}}}}} (R, q)] \big \rangle = \left \{
\begin{array}{lll}
\langle \mathcal{O}_X (-N-2), \ldots, \mathcal{O}_X, \mathcal{S} \rangle &  N \mbox{ odd} \\
\langle \mathcal{O}_X (-N-2), \ldots, \mathcal{O}_X, \mathcal{S}_{+}, \mathcal{S}_- \rangle& N \mbox{ even},
\end{array}
\right.
\end{equation}
where $\mathcal{S}, \mathcal{S}_{\pm}$ are the spinor bundles on $X$. 
\end{theorem}
This shows that an even-dimensional smooth quadric has essentially only two matrix factorizations,
corresponding to its spinor bundles $\mathcal{S}_\pm$, and an odd-dimensional one has just one
matrix factorization, corresponding to its spinor bundle $\mathcal{S}.$ The relevance of the spinor
bundles for the study of representations comes from the fact, that they are the only elements in the
above decomposition with linear free resolution -- remember that the complex 
$C^\bullet (\mathpzc{A})$ associated to an Adinkra $\pzA$ 
is in fact linear\footnote{In physics, non-linearities introduce gauge degrees of freedom and the
resulting multiplet are no longer strict representations, but rather up to homotopy. This is a very
interesting topic in itself. We refer to \cite{Eager22} for more details on this.}.

Likewise, we are interested in \emph{graded} maximal Cohen--Macaulay modules -- again, we will denote
their category with $\mbox{\sffamily{MCM}}_{\mbox{\scriptsize{\sffamily{gr}}}} (R)$. Remarkably,
this category is equivalent to the category of \emph{arithmetically Cohen--Macaulay} sheaves
$\mbox{\sffamily{ACM}} (X)$ on (smooth enough) projective schemes (this is the case of the
nilpotence variety $Y_N \subset \mathbb{P}^{N-1}$ for $N\geq3$). More precisely, arithmetically
Cohen--Macaulay sheaves are locally Cohen--Macaulay coherent sheaves $\mathcal{E}$ with \emph{no}
intermediate cohomology. This means that if the projective scheme $X$ has dimension $n$, we require
that
\begin{equation}
H_\ast^i (X, \mathcal{E}) \defeq \bigoplus_{\ell \in \mathbb{Z}} H^i (X, \mathcal{E} (\ell)) = 0 
\end{equation}
for every $i = 1, \ldots, n-1.$ The functor $\Gamma_\ast \defeq H^0_\ast $ maps ACM sheaves to graded MCM
modules and defines an equivalence of categories (whose inverse functor is given by the
sheafification functor $M \mapsto \widetilde{M}$). 

\subsection{BGG correspondence for quadratic complete intersections}
 
Given that we are eventually interested in understanding supersymmetric multiplets, in the sense of 
definition \ref{inthesense}, we look to relate the equivalence of categories established in corollary \ref{eqc1} 
to a representation-theoretic context. For this, we will place the above equivalence in the 
much broader context of Koszul duality \cite{Beilinson96}. Notably, in our relatively simple setting -- that of a single quadric 
hypersurface ring, the coordinate ring of the nilpotence variety $Y_N$ -- Koszul duality can be phrased as a
generalized (or ``deformed'') version of the celebrated Bernstein--Gel'fand--Gel'fand (BGG)
correspondence \cite{BGG}\footnote{For a textbook account and a  geometry-to-representation theory ``dictionary''
in the original case see \cite{EisenbudSyzygies}.  For the deformed correspondence see also 
\cite{Sam16}.}.

 For the sake of generality, let $R/I = \oplus_j (R/I)_j$ be the homogeneous coordinate ring of a
complete intersection of quadric hypersurfaces, instead of just the coordinate ring of the nilpotence
variety $Y_N$ -- to ease the notation we will define $A \defeq R/I$ and $A_j \defeq (R/I)_j$. Notice
that the ground field $\kk = A_0$ has the structure of an $A$-module since $\kk\cong A / {\oplus_{j
>0}A_j}$ via augmentation map, namely ${\oplus_{j >0}A_j}$ is the maximal ideal of the graded ring
$A$.
\begin{definition}[Yoneda-Ext Algebra $A^!$] In the previous setting, we call the graded $A$-algebra
${A}^! \defeq \Ext^{\bullet}_A({k,k})$ the Yoneda algebra over $A$ of the ground field $\kk$, where
the grading coincides with the homological grading and the algebra structure is induced by the
Yoneda product of extension classes. 
\end{definition}
A remarkable fact that is specific to complete intersections of quadrics is the following, which is
adapted from the celebrated \cite{Priddy} (the interested reader can also look into \cite{Backelin}
and especially \cite{Beilinson96}).
\begin{theorem}[Complete Intersections of Quadrics \& Linear Resolution] Let $A$ be the homogeneous
coordinate ring of a complete intersection of quadrics and let $A^!$ its Yoneda-Ext algebra. Then 
\begin{enumerate}
\item the \emph{double} Yoneda algebra $\mbox{\emph{Ext}}^\bullet_{A^!} (\kk, \kk)$ is canonically isomorphic to $A$,
\item $(A, A^!)$ is a pair of \emph{Koszul algebras}, \emph{i.e.}\ the minimal graded free resolution of $\kk$ is linear both over $A$ and over $A^!$.   
\end{enumerate}
\end{theorem} 

Given a complete intersection of quadrics, the previous theorem makes it natural to think in terms
of \emph{Koszul pairs} $(A, A^!)$, and we call $A^!$ the \emph{Koszul dual algebra} of $A$. The
relevance of this notion resides in the fact that it is for such pairs that one can set up functors
as in the original Berstein-Gel'fand-Gel'fand correspondence \cite{BGG}. To this end, we let
$A$-$\mbox{\sffamily{Mod}}$ and $A^!$-$\mbox{\sffamily{Mod}}$ be respectively the category of (left)
$A$-modules and $A^!$-modules and $\mbox{\sffamily{D}}^\flat (A\mbox{-\sffamily{Mod}})$ and
$\mbox{\sffamily{D}}^\flat (A^!\mbox{-\sffamily{Mod}})$ their bounded derived categories. \\
The following theorem -- adapted from the appendix of \cite{Buchweitz06} -- shows that the usual
Bernstein--Gel'fand--Gel'fand correspondence for projective spaces generalizes to the more general
context of complete intersections of quadrics, relating two full triangulated subcategories of the bounded
derived categories of modules on $A$ and $A^!$ respectively\footnote{This is achieved via the 
same pair of natural functors $(\epsilon, \rho)$ of the original BGG correspondence, which are given by derived 
tensor products, $\epsilon: M_\bullet \mapsto M_\bullet \stackrel{\mathbb{L} \;\;\;}{\otimes_{A}} A^!$ and $\rho : \Gamma^\bullet \mapsto \Gamma^\bullet \stackrel{\, \mathbb{L} \quad}{\otimes_{A^!}} A$, for $M \in \mbox{\sffamily{D}}^{\flat}(A\mbox{-\sffamily{{Mod}}})$ and $\Gamma \in \mbox{\sffamily{D}}^{\flat}(A^!\mbox{-\sffamily{{Mod}}})$.}.  

\begin{theorem}[Generalized Bernstein--Gel'fand--Gel'fand Correspondence] \label{DBGG} Let $(A, A^!)$ be the Koszul pair associated with a complete intersection of quadrics. Then the following are true: 
\begin{enumerate}
\item The bounded derived categories $\mbox{\sffamily{\emph{D}}}^{\flat}(A\mbox{-\sffamily{{\emph{Mod}}}})$ and $\mbox{\sffamily{\emph{D}}}^{\flat}(A^!\mbox{-\sffamily{{\emph{Mod}}}})$ are equivalent.

\item Under this equivalence, the full triangulated subcategories of perfect objects $\mbox{\sffamily{\emph{Perf}}} (A) \subset \mbox{\sffamily{\emph{D}}}^{\flat} (A) $ and Artinian objects $\mbox{\sffamily{\emph{Art}}} (A^!) \subset \mbox{\sffamily{\emph{D}}}^{\flat} (A^!)$ are transformed into each other, \emph{i.e.}
\begin{eqnarray}
\xymatrix{
\mbox{\sffamily{\emph{Perf}}} (A)  \ar@/^1pc/[rr]|{\, \epsilon \,}
 &&  \ar@/^1pc/[ll]|{\, \rho \,}  \mbox{\sffamily{\emph{Art}}} (A^!).
}
\end{eqnarray}
In particular, the following is an equivalence of triangulated categories
\begin{equation}
\mbox{\sffamily{\emph{D}}}^{\flat}(A\mbox{-\sffamily{{\emph{Mod}}}}) / \mbox{\sffamily{\emph{Perf}}} (A) \cong \mbox{\sffamily{\emph{D}}}^{\flat}(A^!\mbox{-\sffamily{{\emph{Mod}}}}) / \mbox{\sffamily{\emph{Art}}} (A^!).
\end{equation}
\end{enumerate}
\end{theorem} 

\subsection{Quadric hypersurfaces and Koszul duality}

Theorem \ref{DBGG} takes a beautiful -- and very concrete -- form when specialized to the case we
are concerned with, where instead of a complete intersection of quadrics we consider a single
quadric hypersurface, the nilpotence variety $Y_N$ of the ($d=1$, $N$-extended) supersymmetry
algebra $\mathfrak{t}$ as in definition \ref{SUSYalg}. The reason for this is the fact that in this
setting the Koszul dual algebra of the algebra of functions on $Y_N$, is the universal enveloping algebra of
$\mathfrak{t}$ (seen as a $\mathbb{Z}$-graded algebra concentrated in degree 1 and 2, just as in
definition \ref{SUSYalg}) -- this is originally due to Milnor and Moore \cite{MilnorMoore} and
extensively discussed in \cite{Avramov}, where the algebra $\mathfrak{t}$ is called the
\emph{homotopy Lie algebra} of the ring of functions. 
\begin{theorem} Let $Y_N$ the nilpotence variety of the super-translation algebra
$\mathfrak{t}$. Then the Koszul dual of the algebra of functions of $Y_N$ is the universal
enveloping algebra $U_{\kk} (\mathfrak{t})$ of the super-translation algebra $\mathfrak{t}.$  
\end{theorem}
Thanks to this, theorem \ref{DBGG} takes the form of a Koszul duality-like result, relating a
geometric category-- the derived category of $R/I$-modules on $Y_N$ -- to a representation-theoretic
one -- the derived category of representations of universal enveloping algebra $U_\kk
(\mathfrak{t})$ of the supersymmetry algebra $\mathfrak{t}.$ \\
More can be said once the generators of $\mathfrak{t}$ are made explicit, as in definition \ref{SUSYalg}. Recalling that, as
$\mathbb{Z}$-graded Lie algebra concentrated in degrees 1 and 2, one has
given by 
\begin{equation}
\mathfrak{t} \defeq \mathfrak{t}_{1} \oplus \mathfrak{t}_2 = \underbrace{\langle Q_1, \ldots , Q_N \rangle}_{\mbox{\tiny{deg 1}}} \oplus \underbrace{\langle H \rangle}_{\mbox{\tiny{deg 2}}}, 
\end{equation}
where the $Q$'s are dual to the $\lambda$'s and $H$ is a central element dual to the quadratic form $q_N$, then the universal enveloping algebra can be given the following presentation 
\begin{align} \label{UEA}
U_{\kk} (\mathfrak{t}) & \cong T^\bullet_\kk \langle Q_1,\ldots, Q_N \rangle \otimes_{\kk} S^\bullet_\kk \langle H \rangle \, \mbox{{mod}}\, \langle Q \otimes Q - q_N (Q \odot Q) \, H \rangle \nonumber \\
& \cong  ( T^\bullet_\kk \langle Q_1,\ldots, Q_N \rangle  [H] ) \, \mbox{{mod}}\, \langle Q \otimes Q - q_N (Q\odot Q) \, H \rangle,
\end{align}
for $Q$ any supercharge in the span of the above $Q_i$ and $q_N (Q \odot Q) \in \kk.$ Since $\deg H = 2$, then $U_\kk (\mathfrak{t}^N)$ is endowed with the structure of a $\mathbb{Z}/2$-graded algebra: it follows that evaluation at $H = 1$ yields a surjective morphism of $\mathbb{Z}/2$-algebras, 
\begin{equation}
\mbox{ev}_{H  = 1} : U_\kk (\mathfrak{t})  \longrightarrow \cliff (q_N),
\end{equation}   
where $C\ell(q_N) $ is the Clifford algebra of the quadratic form $q_N.$ \noindent One can then see that every $\mathbb{Z}/2$-graded Clifford module $M = M_0 \oplus M_1$ can be pulled back to a $\mathbb{Z}/2$-graded $U (\mathfrak{t})$-module via the functor $ - \otimes_{\kk} \kk[H]: C\ell (q_N)\mbox{-\sffamily{Mod}} \rightarrow U_{\kk} (\mathfrak{t})\mbox{-\sffamily{Mod}}$, mapping for every $i\geq 0$
\begin{equation}
M_0 \oplus M_1 {\longmapsto} \left \{ \begin{array}{l}
\tilde M_0 \defeq M_0 \otimes_\kk \kk [H]  \\
\tilde M_1 \defeq M_1 \otimes_\kk \kk[H].
\end{array}
\right.
\end{equation} 
On the other hand, it is easy to see that not every $U_{\kk} (\mathfrak{t})$-module comes from a
$\mathbb{Z}/2$-graded $C\ell (q_N)$-module by the previous construction -- in fact these two
categories are far from being equivalent. It is then natural to ask, in light of the generalized
Bernstein--Gel'fand--Gel'fand correspondence of theorem \ref{DBGG}, which $R/I$-modules are mapped to Clifford
modules on the $U_\kk(\mathfrak{t})$ side of the correspondence: this is where one recovers the
relation with maximal Cohen--Macaulay modules on the nilpotence variety $Y_N$, introduced in the section \ref{mcmmf}.\\
 Indeed, for regular non zero quadratic forms, the category
$C\ell_{ 0} (q_N)\mbox{-\sffamily{{Mod}}}$ is equivalent to the category of
$\mathbb{Z}/2$-graded $C\ell (q_N)\mbox{-\sffamily{{Mod}}}$. In turn, the derived category
$\mbox{\sffamily{{D}}}^{\flat}(C\ell (q_N)\mbox{-\sffamily{{{Mod}}}})$ is equivalent to $C\ell
(q_N)\mbox{-\sffamily{{{Mod}}}}_{\mbox{\scriptsize{\sffamily{gr}}}}$, the category of
$\mathbb{Z}$-graded (finitely generated) $C\ell (q_N)$-modules. Then theorem \ref{DBGG} takes the
form of a Koszul duality-like result for quadric hypersurfaces.

\begin{theorem} \label{Koszul1} Let $(R/I, U(\mathfrak{t}))$ be the Koszul pair of the nilpotence variety $Y_N$. Then the following is an equivalence of categories
\begin{equation}
\mbox{\sffamily{\emph{D}}}^{\flat}(R/I\mbox{-\sffamily{{\emph{Mod}}}}) \cong \mbox{\sffamily{\emph{D}}}^{\flat}(U_\kk (\mathfrak{t})\mbox{-\sffamily{{\emph{Mod}}}}).
\end{equation}
In particular, under this equivalence, the following (Abelian) categories are mapped into each other:
\begin{equation}
\xymatrix{
 \underline{\mbox{\sffamily{{\emph{MCM}}}}}_{\mbox{\scriptsize{\sffamily{\emph{gr}}}}} (R/I) \ar@/^1pc/[rr]|{\, \epsilon \,}
 &&  \ar@/^1pc/[ll]|{\, \rho \,}  C\ell (q_N)\mbox{-\sffamily{\emph{Mod}}}_{\mbox{\scriptsize{\sffamily{\emph{gr}}}}}. 
}
\end{equation}
\end{theorem}

Notice that, implicitly, the above result says that perfect objects $\mbox{\sffamily{{Perf}}}
(R/I)\subset \mbox{\sffamily{{D}}}^{\flat}(R/I\mbox{-\sffamily{{{Mod}}}})$ are mapped to Artinian
objects $\mbox{\sffamily{{Art}}} (U_\kk (\mathfrak{t})) \subset \mbox{\sffamily{{D}}}^{\flat}(U_\kk
(\mathfrak{t})\mbox{-\sffamily{{{Mod}}}})$ and vice versa. Due to its relevance for the present
paper, the interested reader can find in appendix \ref{AppendixC} a more detailed -- and
down-to-earth -- discussion of the correspondence between Clifford modules and maximal Cohen--Macaulay
modules on quadrics via matrix factorization, following \cite{Buchweitz06}\footnote{A nice reference for
Clifford modules in relation to supersymmetry is \cite{Deligne}.}.

In the next sections, we will spell out the significance of these equivalences (and especially of
theorem \ref{Koszul1}) in the context of Adinkras---see in particular sections
\ref{sec:valise} and \ref{sec:flipping} and related sub-sections. As a concluding
remark to this section, we observe that the BGG correspondence in the form of theorem \ref{Koszul1}
offers (in some sense) a first concrete and geometric realization of the derived equivalence
proved in the recent \cite{ElliottDerived}, see also theorem \ref{derivedPS} in the appendix to this
paper. In this sense, the relatively easy setting of supersymmetric quantum mechanics allows for a
friendly formulation of the (derived) pure spinor superfield formalism in terms of standard
algebro-geometric notions, such as the universal enveloping algebra. Finally, once again, the
relevance of a regularity condition, such as Cohen--Macaulayness for modules, manifests itself in a
very natural fashion. 

\begin{remark}[On Koszul Duality and Supersymmetry]
In general, Koszul duality provides an equivalence between the derived categories of modules over
$\mbox{\sffamily{D}}^\flat (A\mbox{-\sffamily{Mod}})$ and $\mbox{\sffamily{D}}^\flat
(A^!\mbox{-\sffamily{Mod}})$.  One might be afraid that this simply exchanges one difficult problem
for another.  The universal enveloping algebra $U_{\mathbf{k}} (\mathfrak{t})$ of the
supertranslation algebra is an associative but non-commutative algebra, while its Koszul dual---the
ring of functions over the nilpotence variety is commutative.  This offers some hope.  Similar to
the Borel--Weil--Bott theorem, the powerful geometric technique of Kempf--Lascoux--Weyman produces
syzygies from desingularizations of vector bundles over projective spaces \cite{Weyman03}.
Ultimately, it appeals to Grothendieck duality \footnote{For a modern introduction, see
\cite{Neeman20}.}.  
 
A beautiful introduction to this set of ideas is in \cite{Sam16}, which goes on to produce
representations of the orthosymplectic group using modules over the corresponding nilpotence
variety in the complete intersection case.  This is the case for the $d=6$ $(\mathcal{N}, 0)$ supertranslation algebras for $\mathcal{N} \ge 3.$  A suitable generalization is expected to apply to the $d=6$ $(2, 0)$ case.  This can be seen as a far-reaching generalization
of the classical Buchsbaum--Rim complex used to produce the BV complex for the $d=6$ (1,0)
hypermultiplet in \cite{Eager22} and \cite{6Dmultiplets}.

\end{remark}

%% file: S5examples.tex
\section{Constructions and Examples}
\label{props}

In this section, we will describe various constructions of Adinkras using the tools developed in section \ref{Adinkras} and the results of section \ref{mcm}. \\
First, we define some basic invariants. As in
construction \ref{constr}, given an Adinkra $\pzA$, we recall that the associated complex reads $C^\bullet (\mathpzc{A}) = \bigoplus_{i \geq 0} C^{-i} (\mathpzc{A})$, where $C^{-i} (\mathpzc{A}) = R^{n_i}$
for $ n_i \defeq \vert \{ v \in \mathpzc{V} (\mathpzc{A}) : \mathpzc{h} (v) = i \} \vert$, i.e.\ $n_i$ is the rank of the
free module $C^{-i} (\pzA).$ On the other hand, the integers $n_i$ correspond to the number of
component fields of a multiplet of engineering dimension $i$, for this reason we introduce the following notation.
\begin{definition}[Rank Sequence \& Length] Let $\mathpzc{A}$ be an Adinkra together with its complex $C^\bullet (\mathpzc{A})$
and let $\mathcal{V}_\mathpzc{A}$ be the multiplet related to $\mathpzc{A}$.
We will call the sequence $(n_\ell, \ldots, n_i, \ldots, n_\ell+k) \in \mathbb{N}^{k}$, where $n_i = \mbox{rank}\, C^{-i} (\pzA)$,
the \emph{rank sequence} of the multiplet $\mathcal{V}_\mathpzc{A}$. Moreover, we call the number $k \in \mathbb{N}$ the
\emph{length} of the multiplet $\mathcal{V}_\pzA$.
\end{definition}
Notice that the length of $\mathcal{V}_\mathpzc{A}$ equals the difference between the maximum and minimum
of the degree function $\mathpzc{h} : \mathpzc{V} \rightarrow \mathbb{Z}$ of the Adinkra $\mathpzc{A}$. Also, in the following we will often identify a multiplet with its
rank sequence $\mathcal{V}_\pzA \equiv (n_\ell, \ldots, n_{\ell + k})$, even if this is somewhat improper, as we shall see shortly. \\
In section \ref{sec:valise} we study one of the most fundamental Adinkra -- the
valise Adinkra -- in light of the theory of matrix factorizations and maximal Cohen--Macaulay
modules developed in section \ref{mcm}. As we shall see, these Adinkra correspond to
characteristic bundles on the quadric defined by $q_N$, the spinor bundles, see lemma
\ref{irreducibleVal}. Further, in subsection \ref{BottP} we discuss Bott periodicity for (real) Clifford algebras in relation to Kn\"orrer
periodicity for matrix factorizations, and hence Adinkras in view of the results of section \ref{sec:valise}.

In section \ref{sec:flipping}, we study the operation of ``vertex raising'' on Adinkras. This
produces a new Adinkra with the same dashed chromotopology, but different heights -- we will
interpret this from the point of view of homological algebra, proving lemma
\ref{flipping_theorem}, and making contact with theorem \ref{Koszul1}.

In section \ref{sec:counterexample} we show that the rank sequence of a multiplet is not a complete
invariant, \emph{i.e.}\ there exist two distinct multiplets with the same rank sequence. On the
other hand, we show that the cohomology of the associated complexes is capable of distinguishing
between Adinkras with the same rank sequence, giving new invariants associated with multiplets.

Section \ref{1_7_7_1} offers a detailed study of the peculiar geometry arising from the $(1,7,7,1)$
multiplet in $N=7$.  We will show that the complex associated with this Adinkra admits an
embedding inside the Koszul complex: this is a general feature of a certain class of multiplets,
which admits embedding inside the free superfield, which we discuss in section \ref{sec:embeddings}.

Finally, in section \ref{ExtAdi} we provide a description of the non-Adinkras graphs considered in \cite{Doran13, Hubsch13} as suitable
extensions classes of Adinkras.

We remark that throughout this section of the paper we have tried to always make theory go together with concrete constructions
in the hope to make this paper more comprehensible and readable. Abstract results and theorems never come alone, but they are always accompanied
by explicit examples, illustrating their relevance.

\subsection{Valise Adinkras and maximal Cohen--Macaulay modules}
\label{sec:valise}

The simplest Adinkras $\mathpzc{A}$ one can build are concentrated on just two levels or degrees. Working up to an overall shift, this means that
the vertices can only have height 0 and 1. Following \cite{Doran07} we give the following definition.
\begin{definition}[Valise Adinkra] An Adinkra is called a valise if its vertices are concentrated on two degrees only, i.e.\
\begin{equation}
\max_{v,w \in \mathpzc{V}} | \mathpzc{h} (v) - \mathpzc{h} (w) | = 1.
\end{equation}
We will say that the complex $C^\bullet (\mathpzc{A})$ associated to $\pzA$ is a valise complex if $\pzA$ is a valise Adinkra.
 \end{definition}
In the following, we will assume without loss of generality that the valise Adinkra $\mathpzc{A}$ is concentrated in degree 0 and 1: this means that the associated complex reads
\be \xymatrix{
C^\bullet (\pzA) \equiv \big ( 0 \ar[r] & C^{-1} (\pzA)  \ar[r]^d & C^0 (\pzA) \ar[r] & 0 ).
}
\ee
Notice that it follows from the definition of Adinkra that the number of vertices of degree 0 is the same as the number of vertices of degree 1. This implies immediately that
$n_1 = \mbox{rank} \, C^{-1}(\mathpzc{A}) = \mbox{rank} \, C^0 (\mathpzc{A}) = n_0$.\\
The theory previously developed in section \ref{CplAdMain} and section \ref{mcm} allows us to give the following characterization.
\begin{lemma}[Valise Adinkras and Matrix Factorizations] Let $\pzA$ be a valise $N$-Adinkra with associated valise complex $C^\bullet (\pzA).$ The pair $(M, f)$ with $M \defeq C^0(\pzA) \oplus C^{-1}(\pzA)$ and
\be
\label{MFvalise} f
\defeq \left ( \begin{array}{c|c} 0 & d \\
    \hline
    d^\dagger & 0
    \end{array}
    \right ),
\ee
defines a linear matrix factorization of the quadratic form $q_N$. If $Y_N$ is the nilpotence variety defined by the ideal $I = \langle q_N \rangle $, we have the following correspondences
\be \nonumber
\left \{
\begin{minipage}{2cm}
\begin{center} \emph{Valise} \\
\emph{Adinkras}
\end{center}
\end{minipage}
\right \} \leftrightsquigarrow \left \{ \begin{minipage}{4.8cm}
\begin{center} \emph{Linear Graded \emph{MCM}} \\
\emph{Modules on }$Y_N$
\end{center}
\end{minipage}
\right \}
 \leftrightsquigarrow \left \{ \begin{minipage}{3cm}
\begin{center} $\ZZ/2$-\emph{graded} \\
 $C\ell(q_N)$-\emph{modules}
 \end{center}
\end{minipage}
 \right \}.
 \ee
\end{lemma}
\begin{proof} Thanks to lemma \ref{LaplacianLemma}, we have that $d\circ d^\dagger$ acts on $C^0(\mathpzc{A})$ as multiplication by $q_N$ and analogously
$d^\dagger \circ d$ acts on $C^{-1} (\pzA)$ as multiplication by $q_N$. This shows that the pair of $(C^0 (\pzA) \oplus C^{-1} (\pzA), f)$ as defined above
yields a linear matrix factorization of $q_N$. It follows from theorem \ref{Eisenbud} that the complex $C^\bullet (\mathpzc{A})$ is quasi-isomorphic to $\coker (d)$, a (linear
graded) maximal Cohen--Macaulay module. In turn, maximal Cohen--Macaulay modules correspond to $\mathbb{Z}_2$-graded $C\ell (q_N)$-modules by theorem
\ref{Cl_MCM} -- or as a particular case of theorem \ref{Koszul1}.
\end{proof}
The relation with $\mathbb{Z}/2$-graded $C\ell(q_N)$-modules allows us to attach to valise Adinkras a notion of irreducibility. In fact, recalling that the category of $\mathbb{Z}/2$-graded $C\ell (q_N)$-modules is equivalent to the category of $C\ell_0 (q_N)$-modules via the projection functor $M_0 \oplus M_1 \mapsto M_0$ (whose inverse functor maps $M_0 \longmapsto C\ell  (q) \otimes_{C\ell_{0} (q)} M_0 \cong M_0 \oplus ( C\ell_{{1}} (q) \otimes_{C\ell_{{0}} (q)} M_0)$), we can give the following definition.
\begin{definition}[Irreducible Valise] We say that a valise $N$-Adinkra is irreducible if it corresponds an irreducible $C\ell_0 (q_N)$-module.
\end{definition}
The number of bosons or fermions per degree as a function of $N$ is then given by the (real) dimension $d_\mathbb{R} $ of irreducible representations of the real Clifford algebra $C\ell_{0} (q_N)$, as in the following table ($\nu$ is the number of inequivalent irreducible representations).
\begin{table}[h!]
\label{CLmodule}
\centering
\begin{tabular}{cccc}
\toprule
$N$ \mbox{mod}\, 8 & $C\ell_0 (q_N)$ & $\nu $ &  $d_\mathbb{R}$ \\
\midrule
0  &  $\mathbb{R}(2^{\frac{N-2}{2}}) \oplus \mathbb{R}(2^{\frac{N-2}{2}}) $ & 2 & $2^{\frac{N-2}{2}}$ \\
1,7  &  $\mathbb{R}(2^{\frac{N-1}{2}})$ & 1 & $2^{\frac{N-1}{2}}$ \\
2,6 & $\mathbb{C} (2^{\frac{N-2}{2}})$ & 1 & $2^{\frac{N}{2}}$ \\
3,5 & $\mathbb{H}(2^{\frac{N-3}{2}})$ & 1 & $2^{\frac{N+1}{2}}$ \\
4 & $\mathbb{H}(2^{\frac{N-4}{2}}) \oplus \mathbb{H}(2^{\frac{N-4}{2}}) $ & 2 &$ 2^{\frac{N}{2}} $ \\
\bottomrule
\end{tabular}
\caption{Irreducible $C\ell_0$-modules and their dimensions}
\end{table}
It follows that an irreducible valise Adinkra corresponds to an essentially unique $C\ell_0 (q_N)$-module (or analogously, its related $\mathbb{Z} /2$-graded $C\ell (q_N)$-module) or, on the $R/I$-side of the correspondence, to an essentially unique maximal Cohen--Macaulay module.
\begin{lemma} \label{irreducibleVal} Let $\pzA$ be an irreducible valise $N$-Adinkra with associated valise complex $C^\bullet (\pzA).$ Then we have the following correspondence
\be \nonumber
\left \{
\begin{minipage}{3.2cm}
\begin{center} \emph{Irreducible Valise} \\
\emph{Adinkras}
\end{center}
\end{minipage}
\right \} \leftrightsquigarrow \left \{ \begin{minipage}{3cm}
\begin{center} \emph{Spinor Bundles} \\
\emph{on }$Y_N$
\end{center}
\end{minipage}
\right \}
 \leftrightsquigarrow \left \{ \begin{minipage}{4.2cm}
\begin{center} \emph{Irreducible }$\ZZ/2$-\emph{graded} \\
 $C\ell(q_N)$-\emph{modules}
 \end{center}
\end{minipage}
 \right \}.
 \ee
\end{lemma}
Moreover, every valise Adinkra corresponds to a direct sum of irreducible $\mathbb{Z}/2$-graded $C\ell (q_N)$-modules or, analogously, a direct sum of spinor bundles on the nilpotence variety $Y_N$.
\begin{proof} It follows from \eqref{Orlov} that the category of (graded) matrix factorizations on quadric hypersurfaces is generated by the spinor bundles, which is indeed a maximal Cohen--Macaulay module. The last statement follows from the fact that matrix algebras, such as Clifford algebras, are simple and, as such, they are completely reducible.
\end{proof}
Notice that given a valise $N$-Adinkra, the above result constrains the ranks of $C^0(\mathpzc{A})$ and $C^{-1} (\mathpzc{A})$ to be powers of 2, as they need to be multiple of some of the dimensions $d_\mathbb{R}$ appearing in table \ref{CLmodule}. \\
The above discussion gives an encompassing theoretical framework for valise Adinkras, via maximal Cohen--Macaulay
modules such as the spinor bundles, and Clifford modules.
As a first concrete example, consider the valise Adinkra for $N = 4$ in figure \ref{adk_n4_4_4}.
\begin{figure}[h!]
    \centering
    \includegraphics[height=5cm]{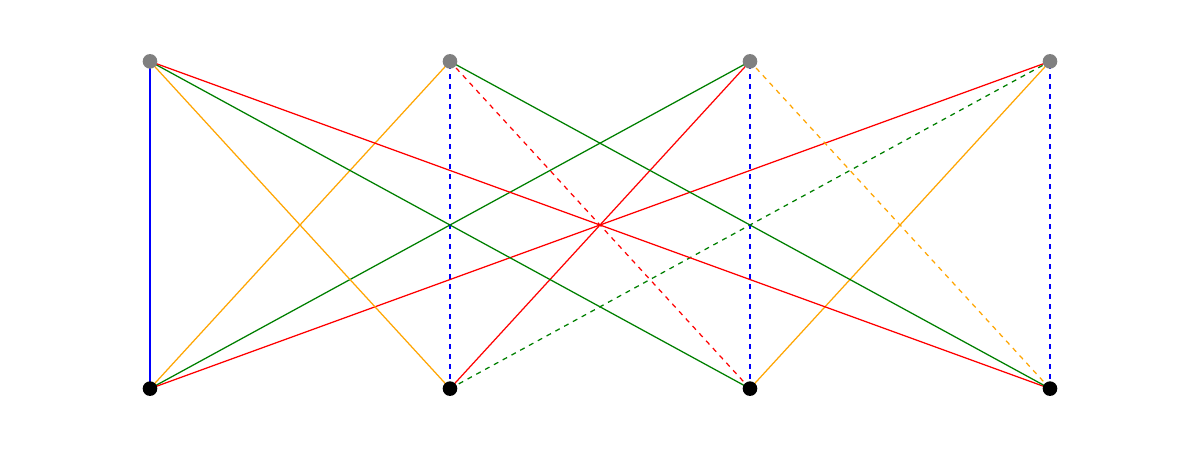}
    \caption{A Valise Adinkra for $N = 4$.}
    \label{adk_n4_4_4}
\end{figure}

The associated complex $C^\bullet (\mathpzc{A})$ is given by
\be \label{irrval}
\begin{tikzcd}
C^\bullet (\mathpzc{A}) \equiv \big (  0 \arrow[r] & C^{-1} (\pzA) = R^4 \arrow[r,"d"] & C^{0} (\pzA) = R^4 \arrow[r] & 0 \big ).
\end{tikzcd}
\ee
where
\[
d = \bmat
\textcolor{blue}{\lm_1} & \textcolor{orange}{\lm_2} & \textcolor{darkgreen}{\lm_3} & \textcolor{red}{\lm_4} \\
\textcolor{orange}{\lm_2} & \textcolor{blue}{-\lm_1} & \textcolor{red}{\lm_4} & \textcolor{darkgreen}{-\lm_3} \\
\textcolor{darkgreen}{\lm_3} & \textcolor{red}{-\lm_4} & \textcolor{blue}{-\lm_1} & \textcolor{orange}{\lm_2} \\
\textcolor{red}{\lm_4} & \textcolor{darkgreen}{\lm_3} & \textcolor{orange}{-\lm_2} & \textcolor{blue}{-\lm_1} \\
\emat
\] \label{matrixD}
\be
d^\intercal d = d d^\intercal = q_4 \cdot \operatorname{id}_4,
\ee
and where we have written $d^\intercal$ to emphasize that the adjoint differential $d^\dagger$ is
given by the transposed matrix represented in the basis defined by the Adinkra. \\
The sum of the dimensions of the free modules appearing in the complex $C^\bullet (\mathpzc{A})$ (i.e.\ the \emph{total} dimension of the complex)
is $4+4=8$ (this is twice the dimension $d_\mathbb{R}$
of the irreducible $C\ell_0 (q_4)$-module in table \ref{CLmodule}), which is the same as the (total) dimension of one of the two irreducible
$\ZZ / 2$-graded \emph{real} Clifford modules for $N = 4$, of (super)dimension $4|4$ --
and indeed it accounts for 4 bosons and 4 fermions. This complex is in fact the
minimal linear $R$-resolution of the corresponding maximal Cohen--Macaulay module $M = \coker (d)$, which in turn corresponds
to one of the spinor bundles on $Y_4$, i.e.\
\be
\begin{tikzcd}
0 \arrow[r] & R^4 \arrow[r, "d"] & R^4 \arrow[r] & M = \coker (d).
\end{tikzcd}
\ee

\subsubsection{Topologies on Valise Adinkras}
In general, we can have different topologies of graphs underlying valise Adinkras for fixed $N$. To
see an example, consider again the relevant case $N = 4$. It is easy to see that
the fully connected graph in figure
\ref{n4_koszul_folded} is an $N=4$ valise Adinkra.
\begin{figure}[ht!]
    \centering
    \includegraphics[height=5cm]{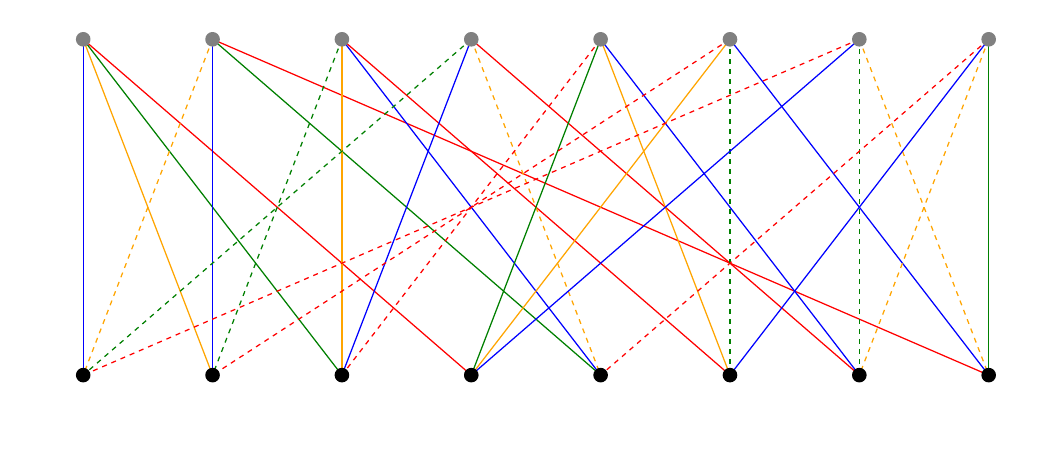}
    \caption{Valise Adinkra associated to the Koszul Adinkra $\pzA_K$ for $N = 4$.}
    \label{n4_koszul_folded}
\end{figure}
It gives rise to the associated complex
\begin{center}
    \begin{tikzcd}
   C^\bullet (\mathpzc{A}) \equiv \big (  0 \arrow[r] &  R^8 \arrow[r,"d"] & R^8 \arrow[r] & 0 \big ),
    \end{tikzcd}
\end{center}
with $d$ given by the matrix
\begin{equation}
    {d} =  \bmat
\textcolor{blue}{\lm_1} & \textcolor{orange}{\lm_2} & \textcolor{darkgreen}{\lm_3} & \textcolor{red}{\lm_4} & \textcolor{black}{0} & \textcolor{black}{0} & \textcolor{black}{0} & \textcolor{black}{0} & \\
\textcolor{orange}{-\lm_2} & \textcolor{blue}{\lm_1} & \textcolor{black}{0} & \textcolor{black}{0} & \textcolor{darkgreen}{\lm_3} & \textcolor{black}{0} & \textcolor{black}{0} & \textcolor{red}{\lm_4} & \\
\textcolor{black}{0} & \textcolor{darkgreen}{-\lm_3} & \textcolor{orange}{\lm_2} & \textcolor{black}{0} & \textcolor{blue}{\lm_1} & \textcolor{red}{\lm_4} & \textcolor{black}{0} & \textcolor{black}{0} & \\
\textcolor{darkgreen}{-\lm_3} & \textcolor{black}{0} & \textcolor{blue}{\lm_1} & \textcolor{black}{0} & \textcolor{orange}{-\lm_2} & \textcolor{black}{0} & \textcolor{red}{\lm_4} & \textcolor{black}{0} & \\
\textcolor{black}{0} & \textcolor{black}{0} & \textcolor{red}{-\lm_4} & \textcolor{darkgreen}{\lm_3} & \textcolor{black}{0} & \textcolor{orange}{\lm_2} & \textcolor{blue}{\lm_1} & \textcolor{black}{0} & \\
\textcolor{black}{0} & \textcolor{red}{-\lm_4} & \textcolor{black}{0} & \textcolor{orange}{\lm_2} & \textcolor{black}{0} & \textcolor{darkgreen}{-\lm_3} & \textcolor{black}{0} & \textcolor{blue}{\lm_1} & \\
\textcolor{red}{-\lm_4} & \textcolor{black}{0} & \textcolor{black}{0} & \textcolor{blue}{\lm_1} & \textcolor{black}{0} & \textcolor{black}{0} & \textcolor{darkgreen}{-\lm_3} & \textcolor{orange}{-\lm_2} & \\
\textcolor{black}{0} & \textcolor{black}{0} & \textcolor{black}{0} & \textcolor{black}{0} & \textcolor{red}{-\lm_4} & \textcolor{blue}{\lm_1} & \textcolor{orange}{-\lm_2} & \textcolor{darkgreen}{\lm_3} & \\
\emat .
\end{equation}
Computing $d d^\intercal + d^\intercal d$ one easily checks that this indeed defines a matrix
factorization of the quadratic form $q_4$. As seen above, the irreducible $\ZZ / 2$-graded Clifford module for $N =
4$ has dimension $4|4$, hence we know that $C^\bullet(\pzA)$ should decompose as
\be
C^\bullet (\mathpzc{A}) = C_{\mathpzc{irr},+}^\bullet (\pzA)\oplus
C_{\mathpzc{irr}, -}^\bullet (\pzA)
\ee
 where $C^\bullet_{\mathpzc{irr}, \pm} (\pzA)$ are the inequivalent complexes
 associated to the irreducible $N=4$ valise Adinkras, just like the one of
 equation \eqref{irrval}. Indeed, by a suitable change of basis to $d \in \mbox{Hom}
 (R^4, R^4)$, one finds
\begin{equation}
    C \cdot {d} \cdot B^{-1} =
\bmat
\textcolor{blue}{\lm_1} & \textcolor{orange}{\lm_2} & \textcolor{darkgreen}{\lm_3} & \textcolor{red}{\lm_4} & \textcolor{black}{0} & \textcolor{black}{0} & \textcolor{black}{0} & \textcolor{black}{0} \\
\textcolor{orange}{\lm_2} & \textcolor{blue}{-\lm_1} & \textcolor{red}{\lm_4} & \textcolor{darkgreen}{-\lm_3} & \textcolor{black}{0} & \textcolor{black}{0} & \textcolor{black}{0} & \textcolor{black}{0} \\
\textcolor{darkgreen}{\lm_3} & \textcolor{red}{-\lm_4} & \textcolor{blue}{-\lm_1} & \textcolor{orange}{\lm_2} & \textcolor{black}{0} & \textcolor{black}{0} & \textcolor{black}{0} & \textcolor{black}{0} \\
\textcolor{red}{\lm_4} & \textcolor{darkgreen}{\lm_3} & \textcolor{orange}{-\lm_2} & \textcolor{blue}{-\lm_1} & \textcolor{black}{0} & \textcolor{black}{0} & \textcolor{black}{0} & \textcolor{black}{0} \\
\textcolor{black}{0} & \textcolor{black}{0} & \textcolor{black}{0} & \textcolor{black}{0} & \textcolor{blue}{\lm_1} & \textcolor{orange}{\lm_2} & \textcolor{darkgreen}{\lm_3} & \textcolor{red}{\lm_4} \\
\textcolor{black}{0} & \textcolor{black}{0} & \textcolor{black}{0} & \textcolor{black}{0} & \textcolor{orange}{\lm_2} & \textcolor{blue}{-\lm_1} & \textcolor{red}{-\lm_4} & \textcolor{darkgreen}{\lm_3} \\
\textcolor{black}{0} & \textcolor{black}{0} & \textcolor{black}{0} & \textcolor{black}{0} & \textcolor{darkgreen}{\lm_3} & \textcolor{red}{\lm_4} & \textcolor{blue}{-\lm_1} & \textcolor{orange}{-\lm_2} \\
\textcolor{black}{0} & \textcolor{black}{0} & \textcolor{black}{0} & \textcolor{black}{0} & \textcolor{red}{\lm_4} & \textcolor{darkgreen}{-\lm_3} & \textcolor{orange}{\lm_2} & \textcolor{blue}{-\lm_1} \\
\emat.\end{equation}
which makes it apparent that the complex $C^\bullet (\pzA)$ is indeed the direct sum of two
complexes, in agreement with lemma \ref{irreducibleVal}. We can produce this complex directly from the two Adinkras shown in figure \ref{n4_disjoint}.
\begin{figure}[ht!]
    \centering
    \includegraphics[height=5cm]{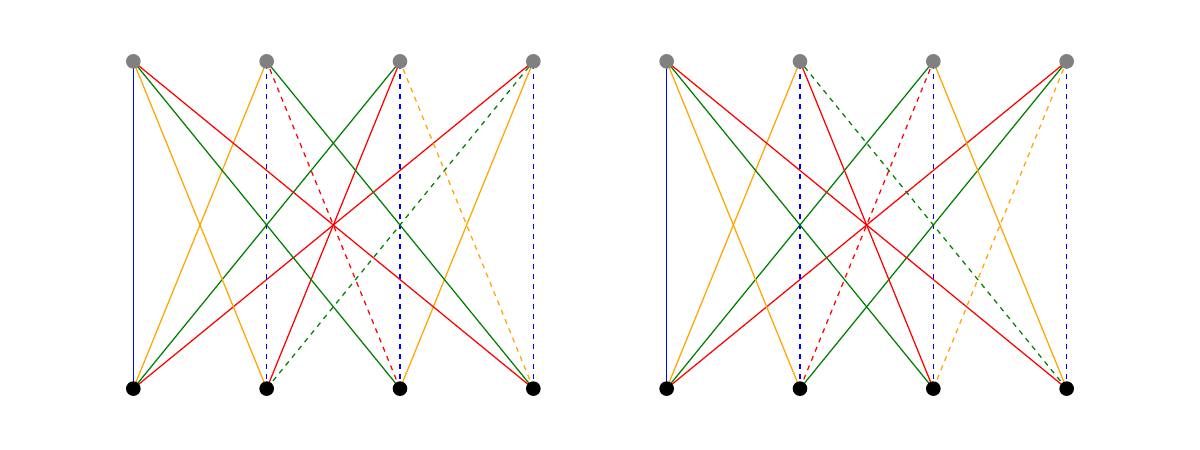}
    \caption{Disconnected valise Adinkras for $N = 4$.}
    \label{n4_disjoint}
\end{figure}
Physically, dimensional reduction of the $d=2$ $\mathcal{N} = (2,2)$ chiral and twisted chiral multiplets
\cite{Gates84b, Gates84} results in two distinct $(2,4,2)$ multiplets, which can be distinguished by
their two distinct (irreducible) valise Adinkras shown in
Figure~\ref{n4_disjoint}. One way to see that the two component Adinkras are
inequivalent is to consider the action of $\gamma_5$ on bosonic vertices. By
definition $\gamma_5$ is the product of generators $Q_4 Q_3 Q_2 Q_1$. The action
thus visually corresponds to going along the edges of color blue, yellow, green,
red in this order.  This interpretation of the action of $\gamma_5$ is whimsically described  as following the {\it rainbow} in \cite{Doran:2013cwa}.  Remarkably, for both Adinkras in Figure~\ref{n4_disjoint}, these
$4$-colored cycles always close (which is from for instance for the Adinkras
shown in Figure~\ref{n4_koszul_folded}). This means all vertices correspond to
eigenvectors of $\gamma_5$. On the left all bosonic vertices have eigenvalue $+1$ while
on the left all bosonic vertices have eigenvalue $-1$.  The homology of the two complexes
$C_{\mathpzc{irr}, \pm}^\bullet (\pzA)$ are the two modules $\mathbb{H}$ listed in Table~\ref{CLmodule}.  Since the quaternions are noncommutative, $\mathbb{H}$ can be viewed both as a left and a right $\mathbb{R}$ module.  Indeed, the two Adinkras in Figure~\ref{n4_disjoint} correspond to the left and right regular representation of the quaternions over $\mathbb{R}$ and encode the multiplication table of the quaternions.  Similarly, the left and right regular representation of the octonions correspond to the two irreducible modules for $N=8.$

The above discussion shows that an Adinkra encodes a complex \emph{together with a choice of basis}. Changing the basis
may yield an Adinkra with a different topology. However, the Adinkras giving rise to the same
complexes also give rise to the same representation of the supersymmetry algebra and vice versa.
Hence, while the topology of an Adinkra is not an invariant of a multiplet, their associated
complexes are (indeed complete) invariants.

\subsubsection{Bott-Kn\"orrer Periodicity for Adinkras}

\label{BottP}

As we have seen in section \ref{sec:valise}, multiplets coming from valise Adinkas are in one-to-one correspondence with $C\ell_{0}(q_N)$-modules and matrix
factorizations (up to a change of basis). Working over the real numbers, i.e.\ choosing $\kk = \RR$, it is well-known that real Clifford modules enjoy an 8-fold Bott
periodicity, and the same must be true for matrix factorizations -- indeed, an $8$-fold Kn\"orrer
periodicity for matrix factorization was proven in \cite{Brown16}\footnote{A relevant discussion is also
given in the recent \cite{Spellmann23}, which was inspired by \cite{Hori08}.}. In this section we aim to understand what ``tensoring" Clifford modules means from the point of view of Adinkras. \\
Let us start with two valise complexes
\begin{equation}
    \begin{tikzcd}
      0 \arrow[r] & \RR[\lm_1,\dots, \lm_K]^n \arrow[r, "\varphi"] & \RR[\lm_1,\dots,
        \lm_K]^n  \arrow[r] & 0, \\
      0 \arrow[r]  & \RR[\lm_{K+1}, \dots, \lm_N]^m \arrow[r, "\psi"] & \RR[\lm_{K+1}, \dots, \lm_N]^m  \arrow[r] & 0.
    \end{tikzcd}
\end{equation}
Posing $R = \RR[\lm_1,\dots, \lm_K] \otimes \RR[\lm_{K+1},\dots, \lm_N] = \RR[\lm_1,\dots, \lm_N]$ and denoting with abuse of notation the extension of the maps $\phi$ and $\psi$ to the ring $R$,
we can form a new complex:
\begin{equation}
    \begin{tikzcd}[ampersand replacement=\&]
        0 \arrow[r] \& R^n \otimes R^m \arrow[r, "{\tiny {\bmat \varphi \otimes \operatorname{id} \\
        -\operatorname{id} \otimes \psi \emat}}"] \& R^n \otimes R^m \oplus R^n \otimes
        R^m \arrow[r, "{\tiny {\bmat \operatorname{id} \otimes \psi \; \varphi \otimes \operatorname{id}
\emat}}"] \& R^n \otimes R^m \arrow[r] \& 0.
    \end{tikzcd}
\end{equation}
The associated valise complex is then
given by
\begin{equation}
    \begin{tikzcd}
      0 \arrow[r] & R^n \otimes R^m \oplus R^n \otimes R^m \arrow[r, "\tau"] &
        R^n \otimes R^m \oplus R^n \otimes R^m \arrow[r] & 0,
    \end{tikzcd}
\end{equation}
with the differential being
\begin{equation}
    \tau = \bmat
        \operatorname{id}_n \otimes \psi & \varphi \otimes \operatorname{id} \\
        \varphi^{\dagger} \otimes \operatorname{id}_m & -\operatorname{id}_n \otimes
        \psi^{\dagger}
    \emat.
\end{equation}
One easily checks that
\begin{equation} \label{mfA}
    f = \left ( \begin{array}{c|c}
    0 & \tau \\
    \hline
        \tau^{\dagger} & 0
    \end{array}
    \right )
\end{equation}
is a matrix factorization of $q_N$. This can also be understood as the graded tensor product of the
$\ZZ / 2$-graded matrix factorizations, where the minus sign comes from the Koszul sign rule.

In its matrix form \eqref{mfA}, the matrix factorization $f$ has rank $2 n m $. Specializing to the case $m = 8$ and taking
$\psi$ being a minimal matrix factorization of $q_8$, we find that the dimension of the related
$Cl_{0} (q_N, \mathbb{R})$-module is $16 n$. This implies that if $\varphi$ corresponds to an
irreducible Clifford module, then the Clifford module corresponding to $\tau$ is irreducible as well
(for dimensional reasons), and hence $f$ is a minimal matrix factorization of $q_N$.

On the level of the underlying graphs of Adinkras, the tensor product gives the product of the
graphs, and the above discussion shows that the product of two Adinkras graphs admits a consistent
coloring and dashing. Although we obtain presentations of the minimal matrix factorization for each
$N \ge 8$ by tensoring a minimal factorization with $N - 8$ with a $N = 8$ minimal matrix
factorization, some topologies are not products of simpler graphs -- we remark that all the possible topologies were
extensively studied for example in \cite{Doran08}.

\subsection{Vertex raising and mapping Cones} 
\label{sec:flipping}
Given any Adinkra $\mathpzc{A}$, we can canonically associate to it a valise Adinkra, $\mathpzc{A}_V$, by taking its
$\ZZ$-grading mod $2$, i.e.\ by reducing the degrees modulo 2. In particular, the topology of the graph and the dashing are unchanged and
there is (up to permutation of the vertices) a unique valise Adinkra for a given coloring, dashing,
and topology. 
The complex associated to $\mathpzc{A}_V$ is of the form
\bc
\begin{tikzcd}
0 \arrow[r] &  R^{\oplus 2^n} \arrow[r, "d"] & R^{\oplus 2^n} \arrow[r] & 0,
\end{tikzcd}
\ec
for some $n$.

In the other direction, instead of projecting down to valise Adinkras, one can
try to \enquote{raise} vertices via the following procedure.
\begin{construction}[Vertex raising]  \label{vertflip} Let $\pzA$ be an Adinkra. \begin{enumerate}
\item Let $v \in \mathpzc{V}(\pzA)$ be a vertex\footnote{
We remark that one such vertex exists since a partially ordered finite set always contains a smallest element.}
such that there is no edge $e \in \mathpzc{E}(\pzA)$ with $\mathpzc{t}(e) = v$. Analogously, $v \in \mathpzc{V}(\pzA)$ is such that
for each edge connecting $v$ with $v' \in \mathpzc{V} (\pzA)$ then $\mathpzc{h}_{\pzA}(v') =
\mathpzc{h}_{\pzA}(v) + 1 > \mathpzc{h}_{\pzA}(v)$.
\item Raise the vertex $v \in $
$\mathpzc{V} (\pzA)$ to bring it two degrees up, so that after raising all edges attached to it go down one degree, instead of going up as in the original Adinkra.
\end{enumerate}
We call this construction \emph{vertex raising}.
\end{construction}
It is a simple exercise to verify the following, whose proof is left to the reader.
\begin{lemma} Let $\pzA$ be an Adinkra and let $v \in \mathpzc{V} (\pzA)$ be as
    in construction \ref{vertflip}. The finite graph obtained by raising a vertex is an Adinkra $\pzA^\prime$.
\end{lemma}
In particular, the Adinkra $\pzA^\prime$ obtained by raising a vertex has the
same data of the original $\pzA$ except that now $\mathpzc{h}_{\pzA'}(v) =
\mathpzc{h}_{\pzA}(v) + 2$ and there are no edges $e \in \mathpzc{E} (\pzA^\prime)$
such that with $\mathpzc{s} (e) = v$. In other words, after raising, $v$ is two
degrees higher and all edges attached to it now go down one degree, instead of
going up as in $\pzA$. Let us now look at an example.

\begin{example}[Vertex Raising in $N=3$] Let us consider the $N=3$ Adinkra characterized by the rank sequence $(1, 4,
3)$, as shown in figure \ref{n3_1_4_3}.
\begin{figure}[ht!]
    \centering
    \includegraphics[width=7cm]{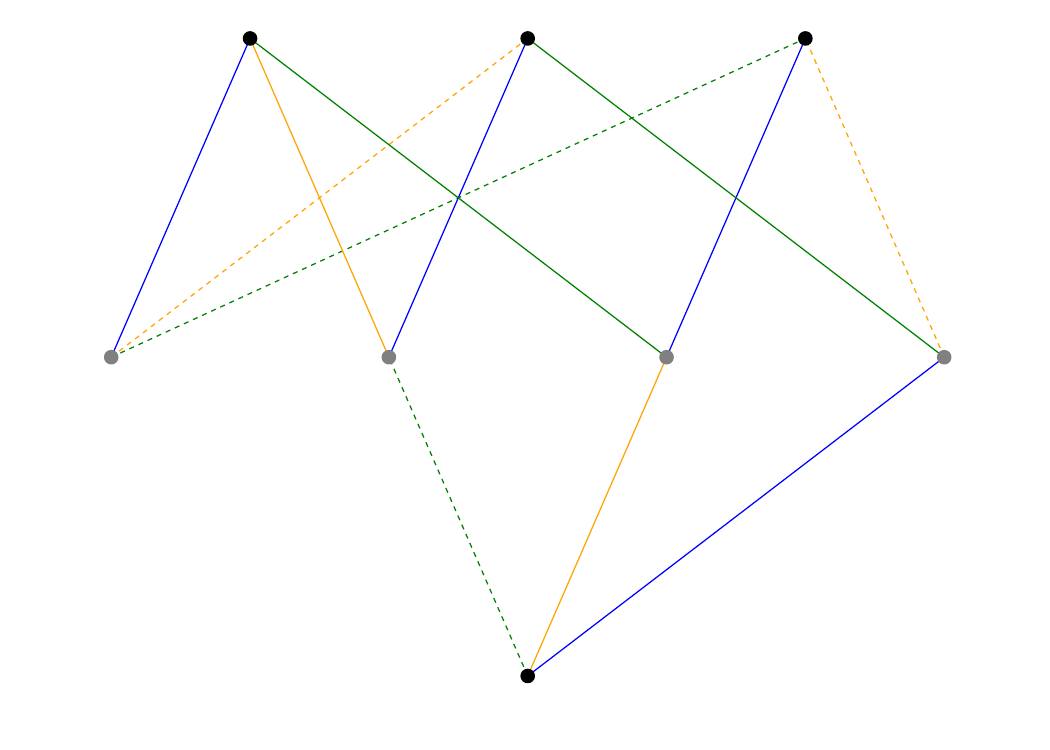}
    \caption{The $(1,4,3)$ Adinkra of $N = 3$.}
    \label{n3_1_4_3}
\end{figure}
Assuming the lowest vertex has degree 0, there is a vertex $v$ with $\mathpzc{h}(v) = 1$ whose attached edges are all going up.
Hence, by raising a vertex, we can obtain another Adinkra $\mathpzc{A}^\prime$, where the only difference is
that now $\mathpzc{h}(v) = 3$. This Adinkra, characterized by the rank sequence $(1,3,3,1)$, is shown in figure \ref{n3_1_3_3_1}.
\begin{figure}[ht!]
    \centering
    \includegraphics[width=7cm]{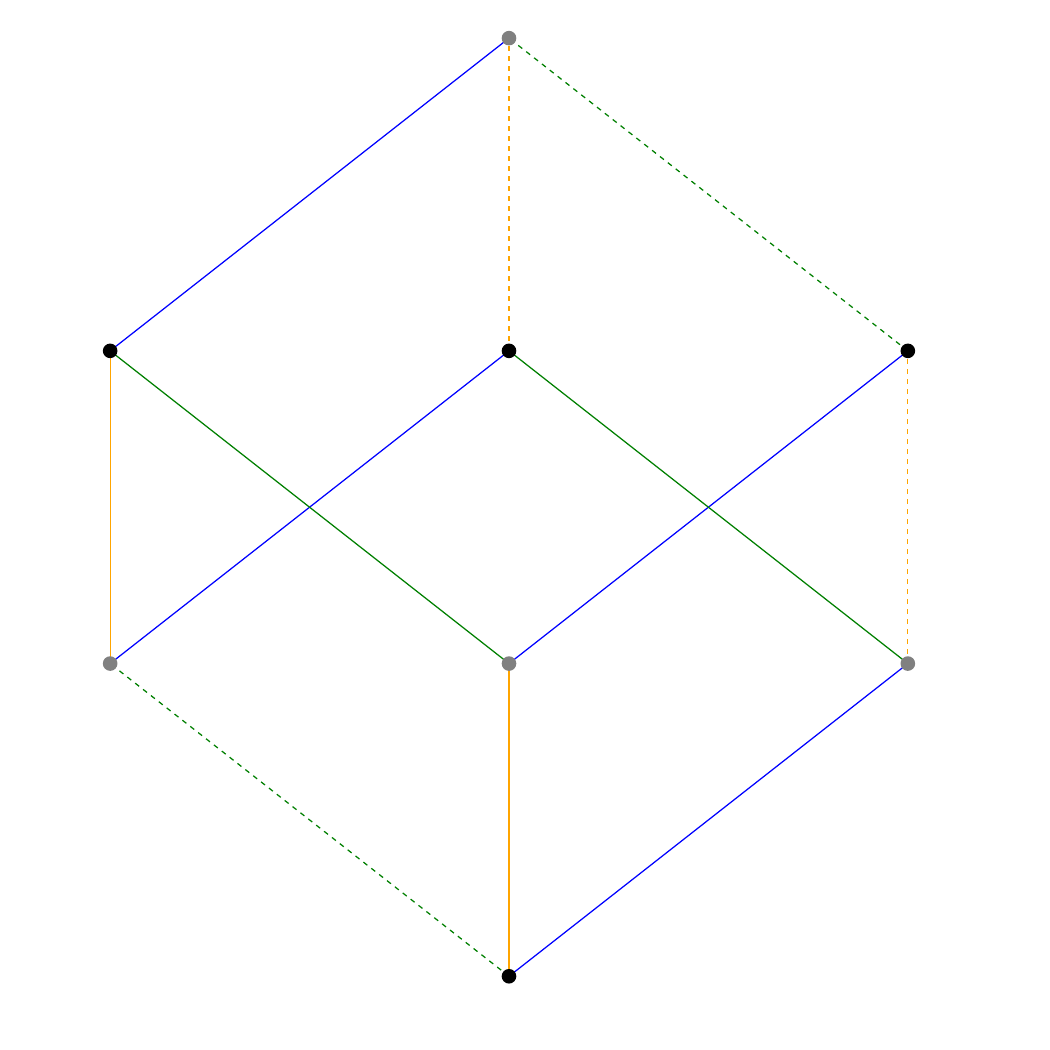}
    \caption{The $(1,3,3,1)$ Adinkras of $N = 3$.}
    \label{n3_1_3_3_1}
\end{figure}
This is an example of a Koszul Adinkra, as first introduced in \ref{KAd}: its associated complex $C^\bullet (\pzA^\prime)$ is indeed the Koszul complex, resolving the $R/q_3$-module $\kk$ in $R$-modules.
\end{example}


Given the previous example, we now aim at understanding {vertex raising} from the perspective
of the Adinkra complexes $C^{\bullet}(\pzA)$ --- this will lead us to another abstract result, theorem \ref{flipping_theorem}, that shows that
the operation of vertex raising is realized by taking a \emph{cone} in the derived category.

Before we discuss this abstract result, though, we consider the simpler case where both
$C^{\bullet}(\pzA)$ and $C^{\bullet}(\pzA')$ are quasi-isomorphic to $R / I$-modules concentrated in
just one degree. For concreteness, we will start with an example: namely, we look again the $N = 4$ irreducible valise Adinkra
$\pzA$ from the previous section and shown in figure \ref{adk_n4_4_4}.

\begin{example}[Complexes and Vertex Raising in $N=4$] Starting from the
    irreducible $N=4$ valise Adinkra, we can raise any vertex in the
bottom row. We choose the one to the outer right: this results in the Adinkra $\pzA'$ shown
in figure \ref{n4_3_4_1}.
\begin{figure}
    \centering
    \includegraphics[width=9cm]{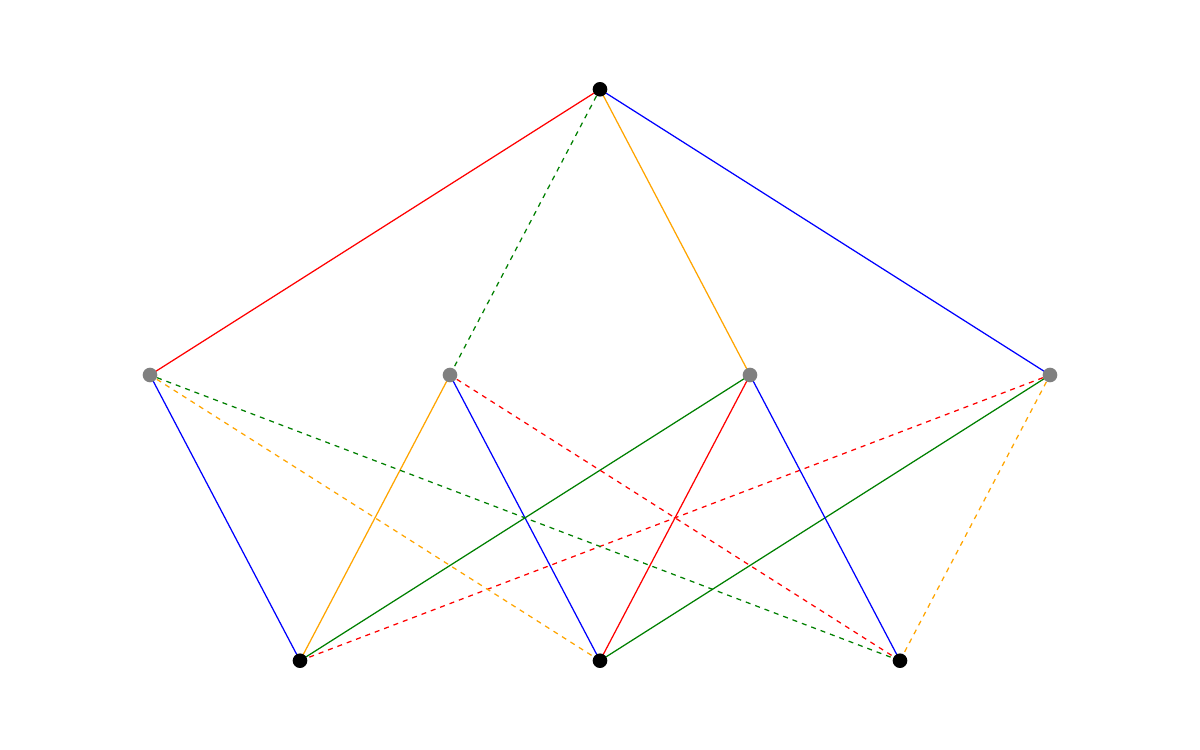}
    \caption{The $(3,4,1)$ Adinkras of $N = 4$.}
    \label{n4_3_4_1}
\end{figure}

We have seen above that the differential $d^0 = d$ of the valise Adinkras realizes a matrix factorization of
$q_4$ (all of the higher differentials are zero). Raising the vertex has the effect of eliminating the last row from the matrix $d$ as given in \eqref{matrixD}, and we have a new differential
\[
\begin{tikzcd}
   d^0_{\pzA} = \bmat \textcolor{blue}{\lm_1} & \textcolor{orange}{\lm_2} & \textcolor{darkgreen}{\lm_3} &
\textcolor{red}{-\lm_4} \\
\textcolor{orange}{-\lm_2} & \textcolor{blue}{\lm_1} & \textcolor{red}{\lm_4} &
\textcolor{darkgreen}{\lm_3} \\
\textcolor{darkgreen}{-\lm_3} & \textcolor{red}{-\lm_4} & \textcolor{blue}{\lm_1} &
\textcolor{orange}{-\lm_2} \\
\textcolor{red}{\lm_4} & \textcolor{darkgreen}{-\lm_3} & \textcolor{orange}{\lm_2} &
\textcolor{blue}{\lm_1} \emat
 \arrow[r, rightsquigarrow, "\mathpzc{Raise}"] &  \bmat \textcolor{blue}{\lm_1} & \textcolor{orange}{\lm_2} & \textcolor{darkgreen}{\lm_3} & \textcolor{red}{-\lm_4} \\
                        \textcolor{orange}{-\lm_2} & \textcolor{blue}{\lm_1} & \textcolor{red}{\lm_4} & \textcolor{darkgreen}{\lm_3} \\
                        \textcolor{darkgreen}{-\lm_3} & \textcolor{red}{-\lm_4} & \textcolor{blue}{\lm_1} & \textcolor{orange}{-\lm_2} \emat = d^0_{\pzA'}
.
\end{tikzcd}\]
It is easy to see that
\[
    d^{0}_{\pzA'} \cdot (d^{0}_{\pzA'})^\intercal = q_4 \cdot \mbox{id}_3,
\]
therefore we have again a matrix factorization of $q_4$. On the other hand, one immediately sees that the kernel of $d^0_{\pzA^\prime}$ is generated by (the transpose of) the last row of $d^0_\pzA$, which has been eliminated from $d^0_{\pzA^\prime}$, i.e.\
\be
\ker (d^0_{\pzA^\prime}) = R \cdot \langle ( \lm_4 , - \lm_3, \lm_2, \lm_1 )^t \rangle.
\ee
Hence, calling $\hat v \defeq ( \lm_4 , - \lm_3, \lm_2, \lm_1 )^t $ and posing
\begin{equation}
\xymatrix@R=1.5pt{
d^\prime_{\pzA^\prime} : R^1 \ar[r] & R^4 \\
x \ar@{|->}[r] & \hat v \cdot x,
}
\end{equation}
one immediately gets the complex associated to $\mathpzc{A}^\prime$,
\begin{equation} \label{Aprime}
\begin{tikzcd}
C^\bullet (\mathpzc{A}^\prime) \equiv \big ( 0 \arrow[r] & R^1 \arrow[r, "d^1_{\pzA'} "]  & R^4
    \arrow[r, "d^0_{\pzA'}"] & R^3 \arrow[r] & 0 \big ),
\end{tikzcd}
\end{equation}
which is quasi-isomorphic to $\coker d^0_{\pzA^\prime}$ by construction.
Finally, observe that the image of $d^0_{\pzA'}$ in $R^3$ generates an ideal containing $q_4$, and hence the
cokernel of $d^0_{\pzA'}$ descends to a module over $R / q_4$, as expected.\\
One can repeat the above procedure and obtain Adinkras -- or complexes -- characterized by the following rank sequences.
\begin{equation}
\begin{tikzcd}
(4,4)\arrow[r, rightsquigarrow, "\mathpzc{Raise}"] & (3,4,1) \arrow[r,
    rightsquigarrow, "\mathpzc{Raise}"] & (2,4,2) \arrow[r, rightsquigarrow,
    "\mathpzc{Raise}"] & (1,4,3).
\end{tikzcd}
\end{equation}
On the other hand, it is to be observed that the last complex, characterized by the $(1,4,3)$ rank sequence, has higher cohomology, i.e.\ it is not quasi-isomorphic to an
$R/I$-module concentrated in one degree. This explains why we need to consider
the full derived category $\mbox{\sffamily{D}}^{\flat}(R/I\mbox{-\sffamily{Mod}})$, rather
than just the abelian category of $R/I$-modules.

\end{example}

In the case explained above, vertex raising has a neat geometric interpretation as taking suitable
quotients of the spinor bundle by line bundles. As above, we start illustrating this in an example:
namely, we will consider the Adinkra characterized by the rank sequence $(3,4,1)$ viewed as the Adinkra
coming from the irreducible $N=4$ valise Adinkra, as in figure \ref{n4_3_4_1} above.
\begin{example}[Geometry of Vertex Raising] Starting from the $N=4$ irreducible
    valise Adinkra $\pzA$, we denote $\mathpzc{A}^\prime$ the vertex raised
Adinkra with rank sequence $(3,4,1)$ and call $v$ the raised vertex.
We have a canonical isomorphism $C^{0}(\pzA) \cong
C^{0}(\pzA') \oplus R v$. Denoting by $\pi$ the projection of $C^{0}(\pzA')$ to the direct
summand $R v$ and by $\iota_v$ the inclusion, we have the following (split) exact sequence
\begin{equation} \label{iotapi}
\xymatrix{
0 \ar[r] & R v \ar[r]_{\iota_v} & \ar@/_1pc/[l]_{\pi_v} C^0(\mathpzc{A}) \ar[r]^{\pi_{\mathpzc{A}^\prime}} & C^0 (\mathpzc{A}^\prime) \ar[r] & 0,
}
\end{equation}
where $\pi_{\mathpzc{A}^\prime} : C^{0} (\pzA) \rightarrow C^{0} (\pzA^\prime)$ denotes the projection, and where $\pi \circ \iota_v = id_{Rv}.$
Employing this notation, it is easy to see that the above short exact sequence can be completed to the
following commutative diagram
\be
\begin{tikzcd} \label{cdvf}
    0 \arrow[r] & 0 \arrow[r] \arrow[d]                    & C^{-1}(\mathpzc{A}^\prime) \arrow[r, "id"]
    \arrow[d, "d^0_{\pzA}"] & C^{-1} (\mathpzc{A}) \arrow[r] \arrow[d, "d^{0}_{\pzA'}"] & 0 \\
    0 \arrow[r] & Rv \arrow[r, "\iota_v"] &  C^0(\mathpzc{A}) \arrow[r, "\pi_{\pzA'}"]                       & C^0(\mathpzc{A}^\prime) \arrow[r]                & 0,
\end{tikzcd}
\ee
so that one has $d^{0}_{\pzA'} = \pi_{\pzA'} \circ d^0_{\pzA'}$. In terms of free $R$-modules this corresponds to
\be
\begin{tikzcd}
    0 \arrow[r] & 0 \arrow[r] \arrow[d]                    & R^4 \arrow[r, "id"]
    \arrow[d, "d^0_{\pzA}"] & R^4 \arrow[r] \arrow[d, "d^{0}_{\pzA'}"] & 0 \\
    0 \arrow[r] & R \arrow[r, "\iota_v"] &  R^4 \arrow[r, "\pi_{\pzA'}"]                       & R^3 \arrow[r]                & 0.
\end{tikzcd}
\ee
By the snake lemma, we get the following exact sequence
\bc
\begin{tikzcd}
    0 \arrow[r] & R \arrow[r, "\cdot q_4"] & R \arrow[r] & \coker d^0_{\pzA} \arrow[r]
    & \coker d^0_{\pzA'}  \arrow[r] & 0.
\end{tikzcd}
\ec
It is an easy diagram-chasing exercise to verify that the first map is given by the multiplication by the
quadratic form $q_4$. Replacing the first map with its cokernel one finds the following (extension) short exact sequence
\bc
\begin{tikzcd}
    0 \arrow[r] & R / I \arrow[r] & \coker d^0_{\pzA} \arrow[r]
    & \coker d^0_{\pzA'}  \arrow[r] & 0
,\end{tikzcd}
\ec
where we observe that the module $\coker (d^0_{\pzA})$ corresponds to the spinor bundle, according to lemma \ref{irreducibleVal}.
\end{example}

The above description holds in full generality, for any $N$. Indeed, assume we have an Adinkra characterized by a
rank sequence $(n_0, n_1, n_2)$, whose related complex has cohomology concentrated in degree zero, {i.e.}\ the complex is the free resolution of some
$R / I$-module obtained from vertex raising. Then, setting up a commutative diagram as above in equation \eqref{cdvf}, gives an exact sequence
\be
\bcd 0 \arrow[r] & R^{n_2} \arrow[r, "q_N"] & R^{n_2} \arrow[r] &
\operatorname{coker} d^{0}_{\pzA} \arrow[r] & \coker {d^0_{\pzA'}} \arrow[r] & 0, \ecd
\ee
where $d^{0}_{\pzA'}: R^{n_1} \rightarrow R^{n_0}$ is the first map in the complex characterized by the rank sequence
$(n_0, n_1, n_2)$ and obtained from raising multiple vertices from level $0$ to level $2$. It is clear that
$n_0 + n_2 = n_1$ and hence $n_2 = n_1 - n_0$. Once again, diagram-chasing shows that the first map is given by the multiplication by the quadratic form $q_N$.
It follows that one gets the exact sequence
\be
\bcd 0 \arrow[r] & (R / I)^{n_2} \arrow[r] & \operatorname{coker} d^0_{\pzA} \arrow[r] &
\operatorname{coker} d^0_{\pzA^\prime} \arrow[r] & 0 \ecd
\ee and hence it is clear that $\coker d^{0}_{\pzA'} \cong
\coker d^{0}_{\pzA} / (R / I)^{\oplus n_2}$. We emphasize again that this only works if the complex
has \emph{no higher cohomology}: indeed, in this case $\operatorname{dim}( \operatorname{ker} d^0_{\pzA'}
) \neq n_2$.

This explanation is not yet complete. Indeed, as we have seen, the complexes related to Adinkras are to be
viewed as elements in a derived category, and the right way of \enquote{taking quotients} in the derived category
is to take the cone of a map.
In view of this, we will now show that vertex raising on Adinkras is equivalent to taking a suitable cone in the derived
category of $R / I$-modules. More precisely, the following theorem holds true.
\begin{theorem}[Vertex Raising \& Mapping Cones] Let $C^\bullet (\pzA)$ be the complex
    associated to an Adinkra $\pzA$.
    \label{flipping_theorem}
    Let $0 \neq v \in C^{-i}(\pzA)$ such that $d v = 0$. Then there
    exists a morphism $j: R / I[i] \rightarrow C^{\bullet}(\pzA)$ in the derived category $\mbox{\emph{\sffamily{D}}}^\flat (R/I)$, represented by a
    cochain map
    \be
    \begin{tikzcd}
        \label{flipping_map}
        \ldots \arrow[r] & 0 \arrow[r] \arrow[d] & R \arrow[r, "q_N"] \arrow[d, "1 \mapsto
        d^{\dagger}v"] & R \arrow[r] \arrow[d,
        "1 \mapsto v"] & 0 \arrow[d] \arrow[r] & \ldots \\
        \ldots \arrow [r] & C^{-i- 2}(\pzA) \arrow[r, "d"] & C^{-i - 1}(\pzA)
        \arrow[r, "d"] & C^{-i}(\pzA) \arrow[r, "d"] & C^{-i+1}(\pzA) \arrow[r, "d"] & \ldots .
    \end{tikzcd}
    \ee
    In particular, if $v \in C^{-i} (\mathpzc{A})$ corresponds to a vertex of $\mathpzc{A}$, 
    the mapping cone $\operatorname{Cn}^\bullet \left(j: R / I[i] \rightarrow C^{\bullet}(\pzA)\right)$ of $j$ is
    quasi-isomorphic to the complex $C^{\bullet}(\pzA')$, where
    $\pzA'$ is the Adinkra with the vertex corresponding to $v$ raised up
    from level $i$ to level $i + 2$.
\end{theorem}
\begin{proof}
    Since $dv = 0$ we have $\Delta v = d d^{\dagger} v = q_N v$ and hence
    \eqref{flipping_map} commutes, showing that $j$ is a cochain map. \\
    Assume now that there is a vertex spanning the free rank-one submodule of $C^{-i}(\pzA)$ containing
$v$. By abuse of notation, we call this vertex $v$ as well: the condition $d v= 0$ is
equivalent to the condition that all edges attached to $v$ go up, hence the
    vertex $v$ can be raised.
Calling the Adinkra $\mathpzc{A}^\prime$ and its complex $C^{\bullet} (\mathpzc{A}^\prime)$
we have that the raised vertex will be of degree $i +2$ and will
correspond to a generator in $C^{-i-2} (\mathpzc{A}^\prime).$
Now consider the mapping cone complex $\operatorname{Cn}^\bullet(j)$ of the map $j$:
    \begin{equation} \label{conej}
        \begin{tikzcd}[ampersand replacement=\&]
            \ldots \arrow[r] \& C^{-i- 2}(\pzA) \oplus Rv \arrow[r,
            "{\tiny {\begin{pmatrix} d &
            d^{\dagger} \\ 0 & -q_N \end{pmatrix}}}"] \&
            C^{-i - 1}(\pzA) \oplus Rv \arrow[r, "{\tiny {\begin{pmatrix} d & id
            \end{pmatrix}}}"]
            \& C^{-i}(\pzA) \arrow[r, "d"] \& C^{-i+1}(\pzA) \arrow[r] \& \ldots.
        \end{tikzcd}
    \end{equation}
  We claim that $C^{\bullet}(\pzA')$ is quasi-isomorphic to $\operatorname{Cn}^\bullet(j)$.



Recalling that $C^{-i-2}(\mathpzc{A}^\prime) = C^{-i-2} (\mathpzc{A}) \oplus Rv$ and $C^{-1} (\mathpzc{A}^\prime) = C^{-i} (\mathpzc{A})$, we first consider the following diagram, where the top row corresponds to $C^\bullet (\mathpzc{A}^\prime)$:
    \be
        \begin{tikzcd}[ampersand replacement=\&]
            \ldots \arrow[r] \& C^{-i- 2}(\pzA) \oplus Rv \arrow[r,
                "{\tiny {\begin{pmatrix} d & d^{\dagger}
                                \end{pmatrix}}}"] \arrow[d, "id"]\&
                    C^{-i - 1}(\pzA) \arrow[r, "d^{-i-1}_{\pzA'}"] \arrow[d, "{\tiny {\begin{pmatrix} id \\
                    -\pi \circ d \end{pmatrix}}}"]
                    \& C^{-i}(\pzA') \arrow[r] \arrow[d, hook] \& C^{-i+1}(\pzA)
                    \arrow[r] \arrow[d, "id"] \&
                    \ldots \\
                    \ldots \arrow[r] \& C^{-i- 2}(\pzA) \oplus Rv \arrow[r,
                    "{\tiny {\begin{pmatrix} d &
             d^{\dagger}  \\ 0 & -q_N \end{pmatrix}}} "] \&
            C^{-i - 1}(\pzA) \oplus Rv \arrow[r, "{\tiny {\begin{pmatrix} d & \iota_v
            \end{pmatrix}}}"]
            \& C^{-i}(\pzA) \arrow[r] \& C^{-i+1}(\pzA) \arrow[r] \& \ldots,
        \end{tikzcd}
    \ee
    where $\pi : C^{-i} (\pzA) \rightarrow Rv$ denotes the projection onto $R v \subseteq C^{-i} (\pzA)$
    and $\iota_v : Rv \hookrightarrow C^{-i} (\pzA^\prime) \oplus Rv \cong C^{-i} (\pzA)$ denotes the canonical immersion
    as in \eqref{iotapi}.
    Using these maps, the differential $d^{-i-1}_{\mathpzc{A}^\prime} : C^{-i-1} (\mathpzc{A}^\prime) \cong C^{-i-1} (\mathpzc{A}) \rightarrow C^{-i} (\mathpzc{A}^\prime)$
    can be written in terms of the original $d^{-i}_{\mathpzc{A}}$ as
    $
     d^{-i-1}_{\mathpzc{A}^\prime } = (id - \iota_v \circ \pi) \circ d^{-i-1}_{\mathpzc{A}}.
    $
    Using this, it is not hard to verify that the above diagram is commutative and hence defines a cochain map
    $f^{\bullet} : C^{\bullet} (\mathpzc{A}^\prime) \rightarrow \operatorname{Cn}^\bullet (j).$
     Similarly, we have another diagram $p : \operatorname{Cn}^\bullet (j) \rightarrow C^{\bullet} (\mathpzc{A}^\prime)$, given by
       \be
        \begin{tikzcd}[ampersand replacement=\&]
            \ldots \arrow[r] \& C^{-i- 2}(\pzA) \oplus Rv \arrow[r,
                "{\tiny{\begin{pmatrix} d & d^{\dagger} \end{pmatrix}}}"] \arrow[d,
                    "id", leftarrow]\&
                    C^{-i - 1}(\pzA) \arrow[r, "d_{\pzA'}"] \arrow[d, "{\tiny {\begin{pmatrix} id
                    & 0 \end{pmatrix}}}", leftarrow]
                    \& C^{-i}(\pzA') \arrow[r] \arrow[d, "id - \iota_v \circ \pi", leftarrow] \& C^{-i+1}(\pzA)
                    \arrow[r] \arrow[d, "id", leftarrow] \&
                    \ldots \\
                    \ldots \arrow[r] \& C^{-i- 2}(\pzA) \oplus Rv \arrow[r,
                    "{\tiny {\begin{pmatrix} d &
             d^{\dagger}  \\ 0 & -q_N \end{pmatrix}}}"] \&
            C^{-i - 1}(\pzA) \oplus Rv \arrow[r, "{\tiny \begin{pmatrix} d & \iota_v
            \end{pmatrix} }"]
            \& C^{-i}(\pzA) \arrow[r] \& C^{-i+1}(\pzA) \arrow[r] \& \ldots
        \end{tikzcd}
    \ee
    Upon remembering that $\pi \circ \iota_v = id,$ it is easy to verify that also the above diagram is commutative and hence
    one defines a cochain map $p^\bullet : \operatorname{Cn}^\bullet (j) \rightarrow C^\bullet (\mathpzc{A}^\prime).$
One immediately sees that the composition $p^\bullet \circ f^\bullet$ is the identity on $C^{\bullet}(\pzA')$. The composition $f^\bullet \circ p^\bullet$ is the identity on $\operatorname{Cn}^\bullet (j)$, except in degrees $-i-1$ and $-i$. \\
We will now define a homotopy $h^\bullet : \operatorname{Cn}^\bullet (j) \rightarrow \operatorname{Cn}^{\bullet-1} (j)$ between $f^\bullet \circ p^\bullet$ and the identity on $\operatorname{Cn}^\bullet (j).$ Consider first degree $-i$: the homotopy $h^{-i + 1}$ is zero and hence one needs that $d^{-i-1} \circ h^{-i}
= 1 - f^{i} \circ p^i = \iota_v \circ \pi$. Then it is enough to choose
$
h^{-i} \defeq ( 0 , \pi )^t.
$
Next, we look at degree $- i - 1$. The homotopy $h^{-i-2}$ is zero, and the map $f^{-i-1} \circ p^{-i-1}$ is given by $(x , rv)^t \mapsto (x , -\pi \circ d x)^t$. On the other hand $h^{-i}\circ d^{-i-1}$ acts as $(x , r v)^t \mapsto (0, \pi \circ d x + r v)$ and it follows that $f^{-i -1}\circ p^{-i-1} +  h^{i-1} \circ d^{-i-1} = id^{-i-1}$, thus concluding the proof.
\end{proof}

\subsubsection{Vertex raising in \texorpdfstring{$U_{\kk}
(\ft)$-\sffamily{\emph{Mod}}}{U_k(t)-mod}}

In the previous section we focused on the operation of vertex raising from the point of view of the derived category
of $R/I$-modules. On the other hand, in view of the Bernstein--Gel'fand--Gel'fand correspondence as given in the form
\ref{Koszul1}, the previous theorem \ref{flipping_theorem} can be given an interpretation in
also in terms of $U_\kk (\ft)$-modules. Indeed, roughly speaking, the quotient by $\mbox{\sffamily{{Art}}} (U_\kk (\ft))$ in theorem \ref{Koszul1}
identifies Adinkras related by vertex raising in the Verdier quotient.


Following definition \ref{inthesense} and the remark after \ref{N_Adinkra}, we denote by $\mathcal{V}(\pzA)$ the local
field representation (multiplet) of corresponding to the Adinkra $\pzA$. 
Just like above, consider some $\pzA'$ coming from $\pzA$ by raising a vertex $v
\in \mathpzc{V} (\mathpzc{A})$ with $d v = 0$ in $C^{\bullet}(\pzA)$. This gives
a
distinguished triangle
   \be
        \begin{tikzcd}
            R/I[i] \arrow[r, "j"] & C^{\bullet}(\pzA) \arrow[r] &
            C^{\bullet}(\pzA').
        \end{tikzcd}
   \ee
   where $j : R/I[i] \rightarrow C^\bullet (\mathpzc{A})$ as in \ref{flipping_map} and $C^\bullet (\mathpzc{A}^\prime)$ is
   quasi-isomorphic to $\operatorname{Cn}^\bullet (j)$ according to theorem \ref{flipping_theorem}.
   After applying the pure spinor functor $\mathcal{A}^{\bullet}$, we get another distinguished triangle
   \be
        \begin{tikzcd}
            \kk \arrow[r] & \mathcal{V}(\pzA) \arrow[r] & \mathcal{V}(\pzA'),
        \end{tikzcd}
   \ee
where we have used that $\mathcal{A}^\bullet (R / I)$ yields the de Rham complex, which in turn is quasi-isomorphic to $\kk$.
Since all of the above are honest representations, this distinguished triangle is actually an ``exact sequence'' in the
category of $U_\kk (\ft)$-modules -- the first map being injective and the second map surjective.
Physically, the trivial submodule $\kk \subset \mathcal{V}(\pzA)$ should be thought of as a {\it zero mode},
{i.e.}\ it is annihilated by the action of all $Q$'s and the time translation $H$. In the free
superfield (which is indeed an element of the de Rham complex) such a submodule is given by \enquote{evaluating} at $\theta_i = t = 0$: this amounts to
first modding out the fermionic directions of the related supermanifold, and then evaluate the field
on the reduced space at $t=0$.

\begin{figure}[ht!]
    \centering
    \raisebox{.5cm}{
    \includegraphics[height=4cm]{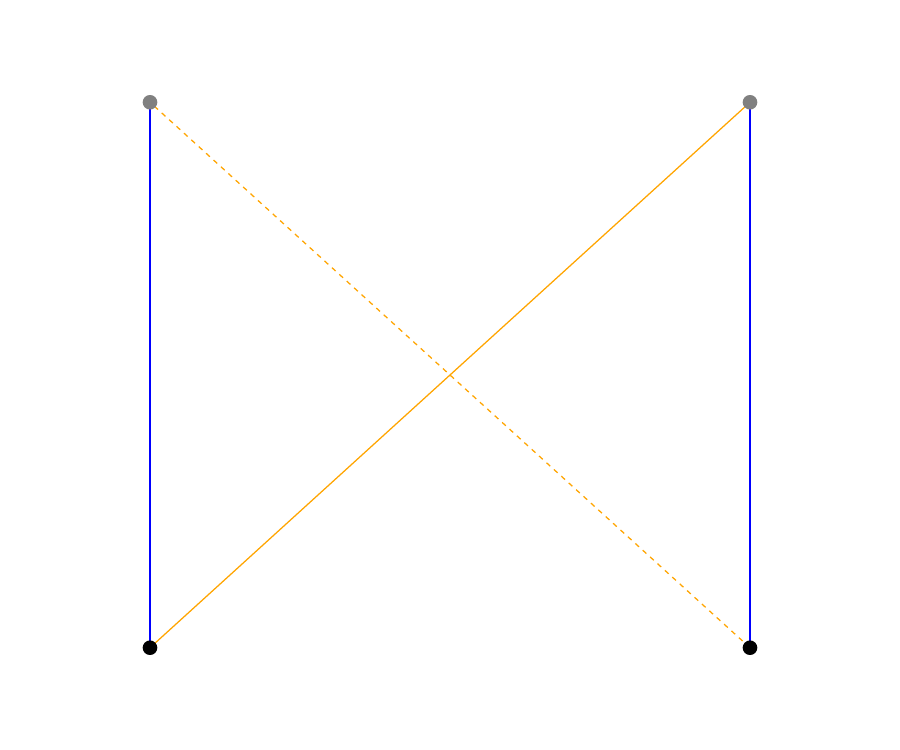} \vspace{.5cm}}
    \qquad
    \includegraphics[height=5cm]{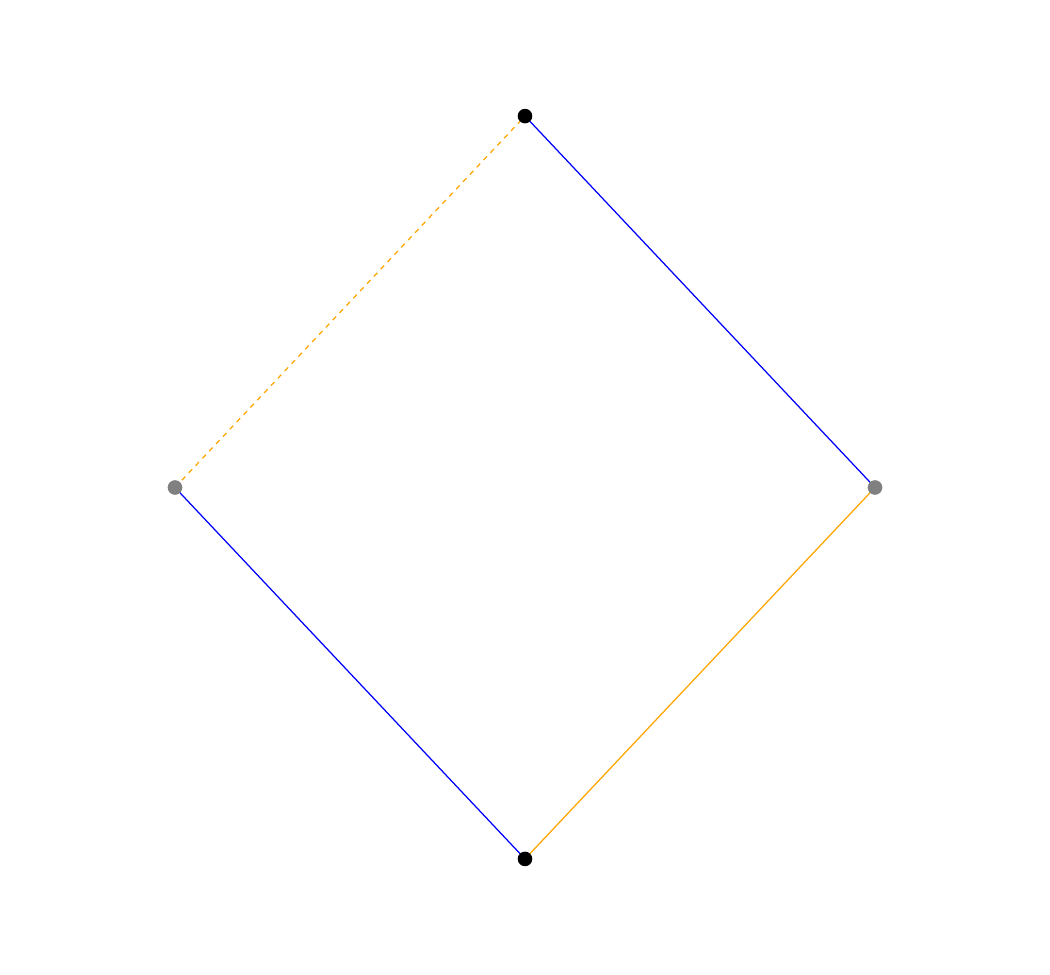}
    \caption{Valise (left) and Koszul (right) Adinkras for $N=2$.}
    \label{n2_hd}
\end{figure}
\begin{example}[Reprise: $N=2$ Valise \& Koszul Adinkras] It is possible to reinterpret the $N =2$ (irreducible) valise and the $N=2$ Koszul Adinkra introduced in section \ref{representations} considering the above discussion. The multiplets related to the irreducible valise and Koszul Adinkras in figure \ref{n2_hd} sits inside an exact sequence
   \be
        \label{Ug_sequence}
        \begin{tikzcd}
            0 \arrow[r] & \kk \arrow[r] & \mathcal{V}(\pzA_{V}) \arrow[r] &
            \mathcal{V}(\pzA_{K}) \arrow[r] & 0,
        \end{tikzcd}
   \ee
where we have denoted the valise and the Koszul Adinkras by $\pzA_V$ and $\pzA_{K}$ respectively. 
Recall that we described the multiplet $\mathcal{V} (\mathpzc{A}_V)$ in terms of two copies of the free superfield with the constrained coefficients
    \be
        \begin{split}
        X(t, \theta_1, \theta_2) &= x(t) +  \theta_1 \psi (t)- \theta_2 \chi(t) + \theta_2\theta_1 \dot{y} (t) \\
        Y(t, \theta_1, \theta_2) &= y(t) + \theta_1 \chi(t)  + \theta_2 \psi (t)  - \theta_2\theta_1\dot{x} (t).
        \end{split}
    \ee
There are two zero-modes, given by $x(0)$ and $y(0)$. Vertex raising allows one to cancel either of these.
For example, removing the constant term $y(0)$ from
$y$ has the effect that $\dot{y}$ already captures all the information -- in particular, the
superfield $X$ has now enough information to describe the quotient $\mathcal{V}(\pzA_V) / \kk$. Indeed,
as it is witnessed by the short exact sequence, the superfield $X$ becomes a free superfield and as such it is
 related to the Koszul Adinkra $\mathpzc{A}_K$.
\eqref{Ug_sequence}.

Finally, let us stress that another way to read sequence \eqref{Ug_sequence} is that we obtain the multiplet $\mathcal{V}(\pzA_V)$ as
an extension of $\mathcal{V}(\pzA_{K})$ by the trivial $U_{\kk}(\ft)$-module $\kk$. We will come back
to this point of view later in this manuscript.
\end{example}

\subsection{Cohomology modules: two \texorpdfstring{$N = 6$}{N=6} multiplets of rank \texorpdfstring{$(2,8,6)$}{(2,8,6)}}
\label{sec:counterexample}

In the previous sections, we have extensively employed the rank sequence $(n_1, \ldots, n_\ell)$ as a
useful bookkeeping device to keep track of different Adinkras. In this section, we will show that
attention must be paid, as the rank sequence related to $C^\bullet (\mathpzc{A})$ is not a complete invariant for
a multiplet $\mathcal{V} (\mathpzc{A})$ coming from $\pzA$\footnote{A priori, these dimensions or ranks are
not even a well-defined invariant in the derived
category $\mbox{\sffamily{D}}^{\flat}(R\mbox{-{\sffamily{Mod}}})$. However we can compute the
\enquote{rank} $n_i$ of the complex $C^{\bullet}(\pzA)$ in degree $i$ for some Adinkra $\pzA$
as the dimension of the $k$-vector space $\operatorname{Tor}_R^i(k, C^{\bullet}(\pzA))$ which is a
well-defined functor from $\mbox{\sffamily{D}}^{\flat}(R\mbox{-{\sffamily{Mod}}}) \rightarrow
\kk{\mbox{-\sffamily{Mod}}}$.}. In other words, this means that the multiplet cannot be \emph{not} uniquely defined by the
number of component fields in all available engineering dimensions within the multiplet. We
illustrate this phenomenon with an $N=6$ example (these Adinkras appeared first in \cite{Doran06counter}).
\begin{example}[Rank Sequence and Multiplets: an $N=6$ example]
Consider two Adinkras $\mathpzc{A}_1$ and $\mathpzc{A}_2$ both characterized by complexes of the form
\bc
\begin{tikzcd}
    0 \arrow[r] & R^6 \arrow[r] & R^8 \arrow[r] & R^2 \arrow[r] & 0.
\end{tikzcd}
\ec
\begin{figure}[ht!]
    \centering
    \includegraphics[width=9cm]{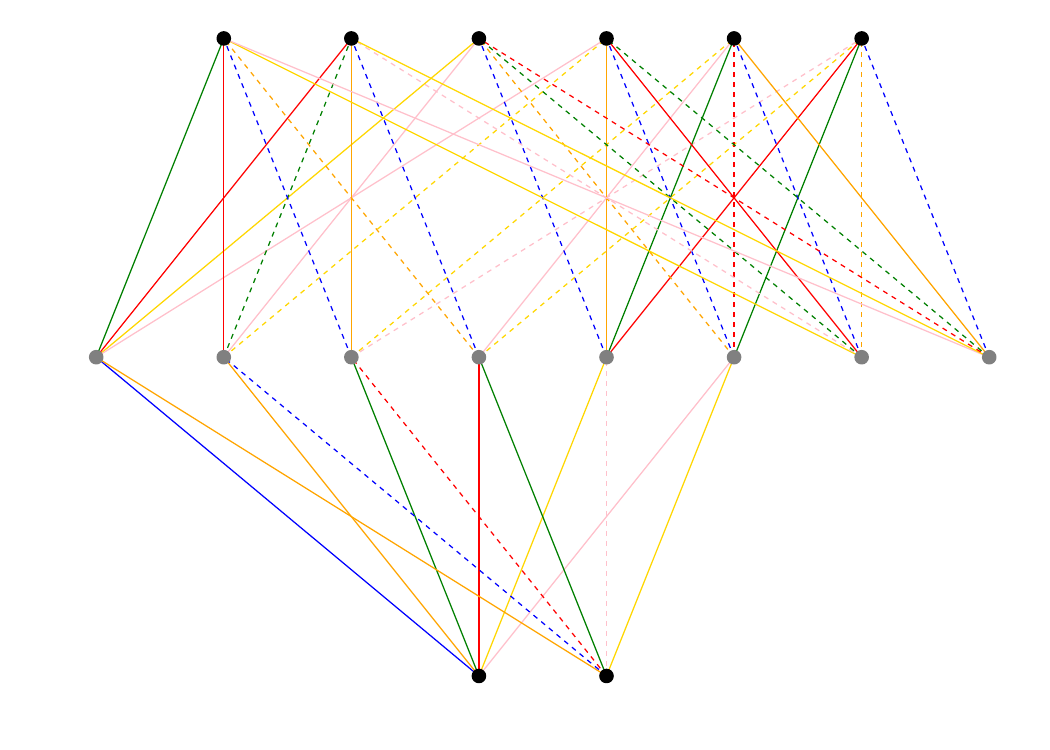}
    \caption{An $N = 6$ Adinkras with rank sequence $(6,8,2)$.}
    \label{n6_6_8_2_S}
\end{figure}

\noindent The first Adinkra $\mathpzc{A}_1$ shown in figure \ref{n6_6_8_2_S} has differentials
\begin{equation}
    \begin{split}
        d^0_{\mathpzc{A}_1} &=
    \bmat
\textcolor{blue}{\lm_1} & \textcolor{orange}{\lm_2} &
        \textcolor{darkgreen}{\lm_3} & \textcolor{red}{\lm_4} & \textcolor{gold}{\lm_5} &
        \textcolor{pink}{\lm_6} & \textcolor{black}{0} & \textcolor{black}{0} \\
\textcolor{orange}{\lm_2} & \textcolor{blue}{-\lm_1} & \textcolor{red}{-\lm_4} &
        \textcolor{darkgreen}{\lm_3} & \textcolor{pink}{-\lm_6} &
 \textcolor{gold}{\lm_5} & \textcolor{black}{0} &
        \textcolor{black}{0} \\
\emat, \\
        d^1_{\mathpzc{A}_1} &= \bmat
\textcolor{darkgreen}{\lm_3} & \textcolor{red}{\lm_4} &
        \textcolor{gold}{\lm_5} & \textcolor{pink}{\lm_6} & \textcolor{black}{0} &
        \textcolor{black}{0} \\
\textcolor{red}{\lm_4} & \textcolor{darkgreen}{-\lm_3} &
        \textcolor{pink}{\lm_6} & \textcolor{gold}{-\lm_5} & \textcolor{black}{0} &
        \textcolor{black}{0} \\
\textcolor{blue}{-\lm_1} & \textcolor{orange}{\lm_2} &
        \textcolor{black}{0} & \textcolor{black}{0} & \textcolor{gold}{-\lm_5} &
        \textcolor{pink}{-\lm_6} \\
\textcolor{orange}{-\lm_2} & \textcolor{blue}{-\lm_1} &
        \textcolor{black}{0} & \textcolor{black}{0} & \textcolor{pink}{\lm_6} &
        \textcolor{gold}{-\lm_5} \\
\textcolor{black}{0} & \textcolor{black}{0} &
        \textcolor{blue}{-\lm_1} & \textcolor{orange}{\lm_2} & \textcolor{darkgreen}{\lm_3} &
        \textcolor{red}{\lm_4} \\
\textcolor{black}{0} & \textcolor{black}{0} &
        \textcolor{orange}{-\lm_2} & \textcolor{blue}{-\lm_1} & \textcolor{red}{-\lm_4} &
        \textcolor{darkgreen}{\lm_3} \\
\textcolor{gold}{\lm_5} & \textcolor{pink}{-\lm_6} &
        \textcolor{darkgreen}{-\lm_3} & \textcolor{red}{\lm_4} & \textcolor{blue}{-\lm_1} &
        \textcolor{orange}{-\lm_2} \\
\textcolor{pink}{\lm_6} & \textcolor{gold}{\lm_5} &
        \textcolor{red}{-\lm_4} & \textcolor{darkgreen}{-\lm_3} & \textcolor{orange}{\lm_2} &
\textcolor{blue}{-\lm_1} \\
\emat. \\
    \end{split}
\end{equation}
The related complex $C^\bullet (\mathpzc{A}_1)$ has $d$-cohomology given by
\begin{align}
    & H^0 (C^{\bullet} (\pzA_1)) = \coker \bmat
        \lm_5 & \lm_6 & \lm_3 & \lm_4 & \lm_1 & -\lm_2 \\
        -\lm_6 & \lm_5 & -\lm_4 &\lm_3 &\lm_2 & \lm_1 \\
    \emat, \\
   & H^1 (C^{\bullet} (\pzA_1))= R/I \oplus R/I, \\
    & H^2 (C^{\bullet} (\pzA_1))= 0.
\end{align}

\begin{figure}[ht!]
    \centering
    \includegraphics[width=9cm]{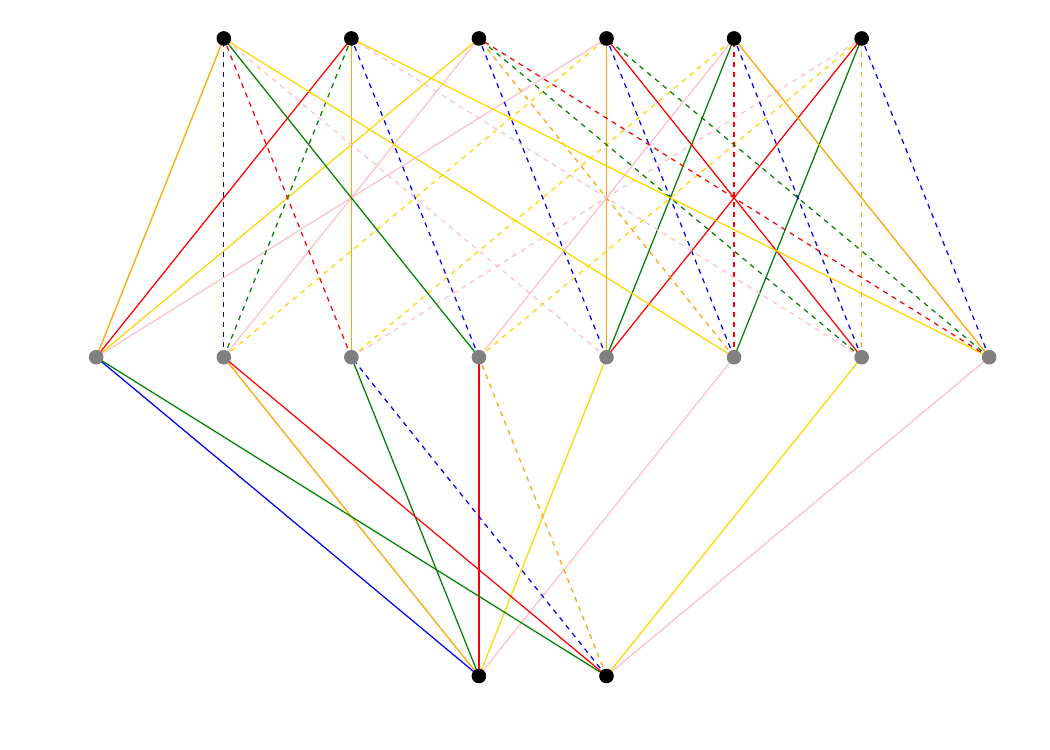}
    \caption{Another $N=6$ Adinkras with rank sequence $(6,8,2)$.}
    \label{n6_6_8_2_T}
\end{figure}

The second Adinkra $\mathpzc{A}_2$ shown in figure \ref{n6_6_8_2_T} has differentials
\begin{equation}
    \begin{split}
        d^0_{\mathpzc{A}_2} &=\bmat
\textcolor{blue}{\lm_1} & \textcolor{orange}{\lm_2} &
        \textcolor{darkgreen}{\lm_3} & \textcolor{red}{\lm_4} & \textcolor{gold}{\lm_5} &
        \textcolor{pink}{\lm_6} & \textcolor{black}{0} & \textcolor{black}{0} \\
\textcolor{darkgreen}{\lm_3} & \textcolor{red}{\lm_4} &
        \textcolor{blue}{-\lm_1} & \textcolor{orange}{-\lm_2} & \textcolor{black}{0} &
        \textcolor{black}{0} & \textcolor{gold}{\lm_5} & \textcolor{pink}{\lm_6} \\
    \emat, \\
    d^1_{\mathpzc{A}_2} &= \bmat
        \textcolor{orange}{\lm_2} & \textcolor{red}{\lm_4} & \textcolor{gold}{\lm_5} &
        \textcolor{pink}{\lm_6} & \textcolor{black}{0} & \textcolor{black}{0} \\
\textcolor{blue}{-\lm_1} & \textcolor{darkgreen}{-\lm_3} &
        \textcolor{pink}{\lm_6} & \textcolor{gold}{-\lm_5} & \textcolor{black}{0} &
        \textcolor{black}{0} \\
\textcolor{red}{-\lm_4} & \textcolor{orange}{\lm_2} &
        \textcolor{black}{0} & \textcolor{black}{0} & \textcolor{gold}{-\lm_5} &
        \textcolor{pink}{-\lm_6} \\
\textcolor{darkgreen}{\lm_3} & \textcolor{blue}{-\lm_1} &
        \textcolor{black}{0} & \textcolor{black}{0} & \textcolor{pink}{\lm_6} &
        \textcolor{gold}{-\lm_5} \\
\textcolor{pink}{-\lm_6} & \textcolor{black}{0} &
        \textcolor{blue}{-\lm_1} & \textcolor{orange}{\lm_2} & \textcolor{darkgreen}{\lm_3} &
        \textcolor{red}{\lm_4} \\
\textcolor{gold}{\lm_5} & \textcolor{black}{0} &
        \textcolor{orange}{-\lm_2} & \textcolor{blue}{-\lm_1} & \textcolor{red}{-\lm_4} &
        \textcolor{darkgreen}{\lm_3} \\
\textcolor{black}{0} & \textcolor{pink}{-\lm_6} & \textcolor{darkgreen}{-\lm_3} &
        \textcolor{red}{\lm_4} & \textcolor{blue}{-\lm_1} & \textcolor{orange}{-\lm_2} \\
\textcolor{black}{0} & \textcolor{gold}{\lm_5} & \textcolor{red}{-\lm_4} &
        \textcolor{darkgreen}{-\lm_3} & \textcolor{orange}{\lm_2} & \textcolor{blue}{-\lm_1}. \\
\emat
\end{split}
\end{equation}
This results in the following cohomology groups
\begin{equation}
    \begin{split}
        H^0 (C^{\bullet} (\pzA_2) )&= \coker \bmat
            \lm_6 &\lm_5 & 0 & 0 & \lm_2 & \lm_1 & -\lm_4 & -\lm_3 \\
            0 & 0 & \lm_6 & \lm_5 & \lm_4 & \lm_3 & \lm_2 & \lm_1 \\
        \emat, \\
        H^1(C^{\bullet} (\pzA_2)) &= \coker \bmat
            \lm_2 & \lm_1 & -\lm_4 & -\lm_3 & \lm_6 & \lm_5 & 0 & 0
            \\
            \lm_4 & \lm_3 & \lm_2 & \lm_1 & 0 & 0 & -\lm_6 & \lm_5\\
            -\lm_5 & \lm_6 & 0 & 0 & \lm_1 & \lm_2 & -\lm_3 & \lm_4\\
            \lm_6 & -\lm_5 & 0 & 0 & -\lm_2 & -\lm_1 & \lm_4 & \lm_3
            \\
            0 & 0 &-\lm_5 & \lm_6 & \lm_3 & -\lm_4 & \lm_1 & \lm_2 \\
            0 & 0 & -\lm_6 & -\lm_5 & -\lm_4 & -\lm_3 & -\lm_2 &
            \lm_1 \\
        \emat, \\
        H^2(C^{\bullet} (\pzA_2)) &= 0.
    \end{split}
\end{equation}
It is apparent that the zeroth and the first cohomology groups of these complexes are not isomorphic -- this can be most easily seen from the
dimensions of their respective minimal free resolutions. To our knowledge, these (derived or cohomological) invariants associated
to Adinkraic multiplets $\mathcal{V} (\mathpzc{A})$ -- the cohomology modules $H^{-n} ({C}^\bullet (\mathpzc{A}))$ for $n\geq0$ -- have not yet been considered:
it would be interesting to give them a physical interpretation, beyond their obvious bookkeeping service to distinguish multiplets in a
homological fashion.
\end{example}

\subsection{The $N=7$ multiplet of rank \texorpdfstring{$(1,7,7,1)$}{(1,7,7,1)}}
\label{1_7_7_1}

In this section, we will study realizations of Adinkra multiplets $\mathcal{V}(\pzA)$ as constrained superfields.
From the point of view of the associated complexes $C^\bullet (\mathpzc{A})$, this amounts to say that there
exists an embedding of $C^\bullet (\mathpzc{A})$ inside the Koszul complex $C^\bullet (\mathpzc{A}_K)$, i.e.
an injective morphism of complexes $C^\bullet (\mathpzc{A}) \hookrightarrow C^\bullet (\mathpzc{A}_K)$.
As usual, we will start with a relevant example in $N=8$, which is of interest on its own.

\begin{figure}[ht!]
    \centering
    \includegraphics[height=5cm]{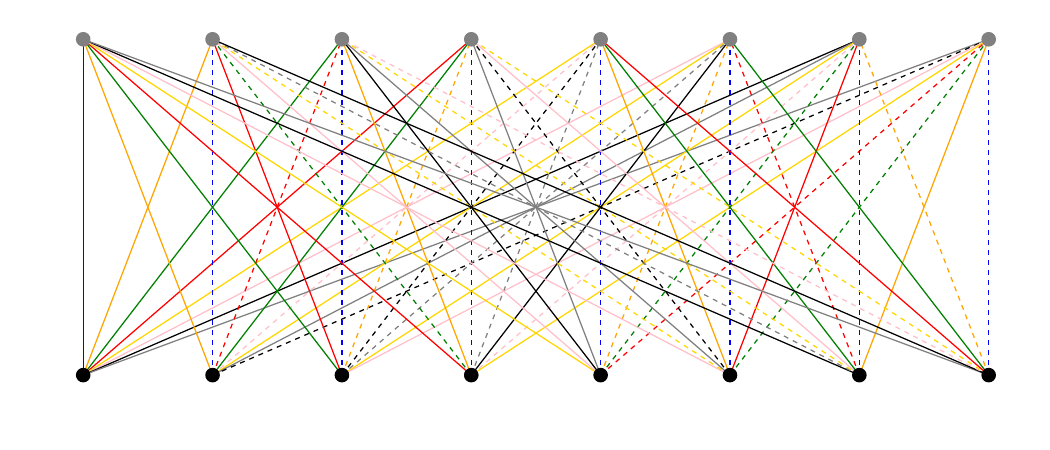}
    \caption{The valise $(8,8)$ Adinkra of $N = 8$.}
    \label{n8_8_8}
\end{figure}

\begin{example}[The rank (1,7,7,1) multiplet.] The irreducible $N= 8$ valise shown in figure \ref{n8_8_8} defines a complex
\be
    \begin{tikzcd}
      0 \arrow[r] &  R^8 \arrow[r, "d"] & R^8 \arrow[r] & 0,
    \end{tikzcd}
\ee
with differential
\be
d =  \bmat \textcolor{blue}{\lm_1} & \textcolor{orange}{\lm_2} & \textcolor{darkgreen}{\lm_3} &
\textcolor{red}{\lm_4} & \textcolor{gold}{\lm_5} & \textcolor{pink}{\lm_6} &
\textcolor{black}{\lm_7} & \textcolor{gray}{\lm_8} \\
\textcolor{orange}{\lm_2} & \textcolor{blue}{-\lm_1} & \textcolor{red}{-\lm_4} &
\textcolor{darkgreen}{\lm_3} & \textcolor{pink}{-\lm_6} & \textcolor{gold}{\lm_5} &
\textcolor{gray}{\lm_8} & \textcolor{black}{-\lm_7} \\
\textcolor{darkgreen}{\lm_3} & \textcolor{red}{\lm_4} & \textcolor{blue}{-\lm_1} &
\textcolor{orange}{-\lm_2} & \textcolor{black}{-\lm_7} &
\textcolor{gray}{-\lm_8} & \textcolor{gold}{\lm_5} & \textcolor{pink}{\lm_6} \\
\textcolor{red}{\lm_4} & \textcolor{darkgreen}{-\lm_3} & \textcolor{orange}{\lm_2} &
\textcolor{blue}{-\lm_1} & \textcolor{gray}{-\lm_8} & \textcolor{black}{\lm_7} &
\textcolor{pink}{-\lm_6} & \textcolor{gold}{\lm_5} \\
\textcolor{gold}{\lm_5} & \textcolor{pink}{\lm_6} & \textcolor{black}{\lm_7} &
\textcolor{gray}{\lm_8} & \textcolor{blue}{-\lm_1} & \textcolor{orange}{-\lm_2} &
\textcolor{darkgreen}{-\lm_3} & \textcolor{red}{-\lm_4} \\
\textcolor{pink}{\lm_6} & \textcolor{gold}{-\lm_5} & \textcolor{gray}{\lm_8} &
\textcolor{black}{-\lm_7} & \textcolor{orange}{\lm_2} & \textcolor{blue}{-\lm_1} &
\textcolor{red}{\lm_4} & \textcolor{darkgreen}{-\lm_3} \\
\textcolor{black}{\lm_7} & \textcolor{gray}{-\lm_8} & \textcolor{gold}{-\lm_5} &
\textcolor{pink}{\lm_6} & \textcolor{darkgreen}{\lm_3} & \textcolor{red}{-\lm_4} &
\textcolor{blue}{-\lm_1} & \textcolor{orange}{\lm_2} \\
\textcolor{gray}{\lm_8} & \textcolor{black}{\lm_7} & \textcolor{pink}{-\lm_6} &
\textcolor{gold}{-\lm_5} & \textcolor{red}{\lm_4} & \textcolor{darkgreen}{\lm_3} &
\textcolor{orange}{-\lm_2} & \textcolor{blue}{-\lm_1} \\
\emat
\ee

\noindent The $N= 7$ irreducible valise Adinkra can be obtained from the above $N = 8$ valise by forgetting
the gray lines in the picture or equivalently setting $\lambda_8 = 0$. Via
    consecutive vertex raising of seven degree 0 vertices followed by the
    raising of the only possible degree 1 vertex,
\be
\begin{tikzcd}
(8,8)  \arrow[r, rightsquigarrow, "\mathpzc{Raise}"] & (1,8,7)  \arrow[r,
    rightsquigarrow, "\mathpzc{Raise}"] & (1,7,7,1),
\end{tikzcd}
\ee
we obtain the $N=7$ Adinkra characterized by the rank sequence $(1,7,7,1)$, we call it $\mathpzc{A}_{1771}, $ as shown in figure \ref{n7_1_7_7_1}.
\begin{figure}[t]
    \centering
    \includegraphics[width=11cm]{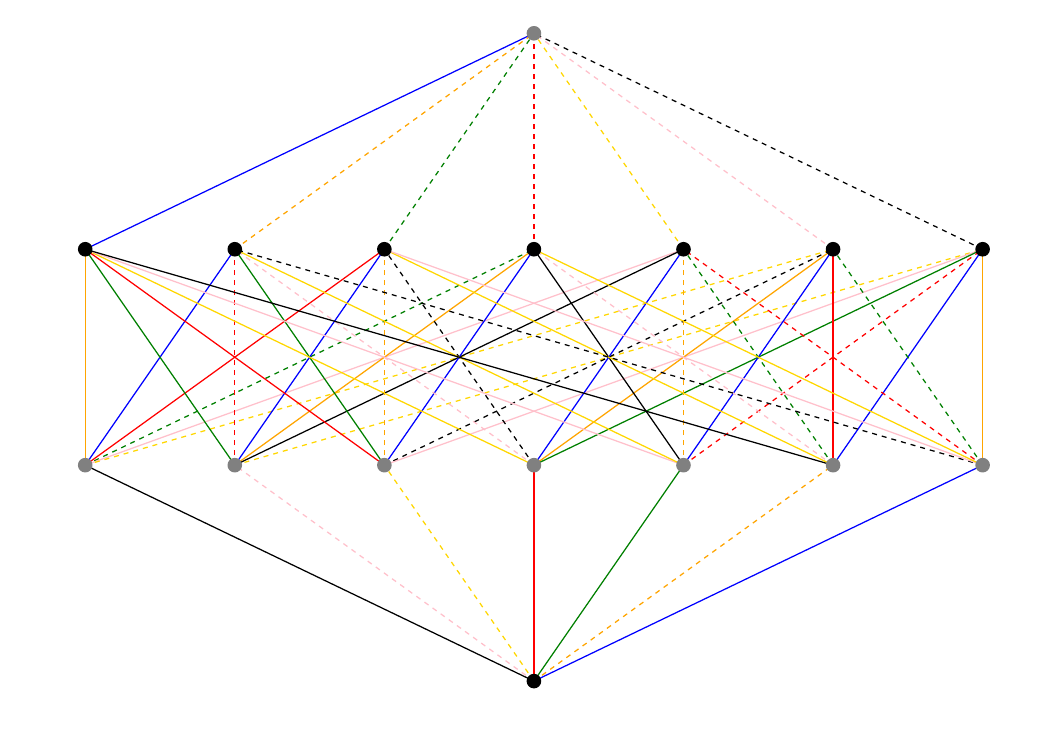}
    \caption{The $N=7$ Adinkra with rank sequence $(1,7,7,1)$.}
    \label{n7_1_7_7_1}
\end{figure}
The associated complex reads
\bc
\begin{tikzcd}
   C^\bullet (\mathpzc{A}_{1771}) \equiv \big ( 0 \arrow[r] & R^1 \arrow[r, "d_2"] & R^7 \arrow[r, "d_1"] & R^7 \arrow[r,
    "d_0"] & R^1 \arrow[r] & 0 \big )
\end{tikzcd}
\ec
with differentials given by 
\begin{equation}
    \begin{split}
    d_2 &= \bmat
        \textcolor{blue}{\lm_1} \\
        \textcolor{orange}{-\lm_2} \\
        \textcolor{darkgreen}{-\lm_3} \\
        \textcolor{red}{-\lm_4} \\
        \textcolor{gold}{-\lm_5} \\
        \textcolor{pink}{-\lm_6} \\
        \textcolor{black}{-\lm_7} \\
    \emat, \\
    d_1 &= \bmat
        \textcolor{orange}{\lm_2} & \textcolor{blue}{\lm_1} & \textcolor{red}{\lm_4} & \textcolor{darkgreen}{-\lm_3} & \textcolor{pink}{\lm_6} & \textcolor{gold}{-\lm_5} & \textcolor{black}{0} \\
        \textcolor{darkgreen}{\lm_3} & \textcolor{red}{-\lm_4} & \textcolor{blue}{\lm_1} & \textcolor{orange}{\lm_2} & \textcolor{black}{\lm_7} & \textcolor{black}{0} & \textcolor{gold}{-\lm_5} \\
        \textcolor{red}{\lm_4} & \textcolor{darkgreen}{\lm_3} & \textcolor{orange}{-\lm_2} & \textcolor{blue}{\lm_1} & \textcolor{black}{0} & \textcolor{black}{-\lm_7} & \textcolor{pink}{\lm_6} \\
        \textcolor{gold}{\lm_5} & \textcolor{pink}{-\lm_6} & \textcolor{black}{-\lm_7} & \textcolor{black}{0} & \textcolor{blue}{\lm_1} & \textcolor{orange}{\lm_2} & \textcolor{darkgreen}{\lm_3} \\
        \textcolor{pink}{\lm_6} & \textcolor{gold}{\lm_5} & \textcolor{black}{0} & \textcolor{black}{\lm_7} & \textcolor{orange}{-\lm_2} & \textcolor{blue}{\lm_1} & \textcolor{red}{-\lm_4} \\
        \textcolor{black}{\lm_7} & \textcolor{black}{0} & \textcolor{gold}{\lm_5} & \textcolor{pink}{-\lm_6} & \textcolor{darkgreen}{-\lm_3} & \textcolor{red}{\lm_4} & \textcolor{blue}{\lm_1} \\
        \textcolor{black}{0} & \textcolor{black}{-\lm_7} & \textcolor{pink}{\lm_6} & \textcolor{gold}{\lm_5} & \textcolor{red}{-\lm_4} & \textcolor{darkgreen}{-\lm_3} & \textcolor{orange}{\lm_2} \\
    \emat, \\
    d_0 &= \bmat
        \textcolor{black}{\lm_7} & \textcolor{pink}{-\lm_6} & \textcolor{gold}{-\lm_5} & \textcolor{red}{\lm_4} & \textcolor{darkgreen}{\lm_3} & \textcolor{orange}{-\lm_2} & \textcolor{blue}{\lm_1} \\
    \emat. \\
    \end{split}
.\end{equation}

Remarkably, the complex $C^\bullet (\mathpzc{A}_{1771})$ has non-trivial cohomology. In particular, $H^{-1} (C^{\bullet} (\mathpzc{A}_{1771}) )$ is a non-trivial module, with linear free resolution given by
\be
    \label{14_34_res}
    \begin{tikzcd}
    \mathcal{R}^\bullet_{H^{-1}} \equiv \big (    0 \arrow[r] & R^1 \arrow[r] & R^7 \arrow[r] & R^{21} \arrow[r] & R^{35}
        \arrow[r] & R^{34} \arrow[r] & R^{14} \arrow[r] & 0 \big ).
    \end{tikzcd}
\ee
Although this complex cannot come from an Adinkra -- indeed the ranks of the free modules do not add up to a power of two --, one can still present it as a colored,
dashed graph, violating only $N$-regularity, i.e.\ a vertex can be attached to
more than $N$ edges. This graphical representation of the minimal free resolution is shown in figure \ref{n7_14_34}.

\begin{figure}[ht!]
    \centering
    \includegraphics[width=8.5cm]{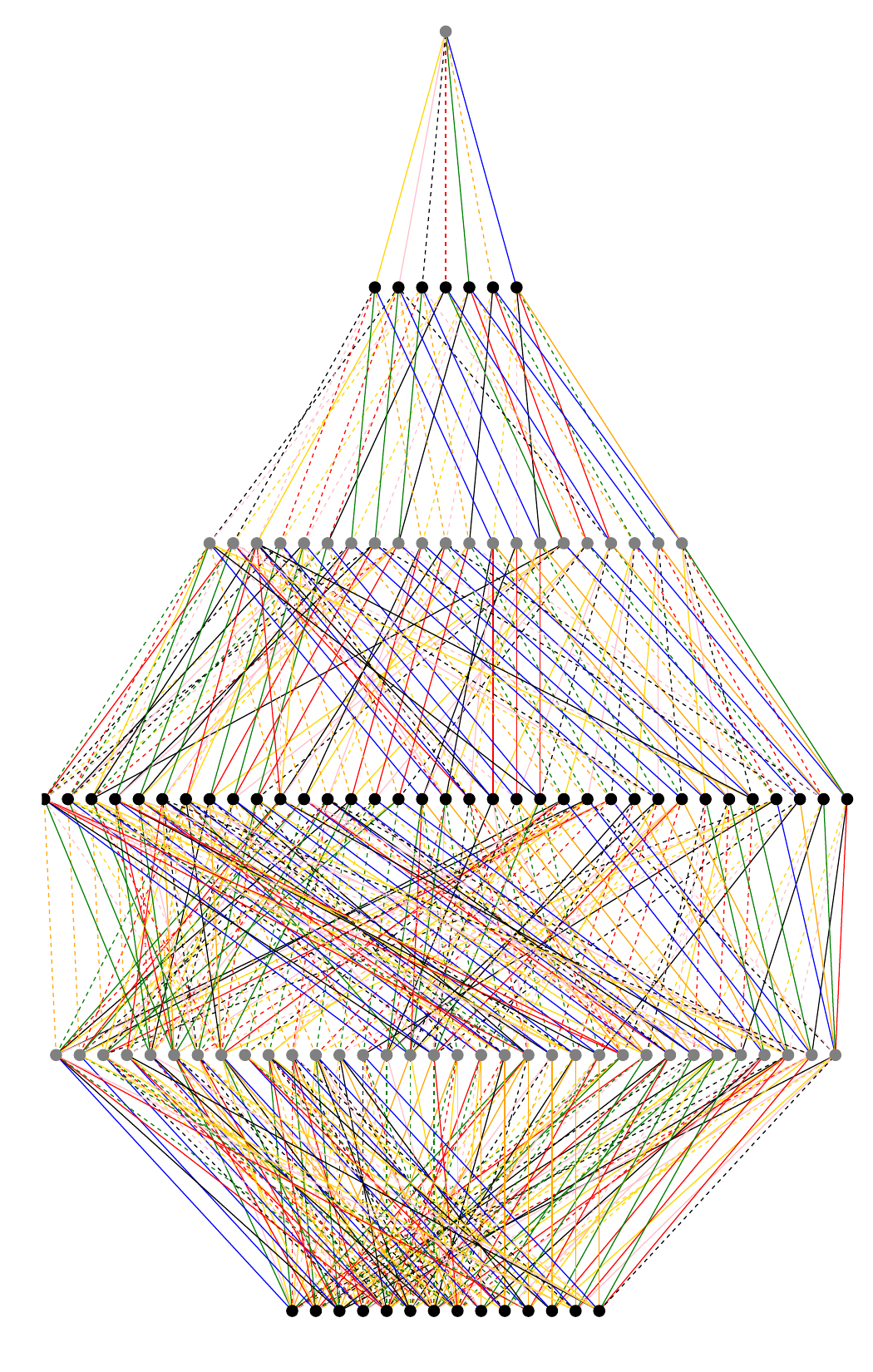}
    \caption{Graphical representation of the linear resolution of $H^{-1} (C^\bullet (\mathpzc{A}_{1771}))$
    }
    \label{n7_14_34}
\end{figure}

On the other hand, consider the $N=7$ Koszul Adinkra $\mathpzc{A}_K$ whose related complex is the $N=7$
Koszul complex $C^\bullet (\mathpzc{A}_K) $
\bc
    \label{n7_koszul_res}
    \begin{tikzcd}
        0 \arrow[r] & R^1 \arrow[r, "d'_7"] & R^7 \arrow[r, "d'_6"] & R^{21}
        \arrow[r, "d'_5"] & R^{35} \arrow[r, "d'_4"] & R^{35} \arrow[r, "d'_3"]
        & R^{21} \arrow[r, "d'_2"] & R^7 \arrow[r, "d'_1"] & R^1 \arrow[r] & 0
    \end{tikzcd}
\ec
whose rank sequence is given by $(1,7,21,35,35,21,7,1)$. It could be noted that the rank sequences of the Adinkra $\mathpzc{A}_{1771}$
and of the resolution $\mathcal{R}^\bullet_{H^{-1}}$ of $H^{-1} (C^{\bullet} (\mathpzc{A}_{1771}))$ add up exactly to that of the $N=7$ Koszul complex $C^\bullet (\mathpzc{A}_K)$:
\be
(0, 0, 0, 0, 1,7,7,1) + (1, 7, 21, 35, 34, 14, 0, 0 ) = (1,7,21,35,35,21,7,1).
\ee
This suggests that there exists a distinguished triangle in $\mbox{{\sffamily{D}}}^\flat (R/I)$ given by
\be \label{triangKos}
C^\bullet (\mathpzc{A}_{1771}) \longrightarrow C^\bullet (\mathpzc{A}_K) \longrightarrow \mathcal{R}^\bullet_{H^{-1}}
\ee
that realizes the complex $C^\bullet (\mathpzc{A}_{1771})$ as a subcomplex of the Koszul complex $C^\bullet (\mathpzc{A}_K)$. In other words,
we have an embedding of complexes $\iota : C^\bullet (\mathpzc{A}_{1771}) \hookrightarrow C^\bullet (\mathpzc{A}_K)$, realized by a commutative
diagram of the form\footnote{It is not too hard to explicitly realize the above embedding by choosing a basis
for $\kk^7$ in the Koszul complex $C^\bullet (\mathpzc{A}_K) = R \otimes_\kk \bigwedge^\bullet \kk^7$.}
\bc
\begin{tikzcd}
    &[-0.2em]   &[-0.2em]  &[-0.2em]  &[-0.2em]  0 \arrow[d] \arrow[r] &[-0.2em] R^1 \arrow[r,"d_3"] \arrow[d, "i_3"] &
    R^7 \arrow[r, "d_2"] \arrow[d, "i_2"] &[-0.2em]R^7 \arrow[r, "d_1"] \arrow[d,
    "i_1"] &[-0.2em] R^1 \arrow[r] \arrow[d, "i_0"] &[-0.2em] 0 \arrow[d] \\
    0 \arrow[r] & R^1 \arrow[r, "d'_7"] & R^7 \arrow[r, "d'_6"] & R^{21}
    \arrow[r, "d'_5"] & R^{35} \arrow[r, "d'_4"] & R^{35} \arrow[r, "d'_3"] &
    R^{21} \arrow[r, "d'_2"] & R^7 \arrow[r, "d'_1"] & R^1 \arrow[r] & 0
\end{tikzcd}
\ec

Taking the cone of the embedding $\operatorname{Cn}^\bullet (i)$ and looking at the long exact homology sequence one finds
\be
\begin{tikzcd}
        & 0 \arrow[r] \arrow[d, phantom, ""{coordinate, name=Z}]& H^{-2}(\operatorname{Cn}^\bullet (i))
\arrow[dll, rounded corners, to path={ -- ([xshift=2ex]\tikztostart.east) |- (Z) [near
end]\tikztonodes -| ([xshift=-2ex]\tikztotarget.west) -- (\tikztotarget)}] & \phantom{0} \\
        H^{-1}(C^\bullet (\mathpzc{A}_{1771})) \arrow[r] & H^{-1}(C^\bullet (\mathpzc{A}_K)) \cong 0 \arrow[r]
        \arrow[d, phantom, ""{coordinate, name=Z}]&
        H^{-1}(\operatorname{Cn}^\bullet (i))
        \arrow[dll, rounded corners, to path={ -- ([xshift=2ex]\tikztostart.east)
|- (Z) [near end]\tikztonodes
-| ([xshift=-2ex]\tikztotarget.west)
-- (\tikztotarget)}]&  \phantom{0} \\
        H^0 (C^\bullet (\mathpzc{A}_{1771})) \cong \kk \arrow[r, "id"] & H^0 (C^\bullet (\mathpzc{A}_K)) \cong \kk \arrow[r] & H^0
        (\operatorname{Cn}^\bullet(i)) \arrow[r] & 0. \\
\end{tikzcd}
\ee
That implies that $\operatorname{Cn}^\bullet (i)$ has non-trivial cohomology only in degree $-2$, and as such it is
quasi-isomorphic to a free resolution of $H^{-1}(C^\bullet (\mathpzc{A}_{1771}))$, that is we have a quasi-isomorphism
$\mathcal{R}^\bullet_{H^{-1}} \cong \operatorname{Cn}^\bullet (i)$, as suggested by equation \eqref{triangKos}.

As seen above, the complex $\mathcal{R}^\bullet_{H^{-1}}$ cannot come from an Adinkra. However, it is Koszul-dual to a cokernel of
$U_\kk (\ft)$-modules, thus providing an example of $N=7$ multiplet which does
not come from an Adinkra. The moduli space of exact triangles in
\eqref{triangKos} coincides with the moduli space of Cayley bundles which is $SO
(7) / G_2$. We
will see more examples of this phenomenology later in this paper.
\end{example}



\subsection{Embeddings in the free superfield and codes}
\label{sec:embeddings}
We now aim to generalize the previous example to a whole class of Adinrkas. Namely, we consider Adinkras
which have only a single vertex such that all attached edges go up (the rank (1,7,7,1) Adinkra considered above is an example).
In the terminology of \cite{Doran08}, the Adinkra \enquote{hangs upside down on
a single vertex}. We call such an Adinkra fully extended. In this case, it is easy to see that
the associated complex $C^\bullet (\mathpzc{A}) $ is such that $H^0 ({C}^\bullet (\mathpzc{A})) \cong \kk,$ and adapting the
proof of theorem 7.7 in \cite{Eisenbud80}, we prove that such complexes can
be embedded in a Koszul complex $C^\bullet (\pzA_K)$.
On our way to prove this result, we will show that the
degree $0$ part of the kernel of the differential of the complex $C^\bullet (\pzA)$
is generated by elements associated to vertices $v$ with $dv = 0$.
This result is interesting on its own since it says that the number of zero modes
is an invariant of the topology of an Adinkra.
\begin{lemma}
    \label{zero_modes_lemma}
    Let $\pzA$ be an $N$-Adinkra. Let $w \in C^{-i}(\pzA)$ of degree
    $0$ in the $\lm$-grading induced by $R = \kk[\lambda_1, \ldots, \lambda_N]$ on $C^{-i} (\pzA)$. Then
  \be
  w \in \mbox{span}_\kk \left \{ v \in C^{-i} (\pzA) : d v = 0 \right \}.
  \ee
\end{lemma}
\begin{proof}
    Let $A$ be the matrix representing the differential $d : C^{-i}(\pzA)
    \rightarrow C^{-i + 1}(\pzA)$ in the basis defined by the
    vertices of $\pzA$. Then, for $w = \sum_{k} c_k v^k$, where the $v^k$ are the distinguished basis of $C^{-i} (\pzA)$,
    we can write the condition $d w = 0$ as
    \begin{equation}
        (dw)_j = \sum_{k} A_{jk} c_k = 0,
    \end{equation}
    where the index $j$ refers to
    the projection to the distinguished basis of $C^{-i + 1}(\pzA)$. Now, since
    there is at most one edge of each color with target $\pzs$ the vertex
    indexed by $j$, hence for fixed
    $k$, the coefficient $A_{jk}$ has each $\lm_l$ at most once as the
    coefficient or it is zero. This implies that we need to have $A_{jk} c_k = 0$ for all
    $j,k$. Now, if $c_k \neq 0$ we need that $A_{jk} = 0$ for all $j$ and hence
    $d v_k = 0$.
\end{proof}
With lemma \ref{zero_modes_lemma} in our inventory, we can now prove the
following embedding result.

\begin{proposition}
    \label{embedding_lemma}
    Let $\pzA$ be an $N$-Adinkra as above, with a unique vertex $v$ with minimal height such that the associated
    element in $C^{\bullet}(\pzA)$ satisfies $d v = 0.$ \\Then there is an embedding
    $C^{\bullet}(\pzA) \hookrightarrow C^{\bullet}(\pzA_K)$ where $\pzA_K$ is
    the Koszul $N$-Adinkra.
    In particular the multiplet $\mathcal{V}(\mathpzc{A})$ can be written as a constrained superfield.
\end{proposition}
\begin{proof}
Without loss of generality, we assume the minimal height is $0$. Note that for
vertices in $\pzA$ with minimal height the associated element in $C^{\bullet}(\pzA)$
always satisfies $d v = 0$. Moreover, there must be exactly $N$ vertices in $\pzV (\pzA)$ with
height equal to $1$, so that the tail of the complex $C^{\bullet}(\pzA)$ reads
\be
\label{embedding_0_1}
        \begin{tikzcd}[ampersand replacement=\&]
            \ldots \arrow[r] \&C^{-1} (\mathpzc{A}) = R^N \arrow[rrr, "{\small \bmat \pm \lm_1  &\ldots & \pm \lm_N
            \small \emat}"] \& \&  \& C^0 (\mathpzc{A}) = R \arrow[r] \& 0.
        \end{tikzcd}
\ee
It follows that $H^{0}(C^{\bullet}(\pzA)) = \kk$. One thus have to prove that $C^{\bullet}(\pzA)$
embeds into the free resolution $C^{\bullet}(\pzA_K)$ of $H^{0}(C^{\bullet}(\pzA)) \cong \kk$.

We write the free modules $C^{-i}(\pzA) \subset C^{\bullet} (\mathpzc{A})$ as $C^{-i} (\pzA) = R \otimes_{\kk} P^{-i}$, where $P^{-i}$ is
the vector space generated by the elements $v'$ with $\mathpzc{h}(v') = i$. The $\lm$-grading on
$R = \kk[\lm_1, \ldots, \lm_N]$ induces a $\mathbb{Z}$-grading on $C^{-i}(\pzA)$ with $C^{-i}(\pzA)_0 = P^{-i}$.

Now, for $i > 0$ by the assumption that there is no $v'$ with $\mathpzc{h}(v') > 0$ such that $dv' =
0$, it follows that $(\ker d^{-i})_0 = 0$, \emph{i.e.}\ the degree zero piece of the kernel of $d^{-i}$ is
zero. In particular, the map $\begin{tikzcd} P^{i} \arrow[r, "d^i"] & \ker d^{i - 1} \end{tikzcd}$
is a monomorphism.

From \eqref{embedding_0_1} we see that we can embed degrees $0$ and $-1$ into $C^{\bullet}(\pzA_K)$
by the identity map. Working by induction, assume now that we can embed up to degree $i-1$
    \be
        \begin{tikzcd}
            R \otimes_{\kk} P^{i - 1} \arrow[r, "d^{i-1}"] & R \otimes_{\kk}
            P^{i-2} \arrow[r] & \ldots
        \end{tikzcd}
    \ee
as a subcomplex into $C^{\bullet}(\pzA_K)$. Consider the diagram
    \be
        \begin{tikzcd}
            &  C^{i}(\pzA_K) \phantom{\subset C^{i- 1}(\pzA_K)} \arrow[d,
            "d^i_{\pzA_K}", two heads, xshift=-6ex] \\
            R \otimes_{\kk} P^{i} \arrow[r, "d"] &
            \ker d^{i - 1}_{\pzA_K} \subset C^{i - 1}(\pzA_K),
        \end{tikzcd}
            \ee
where surjectivity follows from the fact that the Koszul complex has no higher-cohomology. Since $R \otimes_{\kk} P^i$ is projective, we can lift this map. Note that the horizontal map
restricted to $1 \otimes P^{i}$ is injective and thus the restriction of the lift is injective as
well. The lift is of degree zero and maps $P^i$ injectively into the $\kk$-vector space generating the free
    $R$-module $C^{i}(\pzA_{K})$. This implies that the full
map $R \otimes_{\kk} P^i \rightarrow C^{i}(\pzA_K)$ is injective as well, concluding the
verification.
\end{proof}

\begin{table}[h!]
\begin{center}
\begin{tabular}{|c|c|}
\hline
${N}$ & Rank Sequence  \\
\hline
1 & $\varnothing$ \\
\hline
2 & $\varnothing$ \\
\hline
3 & $(1,3,3,1)$ \\
\hline
4 & $\varnothing$ \\
\hline
5 & $(1,5,7,3), (3,7,5,1), (1,6,7,2), (2,7,6,1), (2,6,6,2), (1,7,7,1)$ \\
\hline
6 & $(1,6,7,2), (2,7,6,1), (2,6,6,2), (1,7,7,1)$ \\
\hline
7 & $(1,7,7,1)$ \\
\hline
8 &  $\varnothing$ \\
\hline
\end{tabular}
\end{center}
\caption{Rank sequences of length four multiplets for $N \le 8$}
\label{tab:length4}
\end{table}%

The previous proposition \ref{embedding_lemma} shows that the multiplets with a
unique zero mode (i.e. fully extended Adinkras) in correspond to subcomplexes of the Koszul complex $C^{\bullet}(\pzA')$ -- note nonetheless
that the embedding will never map all vertices of $\pzA$ to $\pzA_K$, since by $N$-regularity, a
connected Adinkra has no non-trivial sub-Adinkra.
A particular interest lies in multiplets whose rank sequence has length 4 (which is the longest possible length for $N \le
8$, when the associated Clifford module is irreducible). The possible length-four sequences were
determined in \cite{Kuznetsova:2005cd} and are listed in table \ref{tab:length4}.
More generally, there is a one-to-one correspondence
between Adinkra topologies and doubly-even codes \cite{Doran08, Doran11}. A doubly-even
code is a subvector space of $(\mathbb{F}_2)^N$
generated by elements such that the sum of the entries is divisible by $4$. The
vector space $(\mathbb{F}_2)^N$ admits a grading by the Hamming weight of a vector,
i.e. the sum of the coordinates.
Since the subspace generated by a code is homogenous, the quotient admits a
grading, which gives the fully extended height assignment for the given
topology. In table \ref{tab:rank_seq} we give the examples for some
interesting maximal codes.
The fully extended rank sequence agree with some of the extremal ranks computed in
\cite{Kuznetsova:2005cd}.  The collection of all possibly rank sequences may have close connections to Boij--S\"oderberg theory \cite{Boij08} in the case of modules over the corresponding quadric.

It is known from \cite{Doran08, Doran11} that for a reducible code, i.e. a code
that can be written as a direct product, the Adinkra is given by the product of
the Adinkras corresponding to the factors. In fact, it is a graded product,
i.e. if we assign to a vertex in the product the sum of the heights of the
factors, we obtain the fully extended height assignment, in other words, the
Hilbert series is the Hilbert series of the smaller building blocks.

We leave a detailed analysis for future works. Open
questions include finding concrete embeddings to the Koszul complex, the
$R$-symmetry preserved by these sub-complexes, and a concrete description of all the possible height
assignments for a given topology.  Spinor groups naturally appear in coding theory \cite{Wood89} and might shed light on the relationship to $R$-symmetry breaking.
\begin{table}
    \begin{center}
        \begin{tabular}{|c|c|c|}
            \hline
            N & Code & Rank sequence \\
            \hline
            $4$ & $d_4$ & $(1,4,3)$ \\
            $5$ &  $t_{1} \oplus d_4$ & $(1,5,7,3)$ \\
            $6$ & $d_6$ & $(1,6,7,2)$ \\
            $7$ & $e_7$ & $(1,7,7,1)$ \\
            $8$ & $e_8$ & $(1,8,7)$ \\
            $9$ & $t_1 \oplus e_8$ & $(1,9,15,7)$ \\
            $10$ & $d_{10}$ & $(1,10, 21,20,10,2)$ \\
            $10$ & $t_2 \oplus e_8$ & $(1,10,24,22,7)$ \\
            $16$ & $e_{16}$ & $(1, 16, 57,112, 70)$ \\
            $16$ & $e_{8} \oplus e_8$ & $(1, 16, 78, 112, 49)$ \\
            \hline
        \end{tabular}
    \end{center}
    \caption{Rank sequences from some doubly-even codes.}
    \label{tab:rank_seq}
\end{table}

\subsection{Extension classes of Adinkras} \label{ExtAdi}

In \cite{Hubsch13} and \cite{Doran13} there appear the graphs in figures
\ref{n4_ext1} and \ref{n3_ext2},
\begin{figure}
    \centering
    \includegraphics[width=11cm]{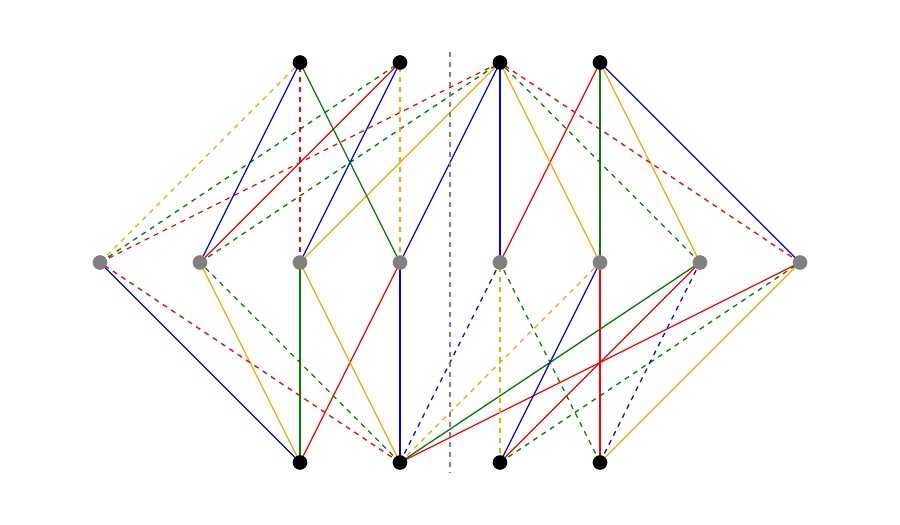}
    \caption{Extension of the $N=4$ Adinkra with ranks $(2,4,2)$ by another copy of itself.}
    \label{n4_ext1}
\end{figure}
\begin{figure}
    \centering
    \includegraphics[width=11cm]{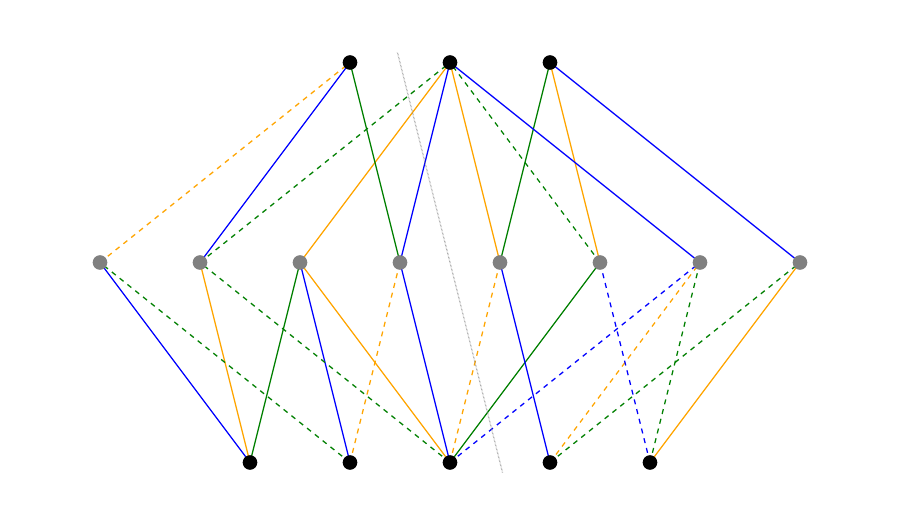}
    \caption{Extension of the $(2,4,2)$ Adinkra by the $(3,4,1)$ Adinkra of $N = 3$.}
    \label{n3_ext2}
\end{figure}
which are similar to Adinkras, except that they do not satisfy $N$-regularity, similar to the case of figure \ref{n7_14_34}.
In this section, we will show that these graphs do define complexes in the derived category of the related quadric
and thus are nicely captured in the pure spinor superfield formalism.

Consider the graph in figure \ref{n4_ext1}. Removing all the lines crossing the
dashed gray line in the middle, leaves us with $2$ copies of the $N=4$ Adinkra with rank sequence $(2,4,2)$.
This suggests that the complex associated to this graph can be understood as an extension of $R / I$
modules, i.e.\ the lines crossing the dashed gray line correspond to an element
in $Ext^1(C^{\bullet}(\pzA_2), C^{\bullet}(\pzA_1))$, where $\pzA_1$ and
$\pzA_2$ are the two copies of $N=4$ Adinkra with ranks $(2,4,2)$ inside the graph in figure
\ref{n4_ext1}. \\
Note that since the differential is defined to go down in the
Adinkra, the complex for the Adinkra $\pzA_1$ on the left is actually the target
of the morphism, i.e.\ we are considering $\mbox{Hom}( C^\bullet (\mathpzc{A}_2), C^\bullet (\mathpzc{A}_1))$. This gives rise to a diagram
\be
\begin{tikzcd}[ampersand replacement=\&, row sep=3em]
    0 \arrow[r] \& R^2 \arrow[from=d] \arrow[r, "{\tiny {\bmat
        -\lm_2 & -\lm_3 \\ \lm_1 & \lm_4 \\ -\lm_4 & \lm_1 \\ \lm_3 & -\lm_2
    \emat}}"] \&[4em] R^4
    \arrow[from=d, "v_{-2}"] \arrow[r, "{\tiny {\bmat \lm_1 & \lm_2 & \lm_3 & \lm_4 \\ -\lm_4 & -\lm_3 & \lm_2 &
    \lm_1 \emat}}"] \&[6em] R^2 \arrow[from=d, "v_{-1}"] \arrow[r]
    \&[6em] 0 \arrow[from=d] \& \\
    \& 0 \arrow[r] \& R^2 \arrow[r, "
        {\tiny {\bmat \lm_1 & \lm_4 \\ \lm_2 & \lm
        _3 \\ -\lm_3 & \lm_2 \\ -\lm_4 & \lm_1 \emat}
        }"] \& R^4
    \arrow[r, "
        {\tiny {\bmat -\lm_2 & \lm_1 & \lm_4 & -\lm_3 \\ -\lm_3 & \lm_4 & -\lm_1 & \lm_2 \emat}
    }"]
        \& R^2 \arrow[r] \&
    0. \&
\end{tikzcd}
\ee
Now giving an extension class $[v] \in Ext^1 (C^\bullet (\pzA_2), C^\bullet (\pzA_1))$ amount to finding the maps $v_{-1}$ and $v_{-0}$ that
makes the diagram commute. These are given by
\begin{equation}
    v_{-2} = \bmat -\lm_4 & 0 \\ -\lm_3 & 0 \\ \lm_2 & 0 \\ \lm_1 & 0 \emat,
    \quad
    v_{-1} = \bmat 0 & 0 & 0 & 0 \\ \lm_1 & \lm_2 & -\lm_3 & -\lm_4 \emat .
\end{equation}
Forming the cone of the map $v : C^{\bullet}(\pzA_2)[-1] \rightarrow
C^{\bullet}(\pzA_1)$ we get the complex associated to the graph in figure
\ref{n4_ext1}. Note that the cohomology class is determined up a multiplicative constant $\alpha \in \kk$, which recovers the
$Q$-continuum of multiplets of \cite{Hubsch13}.

The extension class that leads to the complex in figure \ref{n3_ext2} can be
read off similarly. It is described by two maps
\begin{equation}
    \begin{tikzcd}[ampersand replacement=\&, row sep=3em]
        v_{-1}: R^4 \arrow[r, "{\tiny{\bmat
            0 & 0 & 0 & 0 \\
            0 & 0 & 0 & 0 \\
            -\lm_{2} & \lm_{3} & -\lm_{1} & 0
        \emat}}"] \&[5em] R^3\\
    v_{-2} : R^{2} \arrow[r, "{\tiny {\bmat
        0 & 0 \\
        \lm_{3} & 0 \\
        -\lm_{2} & 0 \\
        -\lm_{1} & 0
    \emat}}"] \& R^4.
    \end{tikzcd}
\end{equation}
One can easily check that this defines a chain map of degree $1$ from the complex
having rank sequence $(2,4,2)$ to the one having rank sequence $(3,4,1)$.

%% file: S6BeyondAdinkras.tex
\section{Beyond Adinkras} \label{beyond}

We now turn toward more speculative directions.  First, we explore the relation between
supersymmetric quantum mechanics and higher-dimensional supersymmetric theories, by looking at
multiplets arising from dimensional reduction. In a different direction, we construct further
multiplets by $R$-symmetry breaking -- this relies on our embedding result from section \ref{1_7_7_1}.
In this context, we also make contact with a generalization of the notion of instantons due to
Carri\'{o}n \cite{Carrion98}, which opens interesting future research directions.

Finally, we generalize the notion of Adinkras to account also for quadratic forms which are not
positive definite, we call them $q$-Adinkras. We provide the relevant graphic computational rules,
and we put them to good use to compute an example of physical interest, the chiral superfield.

To illustrate the power of pure spinor superfield formalism, we conclude by classifying families of
multiplets in $N=4$ using projective algebraic geometry. Similar in spirit to the recent \cite{6Dmultiplets, ElliottDerived},
we classify multiplets whose derived invariants are line bundles or ACM vector bundles on the
complex $N=4$ projective nilpotence variety. Many of these multiplets do not arise from the Adinkra
construction, but they can be derived from the pure spinor formalism.

The reader is advised that the writing style, and particularly the level of mathematical precision,
is not homogeneous throughout this section. The first part is more speculative, while the second
part of this section features precise mathematical results and theorems.


\subsection{Multiplets from dimensional reduction}

Another set of multiplets can be obtained from dimensional reduction of multiplets in higher
dimensions.  As an example, the dimensional reduction of the BV complexes for vector multiplets in
$d=4$ and $d=6$ decompose into a $d=1$ topological vector multiplet and the $(3,4,1)$ and $(5,8,3)$
multiplets respectively.  They both preserve the original higher-dimensional R-symmetry, but they
also obtain new R-symmetry from the Lorentz group in the reduced dimensions.

The $d=10$ vector multiplet has a more circuitous route.  It can be reduced performing the $d=8$
$\mbox{Spin}(7)$ partial topological twist along eight of the nine transverse dimensions and then
further reducing to one dimension \cite{Berkovits93, Baulieu07}. Since the twist breaks Lorentz
symmetry to $\mbox{Spin}(7)$, the reduced multiplet only inherits a $\mbox{Spin}(7)$ R-symmetry and
only 9 supercharges remain. These multiplets, along with their preserved R-symmetry Lie algebra
$\mathfrak{r} \subset \so(N)$, are shown in Table~\ref{tab:VM}.  Remarkably, they can be obtained
from an $\mathfrak{r}$-equivariant version of vertex raising from the valise multiplet.
\begin{table}[htp]
\begin{center}
\begin{tabular}{|c|c|c|c|}
\hline
Multiplet & $\mathfrak{r} \subset \so(N)$ & $N$ & Origin \\
\hline
(3,4,1) & $\mathfrak{spin}(3)$ & 4 & $d=4$ \\
(5,8,3) & $\mathfrak{spin}(5) \times \mathfrak{su}(2)$ & 8 & $d=6$ \\
(9,16,7) & $\mathfrak{spin}(7)$ & 9 & $d=10$ \\
\hline
\end{tabular}
\end{center}
\caption{Shadows of vector multiplets}
\label{tab:VM}
\end{table}%

Finally, we can look for sub-complexes of the de Rham complex by resolving the sheaf of functions on
homogeneous spaces $G/P$.  This can be described as a dual to the BGG resolution \cite{Baston89} --
the similarity with Adinkras was previously noted in \cite{zhang14}. The number of representations
appearing is $\chi(G/P)$ for the Cayley-Rosenfeld planes is displayed in Table~\ref{tab:Cayley}
\cite{Piccinni17}.  Some of these loose threads are tied together by observing that the work on
Koszul duality patterns of \cite{Beilinson96} was in part motivated by parabolic-singular duality
for BGG resolutions.

\begin{table}[htp]
\begin{center}
\begin{tabular}{|c|c|c|c|c|}
\hline
$k$ 	& E. Cartan & $G/P$ & $dim_{\mathbb{R}}$ & $\chi(G/P)$ \\
\hline
$\mathbb{R}$ & $FII$ & $\mbox{F}_4/\mbox{Spin}(9)$ & 16 & 3 \\
\hline
$\mathbb{C}$ & $EIII$ & $\mbox{E}_6/\mbox{Spin}(10) \cdot U(1)$ & 32  & 27 \\
\hline
$\mathbb{H}$ & $EVI$ & $\mbox{E}_7/\mbox{Spin}(12) \cdot Sp(1)$ & 64 &  63 \\
\hline
$\mathbb{O}$ & $EVIII$ & $\mbox{E}_8/\mbox{Spin}(16)_{+}$ & 128 & 135 \\
\hline
\end{tabular}
\end{center}
\caption{Euler characteristics of the Cayley--Rosenfeld Projective Planes}
\label{tab:Cayley}
\end{table}%
The connection with the decomposition of the free multiplet was previously noticed in
\cite{Pengpan98, Brink02} and the Adinkra for the $d = 10$ $\mathcal{N}=1$ free superfield was
determined in \cite{Gates20} -- the Adinkra is the Hasse diagram of the complex Cayley plane
\cite{Iliev05}.

The real Cayley plane with $\chi(G/P) = 3$ suggests breaking the ($N=32$)
$\mbox{SO}(32)$ $R$-symmetry to $\mbox{Spin}(9)$.  Indeed there is a $\mbox{Spin}(9)$-equivariant
multiplet $(84, 128, 44)$ with ${N}= 16$ obtained from the light-cone reduction of
eleven-dimensional supergravity.  The reduction breaks half of the supercharges, resulting in a
representation of the ${N}=16$ $\susy$ algebra. Its fields correspond to the graviton, gravitino,
and 3-form gauge field in the eleven-dimensional theory.


\subsection{Breaking R-symmetry and generalized instantons}

As explained in the introduction of this paper\footnote{See also appendix \ref{Appendix1}}, the
R-symmetry algebra $\fr$ acts on the super-translation algebra $\ft$. This, in turn, induces an
action of $\fr$ on the ring $R = \kk[\lambda_1, \ldots, \lambda_N]$. In the present one-dimensional
case, the R-symmetry algebra is $\so(N)$, which acts on $\ft^{\vee}$ by the vector representation.
It follows that the quadratic form $q_N = \sum_i \lm_i^2$ is preserved and $\fr$ acts on $R / I$ as
well. Moreover, the cone point of $\Spec R / I$ is a fixed-point of the action by $\fr$ and the
skyscraper sheaf on the cone point is an equivariant module. The Koszul complex $C^{\bullet}(\pzA)$
hence admits an equivariant free resolution
\be
    C^{\bullet}(\pzA_K) \cong R \otimes \topwedge{\bullet} \ft
\ee
with differential $D = \sum_i \lm_i \partial_{\theta_i}$ where the $\theta_i$
are the generators of $\ft_1$. Now, let $\pzA$ be an Adinkra with a unique zero
mode and an embedding $C^{\bullet}(\pzA) \hookrightarrow C^{\bullet}(\pzA_K)$ as in theorem \ref{embedding_lemma}.
In general, the subcomplex $C^{\bullet}(\pzA)$ will only be preserved by a subalgebra
of the R-symmetries $\fr$: in other words, these multiplets break the R-symmetry group to
some subgroup.

The equivariant subcomplexes of a Koszul-type complex have been also studied in a different context,
namely in the construction of generalized instantons given in \cite{Carrion98}. Let us briefly
describe the analogous construction in our setting.

Consider a closed subgroup $G \subset \mbox{Spin}(N)$. Since $\topwedge{2} \ft$ is the adjoint
representation of $\mbox{SO}(N)$, the decomposition of $\topwedge{2} \ft$ into irreducible representations of
$\mbox{Spin}(N)$ contains the adjoint representation $\fg$ of $G$. We take the complement
$\fg^{\perp}$ of the adjoint under the decomposition and consider the complex
\be
    \label{r_breaking_complex}
    \begin{tikzcd}[column sep = 1.0em]
        \ldots \arrow[r] & (\fg \wedge (\topwedge{2} \ft))^{\perp} \otimes_{\kk} R \arrow[r] & (\fg \wedge
        (\topwedge{1} \ft))^{\perp} \otimes_{\kk} R \arrow[r] & \fg^{\perp} \otimes_{\kk} R
        \arrow[r] & \ft \otimes_{\kk} R \arrow[r] & R \arrow[r] & 0
    \end{tikzcd}
\ee
This yields by construction a complex that admits an action of $G \subset \mbox{Spin}(N)$. The pure
spinor functor $\mathcal{A}^{\bullet}$ preserves the R-symmetry representations and hence the
associated multiplet break the R-symmetry from $\mbox{Spin}(N)$ to $G$.

A set of interesting examples is related to the infamous \emph{triality} of $\mbox{Spin}(8)$. It is
known that $\mbox{Spin}(8)$ contains three conjugacy classes of $\mbox{Spin}(7)$ \cite{Varadarajan01}. One is the
canonical $\mbox{Spin}(7)$ subgroup under which the vector representation decomposes into a vector
and a trivial representation. 

Starting from one of the two special $\begin{tikzcd} \mbox{Spin}(7) \arrow[r, "\iota'"] &
\mbox{Spin}(8) \end{tikzcd}$ we get the following commutative diagram
\be
    \label{triality_subgroups}
    \begin{tikzcd}
        \mbox{SU}(2) \arrow[r, hook] \arrow[d, hook] & \mbox{SU}(3) \arrow[r, hook] \arrow[d,
        hook]
        & \mbox{G}_2 \arrow[r, hook]
        \arrow[d, hook] & \mbox{Spin}(7) \arrow[d, "\iota'", hook] \\
        \mbox{Spin}(5) \arrow[r, hook] & \mbox{Spin}(6) \arrow[r, hook] & \mbox{Spin}(7) \arrow[r,
        hook] & \mbox{Spin}(8)
    \end{tikzcd}
\ee
where on the bottom row all maps are the canonical embeddings, and each square is a
pullback square (which are intersections of subgroups in this case).
The complexes obtained by the construction \eqref{r_breaking_complex} are
summarized in table \ref{tab:carrion}.

\begin{table}[htp]
\begin{center}
\begin{tabular}{|c|c|c|c|}
\hline
    Group & $\mathfrak{g}^{\perp}$ & $(\mathfrak{g} \wedge \topwedge{1}
    \ft)^{\perp}$ & Multiplet \\
\hline
$\mbox{SU}(2)$ &  $\wedge^2_{-}$ & 0 &  (1,4,3) \\
$\mbox{SU}(2)$ &  ?  & ? &  (1,5,7,3) \\
$\mbox{SU}(3)$ &  $\omega \oplus \Lambda^{2,0}$ & $\Lambda^{3,0}$ &  (1,6,7,2) \\
$\mbox{G}_2$ &  $\mathbb{R}^7$ & $\langle \varphi \rangle$ &  (1,7,7,1) \\
$\mbox{Spin}(7)$ &  $\mathbb{R}^7$ & 0 &  (1,8,7) \\
\hline
\end{tabular}
\end{center}
\caption{Complexes from \cite{Carrion98}}
\label{tab:carrion}
\end{table}%

It is fascinating that the construction \eqref{r_breaking_complex} applied to the subgroups in
\eqref{triality_subgroups} produces fully extended Adinkras with
$5 \le N \le 8$ which can be realized as subcomplexes of the Koszul complex. For example, we
recover the description of the $(1,7,7,1)$ multiplet discussed in section \ref{1_7_7_1} directly
from R-symmetry breaking.

Finally, given a complex constructed by \eqref{r_breaking_complex}, we can construct another
$G$-equivariant complex, as the Koszul complex modulo the subcomplex \eqref{r_breaking_complex}. For
instance, consider the multiplet $(0,0,14,34,35,21, 7, 1)$ shown in figure \ref{n7_14_34} which we
constructed as the cone of the embedding $(1,7,7,1) \rightarrow C^\bullet (\pzA_K)$. In degree $2$
this multiplet carries by construction the adjoint representation of $G$, which for $\mbox{G}_2$ is
indeed $14$.

As argued above, the special multiplets one can construct via R-symmetry breaking are closely
related to the generalized instantons studied in \cite{Carrion98}.

Instantons are non-perturbative objects that play an important role in the rich dynamics of gauge
theory, and their moduli space has been intensively studied by mathematicians.  Enumerative
invariants of moduli spaces of instantons can be defined using {\it orientation data}.  In turn,
orientation data for moduli spaces of instantons can be constructed starting with an elliptic
complex $E_{\bullet}$
\begin{equation}
\xymatrix{
0 \ar[r] & \Gamma^{\infty}(E_0) \ar[r]^{D_0} &   \Gamma^{\infty}(E_1) \ar[r]^{D_1} & \cdots  \ar[r]^{D_{k-1}}  & \Gamma^{\infty}(E_k) \ar[r] & 0\\
}
\end{equation}
where $\Gamma^{\infty}$ denotes smooth sections of the vector bundles $E_i$ over a manifold $X$ with
$G$-structure \cite{Joyce20}. Donaldson's construction of invariants of oriented four-manifolds $X$
uses the Atiyah--Hitchin--Singer complex \cite{AHSselfduality}
\begin{equation}
\xymatrix{
0 \ar[r] & \Gamma^{\infty}(Ad(P) \otimes \topwedge{0} T^{\vee} X ) \ar[r]^{D_0} &   \Gamma^{\infty}(Ad(P) \otimes \topwedge{1} T^{\vee} X) \ar[r]^{D_2} &
\Gamma^{\infty}(Ad(P) \otimes \topwedge{2}_{+} T^{\vee} X) \ar[r] & 0\\
}
\end{equation}
where $P \rightarrow X$ is a principal $G$-bundle over $X$, and $G = \mbox{SU}(2)$
\cite{Donaldson87,Donaldson90}.

There are many other famous examples where these kinds of complexes appear in relation to gauge
theory and enumerative invariants. For example, Chern--Simons theory on a three-manifold can be used
to define the Casson invariant \cite{Tabues90}. Its BV complex is the de Rham complex tensored with
$ad(\mathfrak{g})$, for $\fg$ the gauge algebra \cite{CostelloRenormalization}. Further, there is a
holomorphic analog of Casson's invariant for Calabi--Yau threefolds defined by Donaldson--Thomas
\cite{DT98, Thomas00} which is closely related to the holomorphic twist of $\mathcal{N} = 1$
supersymmetric gauge theory in six-dimensions and the program pursued by Donaldson--Segal of
defining $\mbox{G}_2$-instantons \cite{Donaldson11}.

All of these elliptic complexes have associated simpler complexes (called \emph{shadows}) that can
be tensored with $Ad(P)$ to recover the elliptic complexes. In the case of the
Atiyah--Hitchin--Singer complex \cite{AHSselfduality}, the simpler complex is given by
\begin{equation}
\xymatrix{
0 \ar[r] & \topwedge{0} T^{\vee} X  \ar[r]^{D_0} &    \topwedge{1} T^{\vee} X \ar[r]^{D_2} &
    \topwedge{2}_{+} T^{\vee} X \ar[r] & 0\\
}
\end{equation}

Carri\'{o}n starts with the observation that for a manifold $X$ with $G$-structure, there is a
natural splitting
\begin{equation}
\topwedge{2}  T^{\vee}  = \mathfrak{g} \oplus \mathfrak{g}^{\perp}.
\end{equation}
Using the splitting, Carri\'{o}n then constructs elliptic complexes of the form:
\begin{equation}
\xymatrix{
0 \ar[r] & \topwedge{0} T^{\vee} X  \ar[r]^{D_0} & \topwedge{1} T^{\vee} X \ar[r]^{\; \; D_1} & \mathfrak{g}^{\perp} \ar[r]^{D_2 \qquad }
&    (\mathfrak{g} \wedge \topwedge{1} T^{\vee} X)^{\perp} \ar[r]^{\qquad D_3 } &  \cdots .\\
}
\end{equation}
These complexes are dual to the subcomplexes of the Koszul complex we have considered above. The
relationship between Carri\'{o}n complexes, supersymmetric quantum mechanics and multiplets with
large R-symmetry group is an interesting direction for future research.


\subsection{Adinkras for generic quadratic forms}  \label{qAdinkras} We have established that
Adinkras are a convenient graphical tool to encode (monomial) matrix factorizations of the
non-degenerate positive definite quadratic form $q_N$. As discussed so far, these are inherently
real objects, which is most prominently displayed by the mod $8$ Bott-periodicity induced from the
Bott periodicity of real Clifford algebras.

In the following we extend the notion of Adinkra, to quadratic forms of the form
\begin{equation}
    \label{generalq}
    q = \sum_{i = 1}^{N_{\lm}} \lm_{i}^2 - \sum_{j = 1}^{N_{\mu}} (\mu_j^2) +
    \sum_{k = 1}^{N_{\rho}} \rho_k \rho^{*}_k \in \kk[\ft_1^{\vee}].
\end{equation}
These are non-degenerate quadratic forms for
$\dim \ft_1 = N = N_{\lm} + N_{\mu} + 2 N_{\rho}$. Note that $\rho$ and $\rho^*$
are independent variables!
To explain the form of the quadratic form, choose the
standard bilinear form
$\langle \cdot , \cdot \rangle$ on
$\ft_1$ with respect to a chosen basis, \emph{i.e.}\ we define the bilinear form by the
condition that the chosen basis (which for us it is the basis of the $Q_i$) is
orthonormal with respect to $\langle \cdot , \cdot \rangle$ and extend by
linearity. Then there exists a matrix $\mathpzc{B}$ such that the
bilinear form $B(x,y) \defeq q(x + y) - q(x) - q(y)$ can be represented by
$B(x,y) = \langle x , \mathpzc{B} y \rangle$. We say that $q$ is involutive if
$\mathpzc{B}^2 = 1$. When $q$ has the form as in \eqref{generalq}, $B$
is involutive and moreover $B$ induces an involution on the chosen set of basis
vectors of $\ft_1$ which we denote by $\mathpzc{b}$.
We can define a super-translation algebra $\ft_{q}$ with brackets
\begin{equation}
    \{ f, g \} = 2 B(f,g) H
\end{equation}
and in particular, for a basis $Q_1, \ldots, Q_N$ of $\ft_1$ such that $q$ is of
the form \eqref{generalq}, we have
\begin{equation}
    \{ Q_i , Q_j \} = \left\lbrace \begin{matrix} 2H & \quad & \textrm{if } Q_j =
        \mathpzc{b}(Q_i) \\ 0 & \quad & \textrm{else.} \phantom{Q_j = bQ_i}
    \end{matrix} \right.
\end{equation}

When $B$ is non-degenerate, $\mathpzc{B}$ defines an automorphism on $\ft_1$ and
induces an automorphism $\mathpzc{B}^{\vee}$ on $R = \Sym^\bullet \ft_1^{\vee}$. This
automorphism will play a central role in the generalization of Adinkras.
When the quadratic form is represented as in \eqref{generalq}. The isomorphism
$\mathpzc{B}^{\vee}$ induces an involution $\mathpzc{b}^{\vee}$ on the set of
variables that acts as
\begin{align*}
    \lm_i &\mapsto \lm_i \\
    \mu_j &\mapsto -\mu_j \\
    \rho_k &\mapsto \rho_k^{*} \\
    \rho_k^{*} &\mapsto \rho_k.
\end{align*}

All quadratic
forms on $\ft_1$ over $\RR$ and $\CC$ can be (non-uniquely) written in the form
\eqref{generalq} by choice of a suitable basis.
As before we start from the lowest $N$ and work our way up.

The main point of the definition of a generalized Adinkra is that the
associated complex admits a map $d^{\dagger}$ such that $(d + d^{\dagger})^2 =
q$. Then we can repeat the steps in appendix \ref{CplxAd} to define a complex
$C^{\bullet}_{q}(\pzA)$ that is in the essential image of the fully faithful
embedding
\begin{equation}
    \mbox{\sffamily{D}}^\flat(R / \langle q \rangle\mbox{-\sffamily{Mod}}) \hookrightarrow \mbox{\sffamily{D}}^\flat(R\mbox{-\sffamily{Mod}}).
\end{equation}

\paragraph{$\mathbf{N = 1}$}
Here the only new quadratic form appearing is $q = -\mu^2$. The rules of Adinkras
are the same and we get the familiar $N = 1$ Adinkras, with complex
\begin{tikzcd}
    R \arrow[r, "(\mu)"] & R
\end{tikzcd}
where $R = \kk[\mu]$. However, we are advised to use the involution $\mathpzc{b}$ to find
$d^{\dagger}$, \emph{i.e.}\ instead of just transposing the matrix, we should transpose and apply
$\mathpzc{b}$ to the set of variables\footnote{More abstractly, we should define
$(\mathpzc{B}^{\vee})^{*}C$ as the complex where an element $r \in R$ acts via multiplication via
$B(r)$. Then acting with $d^T$ on $(\mathpzc{B}^{\vee})^{*}C$ amounts to the ad hoc construction we
give without refering to the involution $\mathpzc{b}$.}.

Hence $d^{\dagger} = \bmat -\mu \emat$. This gives a matrix factorization
$(d + d^{\dagger})^2 = -\mu^2$. In general, we do not care about an overall sign
since the ideals $\langle q \rangle = \langle -q \rangle$ agree. From the point
of view of the supersymmetry algebra, we can absorb it in the definition of $H$.
This aligns well with the fact that $\cliff_0(p,q) \cong \cliff_0(q,p)$ which implies
that the theory of matrix factorizations is insensitive to an overall sign.

\paragraph{$\mathbf{N = 2}$}

In this case, there are two new quadratic forms we should consider.
The first one is $q = \lm^2 - \mu^2$. One can check that the minimal matrix
factorization over $\RR$ is given by
\begin{equation}
    \label{pq_1_1}
    d = \bmat \lm & \mu \\
    \mu  & \lm
    \emat
    \quad
    d^{\dagger} = \bmat \lm & -\mu \\
    -\mu & \lm \emat
\end{equation}
where we again obtain $d^{\dagger}$ by transposing the matrix for $d$ and
applying $\mathpzc{b}$ to the set of variables.
We can still raise vertices of
the valise Adinkra defined by the matrix factorization $d$, to obtain a complex
of length $3$. Raising a vertex from level $0$ to level $2$, means
that we delete a row in $d_0 = d$ and adding its transposed as a column to
$d_1$ with the variables transformed by $\mathpzc{b}$, e.g. for $d$ as in \eqref{pq_1_1},
raising a vertex using the lower row, gives the complex
\begin{equation}
    \begin{tikzcd}[ampersand replacement=\&]
        0 \arrow[r] \& R \arrow[r, "{\bmat -\mu \\ \lm \emat}"] \& R^2 \arrow[r,
        "{\bmat \lm & \mu \emat}"] \& R \arrow[r] \& 0
    \end{tikzcd}
\end{equation}
Hence, we recover the usual Koszul complex. These are the two types
of two-colored cycles that can appear, when we restrict an Adinkra to any
two colors.
We note, that for a two-colored cycle that extends over two levels, we need an
odd number of minus signs to have $d^2 = 0$, however for a two-colored cycles
which that is confined two one level, corresponding to \eqref{pq_1_1} we have to
require an even number of minus signs.\\
To capture this property, we need more structure.
We define an internal sign function on the set of colors $\mathpzc{c}$, which
corresponds to the sign of the involution $\mathpzc{b}$ acting on the variable
$\lm_{c}$ for a color $c \in \mathpzc{c}$ or equivalently to the sign of the
summand containing $\lm_c$ in $q$. We denote this function by
\begin{equation}
    \mathpzc{b} : \mathpzc{c} \rightarrow \{ \pm 1 \}
\end{equation}
Note, that this function is dictated by the quadratic form $q$.

Recall that in the definition \ref{N_Adinkra}, we required that for any
$2$-colored $4$-cycle $C^{(2c)}_4$, the sum of the parity $\pzp(e)$ over all edges in a
$e \in C^{(2c)}_4$ is odd mod $2$, \emph{i.e.}\ we require
\begin{equation}
    \label{2c_4}
    \sum_{e \in C_4^{(2\mathpzc{c})}} \mathpzc{p} (e) = 1 \in \mathbb{Z}_2.
\end{equation}
This condition must be altered by the internal sign.
Choose any vertex $v$ of the $2$-colored cycle, then there is precisely, one
other vertex $v'$ which is not connected to $v$ by an edge. The four-cycle
decomposes into two pairs of edges $e^{or}_1, e^{or}_2$ and $f^{or}_1, f^{or}_2$
which together connect $v$ to $v'$. The superscript ${}^{or}$ indicated that the
edges are oriented such they point towards $v'$. In the chosen
orientation, the edges can either go up or down in height by one.
If an edge points up, we define $\pzp(e^{or})$ to be $\pzp(e)$. In case, it points down
we define $\pzp(e^{or}) \defeq \mathpzc{b}(c(e)) \pzp(e)$.
With this notation replace condition \eqref{2c_4} with the new condition
\begin{equation}
    \label{2c_4_halfway}
    \pzp(e^{or}_1) \pzp(e^{or}_2) - \pzp(f^{or}_1) \pzp(f^{or}_2) = 0.
\end{equation}
It is easy to check, that this is equivalent to \eqref{2c_4} when
$\mathpzc{b}(e) =
1$ for all $e$.
This condition can be described more economically. As in construction
\ref{constr} we associate a variable $\lm(e) = \lm_{\pzc (i)}$ to an edge $e \in
\mathpzc{E}$ of color $i$.
Then for any vertex, for the chosen vertex $v$, we can write
\begin{equation}
    \label{halfway_with_vars}
    \begin{split}
        &\sum\limits_{e | \mathpzc{s}(e) = v} \left( \pzp(e) \lm(e) \cdot
        \left( \sum\limits_{e' | \mathpzc{s}(e') = \mathpzc{t}(e)}
            \pzp(e') \lm(e') + \sum\limits_{e' \neq e |
            \mathpzc{t}(e') = \mathpzc{t}(e)} \pzp(e') \pzb(\lm(e'))
        \right)
        \right)\\
        + & \sum\limits_{e|\mathpzc{t}(e) = v}
        \left( \pzp(e) \pzb(\lm(e)) \cdot
        \left( \sum\limits_{e' \neq e | \mathpzc{s}(e') = \mathpzc{s}(e)} \pzp(e')
        \lm(e') + \sum\limits_{e'| t(e') = s(e)} \pzp(e') \pzb (\lm(e'))
        \right) \right) \\
        = 0 & \in \kk [\ft_1^{\vee}]
    \end{split}
\end{equation}
where $\pzb (\lm(e)) \defeq \pzb({\pzc(e)})\lm(e)$. Note that for a given vertex
$v$ only $2$ summands in the expanded sum will be non-zero (for fixed $\lm_i,
\lm_j$) and the resulting two term sum is equivalent to \ref{2c_4_halfway}.

With this definition we can produce Adinkras for all quadratic forms of the form
\begin{equation}
    \label{no_xy}
    q = \sum_i \lm^2 - \sum_j \mu^2.
\end{equation}
Note, that for cycles running over $2$ levels, we can choose $v$ to be the
lowest and $v'$ to be the highest vertex. Then $\mathpzc{b}$ does not enter in
\eqref{2c_4_halfway} which makes obvious the fact that in this case \eqref{2c_4}
and \eqref{2c_4_halfway} are equivalent. Now one can check, that in the
Koszul Adinkra, all $2$-colored $4$-cylces are extended over $2$ levels and
hence the Koszul Adinkras (using definition \ref{N_Adinkra}) is a valid generalized
Adinkras for any $q$ of the form \ref{no_xy}. The $N = 2$ Adinkra in figure
\ref{n2_2_2_gen} has the same chromotopology, but the dashing differs.
\begin{figure}[ht!]
    \centering
    \includegraphics[width=7cm]{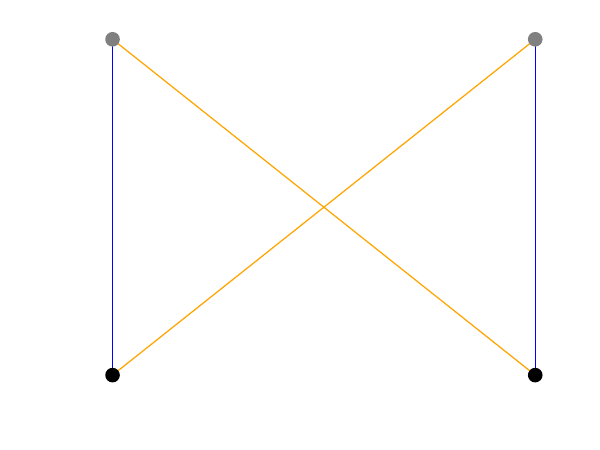}
    \caption{The $(2,2)$ Adinkras of $N = 2$ with $q = \lm^2-\mu^2$.}
    \label{n2_2_2_gen}
\end{figure}

Besides $\lm^2 - \mu^2$ we also have the quadratic form $\rho \rho^*$ for $N =
2$. These two quadratic forms are of course equivalent by the change of
variables $\rho = x + y$, $\rho^* = x - y$. However, writing the quadratic form
of signature $(1,1)$ in the form $\rho \rho^*$ makes it obvious that the minimal
matrix factorization has indeed rank one and not rank two, which one would
naively guess.
As usual we can produce a complex
\begin{equation}
    \label{wave_line_cplx}
    \begin{tikzcd}
        0 \arrow[r] & R \arrow[r, "(\rho)"] & R \arrow[r] & 0
    \end{tikzcd}
\end{equation}
This looks like the complex coming from the $N = 1$ Adinkra. There is however a
fundamental difference, \emph{i.e.}\ the complex in \eqref{wave_line_cplx} is a complex
over $R = \kk [\rho, \rho^*]$ and the corresponding cokernel module is isomorphic
to $\kk [\rho^*]$ which is a module over $\kk [\rho, \rho^*]$. For this reason, we
will treat the variables $\rho$ and $\rho^*$ differently from $\mu$ and
$\lm$.
We introduce a new type of line, that we will display as a wave lined together
with the extra datum of an arrow, i.e. we have four different possibilities
presented in figure \ref{n1_waved}.
\begin{figure}[ht!]
    \begin{minipage}{0.24\textwidth}
        \centering
        \includegraphics[height=7cm]{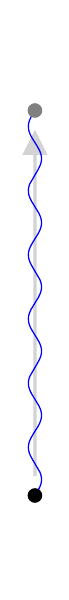}
    \end{minipage}
    \begin{minipage}{0.24\textwidth}
        \centering
        \includegraphics[height=7cm]{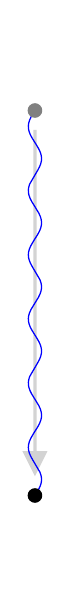}
    \end{minipage}
    \begin{minipage}{0.24\textwidth}
        \centering
        \includegraphics[height=7cm]{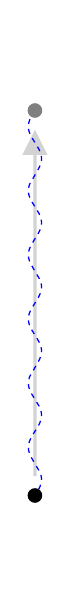}
    \end{minipage}
    \begin{minipage}{0.24\textwidth}
        \centering
        \includegraphics[height=7cm]{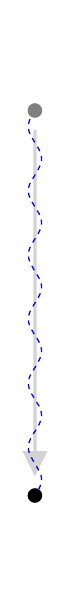}
    \end{minipage}
    \caption{Generalized Adinkras for the quadratic form $q = \rho \rho^*$
    corresponding to the matrices (from left to right): $(\rho), (\rho^*),
    (-\rho), (-\rho^*)$.}
    \label{n1_waved}
\end{figure}
As for the case of ordinary Adinkras, globally inverting the sign does not
matter, i.e. $(\rho)$ and $(-\rho)$ define the same cokernel module. The
orientation of the arrow however does matter, and the two choices correspond
to the two inequivalent $\cliff_0(1,1)$ modules. It is interesting to note, that the
matrix factorization $(\rho) (\rho^*)$ is the minimal matrix factorization of
the non-degenerate quadratic form over $\CC$. Over the complex numbers this is
equivalent to $q_2 = \lm_1^2 + \lm_2^2$ via the change of variables $\rho = \lm_1
+ i \lm_2$ and $\rho^* = \lm_1 - i \lm_2$. In fact, for even $N$,
the generalization of Adinkras (which \emph{a priori} are real objects) for
the quadratic form $q = \sum_{k = 1}^{N / 2} \rho_k \rho^*_k$ always allows to
represent a minimal matrix factorization over $\CC$.
The existence of the isomorphism
\be
    \label{Cl_1_1_iso}
    \cliff(p + 1, q + 1, \RR) \cong \cliff(p,q, \RR) \otimes \cliff(1,1, \RR)
\ee
shows, that the product of a minimal matrix factorization for some signature $(p,q)$ with the
minimal matrix factorization of $(1,1)$ gives a minimal matrix factorization of $(p + 1, q + 1)$.
For instance, the valise associated with the Koszul Adinkras with $N_{\rho}$ colors, all waved (with
consistent arrowing) gives rise to a minimal matrix factorization of the non-degenerate quadratic
form for $N = 2 N_{\rho}$. To be specific, for $N_{\rho} = 2$ and quadratic form $q = \rho_1
\rho_1^* + \rho_2 \rho_2^*$, a minimal matrix factorization is
\be
    \bmat \rho_1 & \rho_2 \\
    -\rho_2^* & \rho_1^* \emat
\ee
Note that we have essentially represented two variables by a single color and
introduced an extra datum of an arrow which is another $\ZZ / 2$ valued
function
\begin{equation}
    a : E(\mathpzc{A}) \rightarrow \{ \pm 1 \}
\end{equation}
on the set of edges of the graph with the restriction that $a(e) = 1$ for all
edges that are not waved. Taking just one color for
pairs $\rho, \rho^{*}$ allows us to keep the definition of the chromotopology of unaltered.
Before stating the full definition of an Adinkra for a quadratic from $q$ of the
form \eqref{generalq}, we must figure
out the rules for the arrows and the dashings. Indeed, the rules for the
dashings are exactly the same as in \ref{N_Adinkra} (if we were to allow terms
of the form $-\rho \rho^*$, we would need to use again the sign of the
involution $\mathpzc{b}$ on $\rho$ in order to get the correct conditions on the
dashings, however since $\rho$ and $\rho^*$ are independent variables we can
just absorb the minus sign by redefinition, without changing the real structure).

Lastly, there is a condition on the direction of arrows in $2$-colored
$4$-cycles. One can easily check that this amounts to the condition that for any
choice of pairs of edges connecting two opposite vertices $v$ and $v'$ in a
$2$-colored $4$-cylces as above, for any pair of waved edges the arrows are either both
aligned or anti-aligned with the orientation in the induced orientation from
pointing from $v$ to $v'$. We formulate this condition in yet another way: Give
any orientation to the cycle (\emph{i.e.} either clockwise or counterclockwise) then
for any pair of waved edges of the same color in the $2$-colored $4$-cycle,
the arrow of precisely one edge aligns with the chosen orientation.

We wish to combine the conditions on arrows and dashing in a single formula:
We can write a formula very similar to \ref{halfway_with_vars}, except that now
instead of applying $\pzc{b}$ if the edge is going down, we apply $\pzc{b}$ if
we go along an edge against the direction of its arrow. Note, that for unwaved
edges, this is equivalent by our conventions for arrows along unwaved edges. We
can write the condition on arrows and dashing as follows: For any vertex $v$ of
$\pzA$
\be
    \begin{split}
        &\sum_{e | t(e) = v} \left( \pzp(e) \pzb^{\frac{1 - a(e)}{2}} (\lm(e))
                    \cdot \left( \sum_{e' | t(e') = s(e)} \pzp(e') \pzb^{\frac{1 -
                    a(e)}{2}} (\lm(e'))
                    + \sum_{\substack{ e' \neq e \\  s(e') =
                    s(e) }} \pzp(e') \pzb^{\frac{1 + a(e)}{2}} (\lm(e'))
                    \right) \right)\\
        + &\sum_{e | s(e) = v} \left( \pzp(e) (\pzb^{\frac{1 + a(e)}{2}} (\lm(e)) \cdot
                    \left( \sum_{\substack{e' \neq e \\ t(e') = t(e)}}
                        \pzp(e') \pzb^{\frac{1 -
                    a(e)}{2}} (\lm(e'))
                    + \sum_{e' | s(e') = t(e)} \pzp(e') \pzb^{\frac{1 + a(e)}{2}} (\lm(e'))
                    \right) \right) \\
        = & \; 0 \in k[\ft_1^{\vee}].\\
    \end{split}
\ee

For clarity, we state the full definition of an Adinkra for a
quadratic form of the form
\begin{equation}
        q = \sum_{i = 1}^{N_{\lm}} \lm_{i}^2 - \sum_{j = 1}^{N_{\mu}} (\mu_j^2) +
                \sum_{k = 1}^{N_{\rho}} \rho_k \rho^{*}_k
\end{equation}
which we will from now on just all a $q$-Adinkras:
\begin{definition}[$q$-Adinkra]
    Let $q$ be as above. Let $\pzb$ be involution on the set of
    variables defined by $q$.
    A $q$-Adinkra is a finite simple graph $\pzA$
    equipped with two function
    $| \cdot | : \pzV \rightarrow \ZZ / 2$ (parity or \enquote{statistics}) and
    $\pzh : \pzV \rightarrow \ZZ$ (degree or \enquote{dimension})
    on the set $\pzV$ of vertices, and three functions
    $\pzc : \pzE \rightarrow c_{\lm} \coprod c_{\mu} \coprod c_{\lm} \defeq
    \{ 1, \ldots N_{\lm} \} \coprod
    \{ 1, \ldots N_{\mu} \} \coprod
    \{ 1, \ldots N_{\rho} \} \cong \{ 1, \ldots, N_{\mathpzc{red}} \}$ (the \enquote{color}),
    $\pzp : \pzE \rightarrow \ZZ / 2$ (the \enquote{dashing}
    and $a: \pzE  \rightarrow \ZZ / 2$ (the \enquote{arrow}),
    subject to the following conditions:
    \begin{enumerate}[leftmargin=*,topsep=0pt,itemsep=-1ex,partopsep=1ex,parsep=1ex]
        \item Edges are odd with respect to $\lvert \cdot \rvert$, \emph{i.e.}, there
            can be an edge between $v, w\in \pzV$ only if $|v| \neq |w|$. In
            particular, $\pzA$ is bipartite. We refer to the vertices in
            $\pzV_0:=\{ v \in \pzV \mid |v| = 0\}$ as \emph{bosonic}, and
            verticies in $\pzV_1\defeq \{v \in \pzV\mid |v|=1 \}$ as
            \emph{fermionic}.
        \item Edges are unimodular with respect to $\pzh$, \emph{i.e.}, there can be an
            edge $e \in \pzE$ between $v,w \in \pzV$ only if $| \pzh(w) -
            \pzh(v) = 1$. In this case, assuming $\pzh(w) = \pzh(v) + 1$,
            we call $v = \pzt(e)$ the target, and $w = \pzs(e)$ the source of
            $e$.
        \item $\pzA$ is \enquote{$N_{\mathpzc{red}}$-color-regular}, \emph{i.e}. every
            vertex has exactly one incident edge of each color.
        \item For edges, with color in $c_{\lm} \coprod c_{\mu}$ the arrow
            always points up, \emph{i.e.} $a(e) = 1$. Edges with color in $c_{\rho}$
            are called \enquote{waved}. Arrows on waved colors can point up or down.
        \item Fix the function
            \be
                \lm(e) : c_{\lm} \coprod c_{\mu} \coprod c_{\lm} \rightarrow \{
                    \lm_1, \ldots \lm_{N_{\lm}}, \mu_1, \ldots, \mu_{N_{\mu}},
                    \rho_1, \rho_2, \ldots, \rho_{N_{\rho}} \}.
            \ee
            that maps $c \in \mathpzc{c}_{\lm}$ to $\lm(c) = \lm_c$,
            $c \in \mathpzc{c}_{\mu}$ to $\lm(c) = \mu_c$ and
            $c \in \mathpzc{c}_{\rho}$ to $\lm(c) = \rho_c$
            Then, for each vertex
            \be
                \begin{split}
                    &\sum_{e | t(e) = v} \left( \pzp(e) \pzb^{\frac{1 - a(e)}{2}} (\lm(e))
                    \cdot \left( \sum_{e' | t(e') = s(e)} \pzp(e') \pzb^{\frac{1 -
                    a(e)}{2}} (\lm(e'))
                    + \sum_{\substack{ e' \neq e \\  s(e') =
                    s(e) }} \pzp(e') \pzb^{\frac{1 + a(e)}{2}} (\lm(e'))
                    \right) \right)\\
                +   &\sum_{e | s(e) = v} \left( \pzp(e) (\pzb^{\frac{1 + a(e)}{2}} (\lm(e)) \cdot
                    \left( \sum_{\substack{e' \neq e \\ t(e') = t(e)}}
                        \pzp(e') \pzb^{\frac{1 -
                    a(e)}{2}} (\lm(e'))
                    + \sum_{e' | s(e') = t(e)} \pzp(e') \pzb^{\frac{1 + a(e)}{2}} (\lm(e'))
                    \right) \right) \\
                =   & \; 0 \in k[\ft_1^{\vee}].\\
                \end{split}
            \ee
            This implies that for every pair $i,j \in c_{\lm} \coprod c_{\mu} \coprod c_{\rho}$
            of distinct colors, the set of edges colored $i$ and $j$ form a
            disjoint union of $4$-cycles, which we denote by $\coprod C^{(2
            \pzc)}_4$.
    \end{enumerate}
\end{definition}

\paragraph{${N = 4}$ Chiral Superfields} In the case $N = 4$ we get the following inclusions of
nilpotence varieties $Y(d; \calN)$ with $d$ the space-time dimension and $\calN$ the number of
supersymmetries, see \cite{Eager18}:
\begin{equation} \label{emb}
    Y(4; 1) \subset Y(2; 2,2) \subset Y(1; 4)
\end{equation}
where \begin{equation}
    \label{4d_N1_coord}
    Y(d;\calN) = \operatorname{Spec} \left(R /
    I_{d; \calN} \right)
\end{equation}
with $R = \CC[\rho_1, \rho_1^*, \rho_2, \rho_2^*]$ and
\be
    \begin{matrix}
        I_{4;1}  =&\langle  \rho_1 \rho_1^* + \rho_2 \rho_2^*, \\
                &           \; \rho_1 \rho_1^* - \rho_2 \rho_2^*, \\
                &           \; \rho_1 \rho_2^* + \rho_2 \rho_1^*, \\
                &           \; \rho_1 \rho_1^* + \rho_2 \rho_1^* \rangle \\
        I_{2;2,2} =& \langle \rho_1 \rho_1^* + \rho_2 \rho_2^*, \\
                &           \; \rho_1 \rho_1^* - \rho_2 \rho_2^* \rangle \\
        I_{1;4}  = & \langle \rho_1 \rho_1^* + \rho_2 \rho_2^* \rangle. \\
    \end{matrix}
\ee
Upon dimensional reduction, the embeddings \eqref{emb} of the varieties are obvious: one takes higher-dimensional variety and throws away some of its defining quadratic equations, then since the total number of supercharges is fixed, the higher dimensional variety is a subvariety of that one appearing via dimensional reduction.
Since we consider the nilpotence varieties as affine schemes, pushforward of
modules is just restriction of scalars with respect to the obvious projection
maps between the coordinate rings, and it is an exact functor.

It was shown in \cite{Eager22} that the chiral superfield corresponds to the
module $\CC[\rho_1^*, \rho_2^*]$ over $R / I_{4 ; 1}$, with free resolution
over $R$ being
\begin{equation}
    \begin{tikzcd}[ampersand replacement=\&]
        0 \arrow[r] \& R \arrow[r, "{\bmat -\rho_2 \\ \rho_1 \emat}"] \&
        R^2 \arrow[r, "{\bmat \rho_1 & \rho_2 \emat}"] \& R \arrow[r] \& 0
    \end{tikzcd}
\end{equation}
By pushing forward, this defines a module over $Y(1; 4)$ and it can be
represented by an Adinkra for the quadratic form $\rho_1 \rho_1^* + \rho_2
\rho_2^*$ defining the ideal $I_{1 ; 4}$ shown in figure
\ref{4d_chiral_antichiral}.
Similarly, the anti-chiral corresponds to the module $\CC[\rho_1, \rho_2]$. The
Adinkra in figure \ref{4d_chiral_antichiral} for the anti-chiral multiplet has the
same dashed chromotopology as the
chiral, but the direction of all the arrows is reversed.
By dimensional reduction these modules restrict to the chiral and
anti-chiral multiplets of $2d$ $\calN = (2,2)$.
\begin{figure}[ht!]
    \centering
    \begin{minipage}{.45\textwidth}
        \includegraphics[height=5cm]{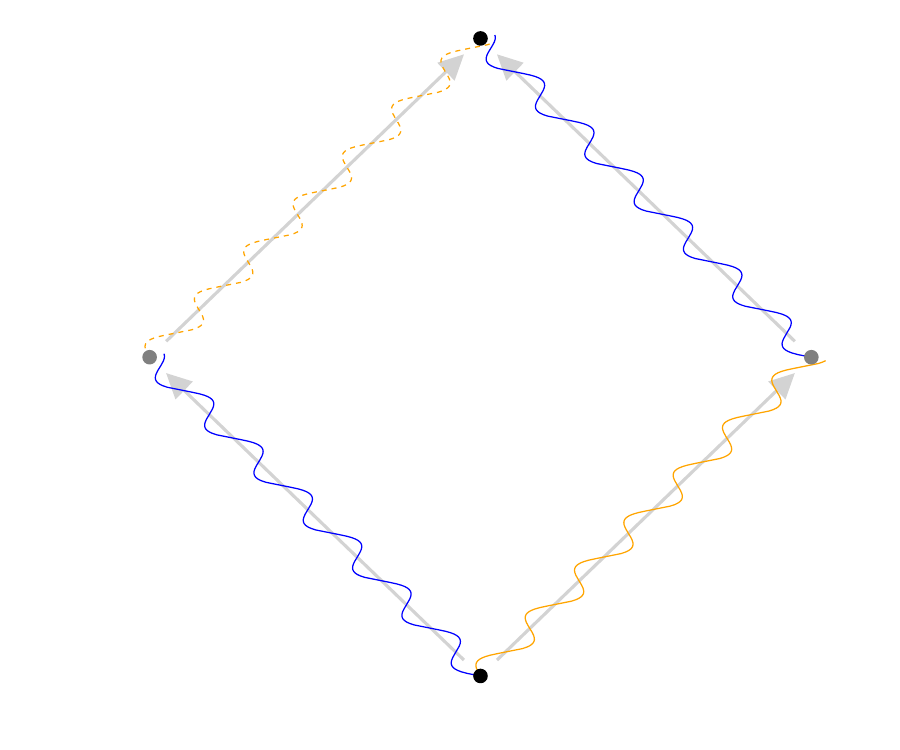}
    \end{minipage}
    \begin{minipage}{.45\textwidth}
        \includegraphics[height=5cm]{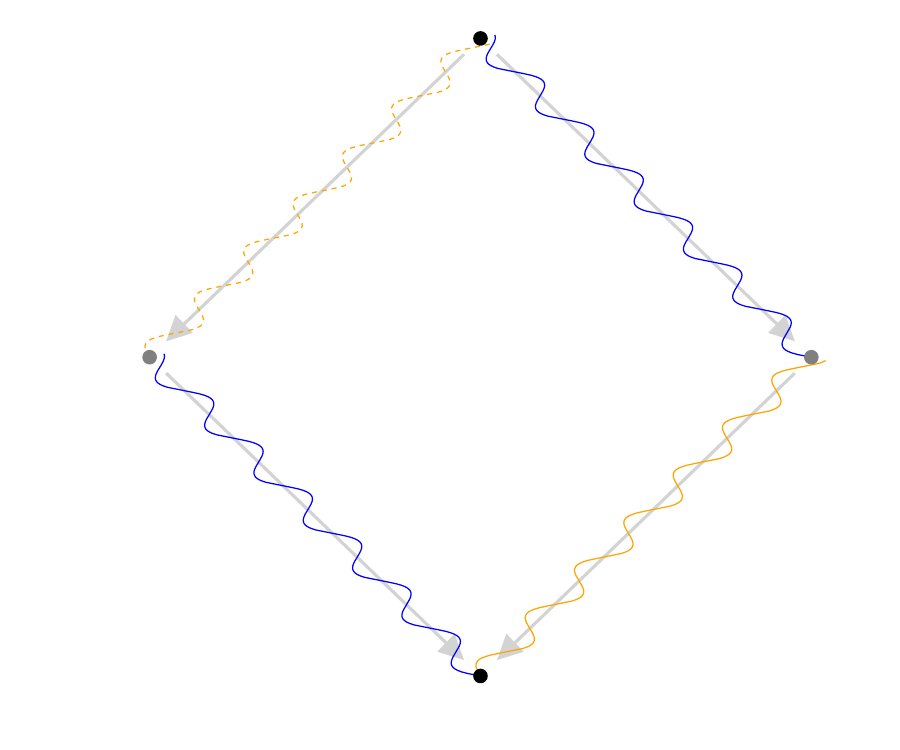}
    \end{minipage}
    \caption{Adinkras representing the $4d$ $\calN = 1$ chiral (left) and anti-chiral
    superfield (right).}
    \label{4d_chiral_antichiral}
\end{figure}

The definition of $q$-Adinkras implies that we can invert the arrows along only one of the
colors. Inverting the arrows produces two more Adinkras, shown in figure \ref{2d_twisted_chiral},
that correspond to the modules $\CC[\rho_1^*, \rho_2]$ and $\CC[\rho_1, \rho_2^*]$. These Adinkras can be
identified with the twisted chiral and twisted anti-chiral multiplet of the $2d$ $\calN = (2,2)$ supersymmetry algebra.
The modules do not lift to modules over $Y(4;1)$, \emph{i.e.} there is no twisted chiral multiplet
in four dimensions.
\begin{figure}[ht!]
    \centering
    \begin{minipage}{.45\textwidth}
        \includegraphics[height=5cm]{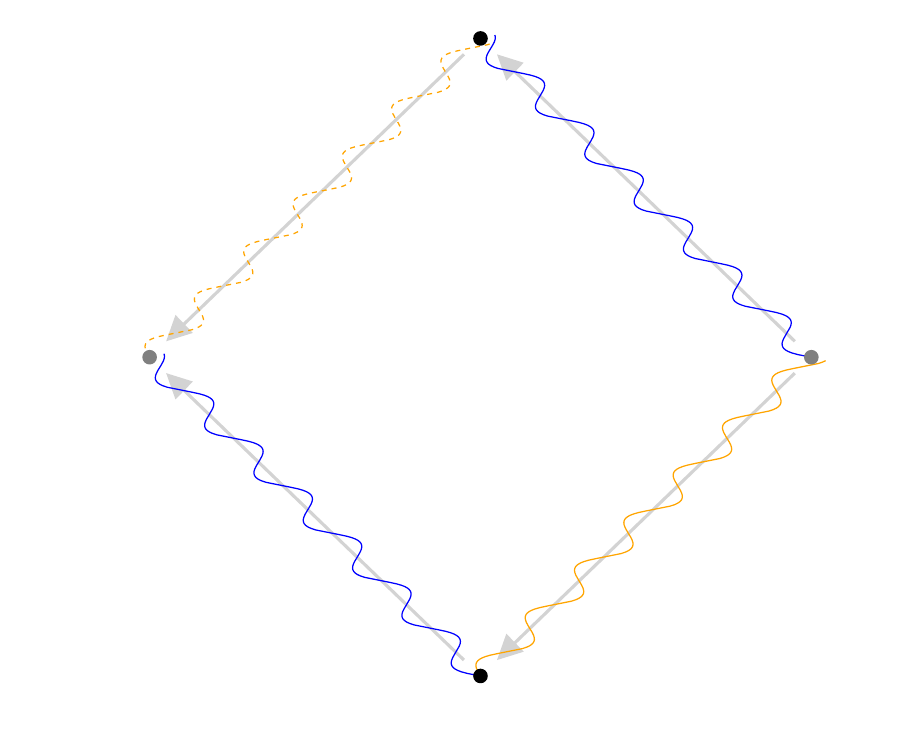}
    \end{minipage}
    \begin{minipage}{.45\textwidth}
        \includegraphics[height=5cm]{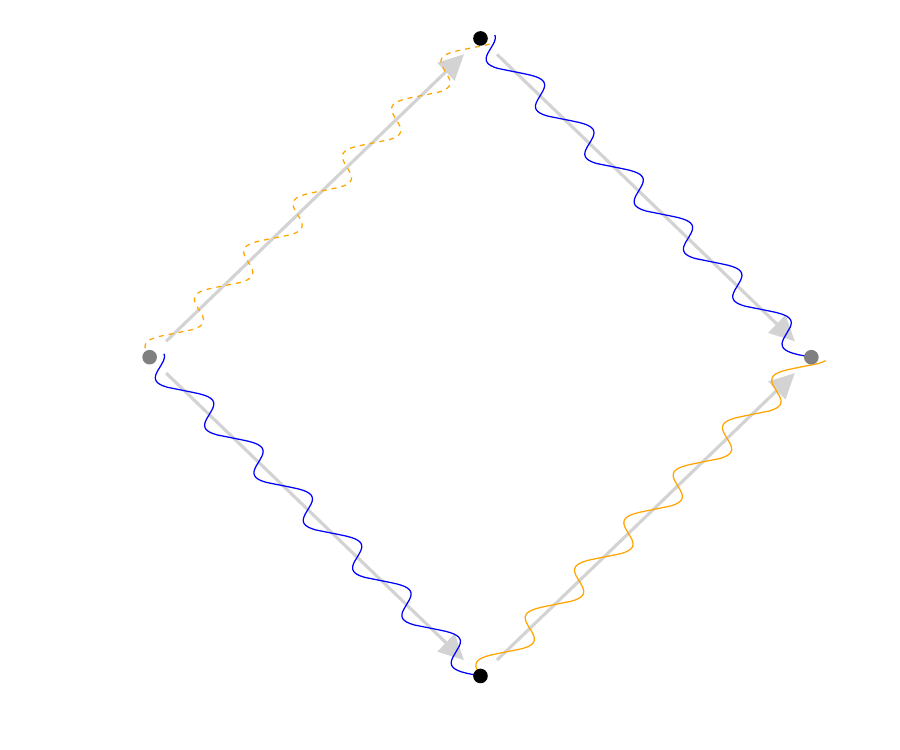}
    \end{minipage}
    \caption{Adinkras representing the $2d$ $\calN = (2,2)$ twisted chiral (left) and anti-chiral
    superfield (right).}
    \label{2d_twisted_chiral}
\end{figure}

\subsection{\texorpdfstring{$N = 4$}{N=4} multiplets from projective geometry}
\label{N=4proj}
\label{nonadinkras}

In this section, we will leave the setting of adinkraic representations, to move in a different
direction. Namely, we will show the power and effectiveness of pure spinor superfield formalism in
the context of supersymmetric quantum mechanics, by constructing families of $N=4$ multiplets coming
from geometric bundles on the related nilpotence variety. To this end, it is convenient to work with
complex projective geometry, using the methods developed in \cite{6Dmultiplets}. Working over the
complex numbers\footnote{Throughout this section, we will always leave the reference to the ground
field $\mathbb{C}$ understood.}, we recall that the projective $N=4$ nilpotence variety is a quadric
surface in $\mathbb{P}^3$ and as such it is isomorphic to the Segre variety $\Sigma_{1,1} =
\mathbb{P}^1 \times \mathbb{P}^1$ via the Segre embedding. Line bundles are easily classified: the
Picard group reads $\mbox{Pic} (\Sigma_{1,1}) \cong \mathbb{Z}\oplus \mathbb{Z}$, \emph{i.e.}\ all
line bundles are pullbacks of line bundles on one of the two copies of the projective line
$\mathbb{P}^1$. Given the natural projections
\begin{equation}
\xymatrix{
& \ar[dl]_{\pi_1}\mathbb{P}^1 \times \mathbb{P}^1 \ar[dr]^{\pi_2}& \\
\mathbb{P}^1 & & \mathbb{P}^1
}
\end{equation}
and denoting $Y_4 \defeq \mathbb{P}^1 \times \mathbb{P}^1$ for short, we write
\begin{equation}
\mathcal{O}_{Y_4} (n, m ) \defeq \pi^\ast_{1} \mathcal{O}_{\mathbb{P}^1 } (n)
    \otimes_{\mathcal{O}_{\mathbb{P}^1 \times \mathbb{P}^1}} \pi^\ast_2
    \mathcal{O}_{\mathbb{P}^1}(m).
\end{equation}
The above tensor product is sometimes denoted by $\mathcal{O}_{Y_4} (n, m) \defeq
\mathcal{O}_{\mathbb{P}^1} (n) \boxtimes \mathcal{O}_{\mathbb{P}^1} (m).$ In this setting, the two
spinor bundles are given by $\mathcal{S}_{+} \defeq \mathcal{O}_{Y_4} (-1,0)$ and $\mathcal{S}_{-}
\defeq \mathcal{O}_{Y_4} (0,-1)$.

\noindent Now, given a line bundle as above, one can obtain the associated (graded)
$\mathbb{C}[X_0,\ldots, X_3] /q_4$-module, via the functor $\Gamma_\ast$, acting as
\begin{equation} \label{gammast}
\Gamma_\ast (\mathcal{O}_{Y_4} (n,m)) \defeq \bigoplus_{d \in \mathbb{Z}} H^0 (Y_4, \mathcal{O}_{Y_4} (n+d, m+d)).
\end{equation}
We are interested in the multiplets arising from these sheaves on $Y_4$. Looking
at \eqref{gammast}, it is easy to see that considering a $d$-twist
$\mathcal{O}_{Y_4}(n+d,m+d)$ in place of $\mathcal{O}_{Y_4} (n, m)$, yields the
same module with a shift in its grading -- and hence the same multiplet with
different cohomological grading, this is for example the case of
$\mathcal{S}_{\pm}:$ indeed $\mathcal{S}_{\pm} (1) \cong
\mathcal{S}_{\mp}^\vee$ will yield the same multiplet as $\mathcal{S}_\pm$,
with a shift in degree. It follows that it suffices to consider line bundles of
the kind $\mathcal{O}_{Y_4} (n, 0)$. For non-negative $n$, one can read the
field content of the related multiplet $\mathcal{A}^\bullet (\Gamma_{\ast}
\mathcal{O}_{Y_4}(n,0))$ by reading out the Betti number of the module. This is
achieved by computing its Hilbert series. With the obvious notation
\begin{equation} \label{grdim}
\mbox{Hilb} (n,0) = \sum_{d\geq 0} (n+1+d)(d+1) t^d,
\end{equation}
where the reference to the module $\Gamma_\ast (\mathcal{O}_{Y_4} (n,0))$ is understood, and where
\begin{equation}
\dim \, H^0 (Y_4, \mathcal{O}_{Y_4} (n+d, d)) = (n+d+1) (n+d)
\end{equation}
by the K\"unneth theorem. The above Hilbert series \eqref{grdim} can be conveniently re-written as the derivative of a geometric series
\begin{equation}
\mbox{Hilb} (n,0) = \frac{\partial}{\partial t } \left ( t^{1-n} \frac{\partial}{\partial t } \sum_{d\geq 0} t^{n+d+1}\right ) = \frac{\partial}{\partial t} t^{1-n} \left ( \frac{\partial}{\partial t} \left ( \frac{t^{n+1}}{1-t}\right ) \right ).
\end{equation}
This gives
\begin{equation}
\mbox{Hilb} (n,0) = \frac{(n+1) - 2n t + (n-1)t^2}{(1-t)^4},
\end{equation}
so that
\begin{equation} \label{bettin}
\mbox{grdim} \, \mathcal{A}^\bullet (n,0) = \big [ \;  b_0 (n,0) = n+1, \quad b_1(n,0) = 2n, \quad b_2 (n,0)= n-1 \; \big ] .
\end{equation}
Equivariantly, it is not hard to read off the above resolution in terms of the representations of $\mathfrak{so}(4)^{\mathbb{C}} = \mathfrak{sl} (2) \oplus \mathfrak{sl} (2)$ appearing. Denoting the representations with $[\ell_1 | \ell_2]$, where $\ell_1$ and $\ell_2$ are the dimensions of the representation of the two copies of $\mathfrak{sl} (2)$, one has that
\begin{equation} \label{lineb}
 \, \mathcal{A}^\bullet (n,0) = \bigg [ \;  [n | 0] = \mbox{Sym}^n \mathbb{C}^2,  \; [n-1 | 1] = \mbox{Sym}^{n-1} \mathbb{C}^2 \otimes \mathbb{C}^2, \; [n-2 | 0]= \mbox{Sym}^{n-2} \mathbb{C}^2   \; \bigg ].
\end{equation}
Notice that according to \eqref{bettin}, the case corresponding to the structure sheaf is exceptional and the corresponding multiplet has only two copies of the trivial representation $[0|0]$, with opposite parity. Likewise, the case corresponding to the spinor bundles $\mathcal{S}_{Y_4, \pm}$ yields one (bosonic) copy of $[1|0]$ and one (fermionic) copy $[0|1]$ representations, for a total of 2 bosonic and two fermionic fields. Furthermore, it is important to notice that, in general, the sum of the Betti numbers of the above multiplets is not a power of two, and as such the corresponding multiplet does not arise from an Adinkra -- yet the pure spinor superfield formalism is capable to account for this kind of multiplets. \\
Following \cite{ElliottDerived} and \cite{6Dmultiplets}, the above discussion leads to the following classification result, which agrees with \ref{Koszul1}.
\begin{theorem}[$N= 4$ Multiplets from Line Bundles] Let $\susy^{N=4}_\CC$ be the complexified $N=4$ supersymmetry algebra in $d=1$. Then the following are true.
\begin{enumerate}[leftmargin=*]
\item All the multiplets of $\susy^{N=4}_{\CC}$ whose derived invariants are line bundles are of the form $\mathcal{A} (n,0)$ for some $n$ as in \eqref{lineb} up to quasi-isomorphism.
\item All the multiplets of $\susy^{N=4}_{\CC}$ whose derived invariants are ACM vector bundles $\mathcal{E}$ of rank $r$ are direct sum of multiplets of the form $\{ \mathcal{A}(0,0), \mathcal{A} (1,0) \}$ up to quasi-isomorphism, \emph{i.e.}\
\begin{equation} \label{summult}
\mathcal{A} (\mathcal{E}) \cong \bigoplus_{\ell = 1}^r \mathcal{A} (n_\ell, 0),
\end{equation}
up to permutation, for $n_\ell = \{ 0, 1\}$. In particular, $\mathcal{A}(0,0)$ and $\mathcal{A}(1,0) $ are the only multiplets whose derived invariants are indecomposable ACM vector bundles -- the structure sheaf and the spinor bundles.
\end{enumerate}
\end{theorem}
\begin{proof} The first part follows from the discussion above. For the second
    part, by Kn\"orrer's classification of ACM bundles on smooth quadric
    surfaces \cite{Knorrer}, each ACM vector bundle splits into a sum of line
    bundles and (possibly twisted) spinor bundles -- in particular, the only
    indecomposable ACM bundles are the structure sheaf and the spinor bundles
    (which yield isomorphic multiplets). Finally, since the pure spinor functor
    commutes with direct sums, the related multiplet will be of the form
    \eqref{summult} because of the first part of the theorem\footnote{Here
    we have used the definition of direct sum of multiplets as in
    \cite{ElliottDerived}.}.
\end{proof}

%% file: AppendixA.tex
\section{Super Poincar\'e algebras, nilpotence varieties and pure spinor formalism}
\label{Appendix1}

\subsection*{Super Poincar\'e Algebras.} Let $V$ be the $d$-dimensional $\kk$-vector space (for $\kk$ the real or complex numbers) interpreted as a flat spacetime. We recall that, depending on the dimension $d$, the group $\mbox{Spin} (V)$ has either one $\mathcal{S}_0$ or two $\mathcal{S}_{\pm}$ inequivalent representations, together with a pairing $\Gamma : \mathcal{S} \otimes \mathcal{S} \rightarrow V$ for a generic spinor representation $\mathcal{S}$. These representations can be used to extend the ordinary Abelian translation algebra -- identified with $V$ via $V = \langle P_1, \ldots, P_d \rangle_{k}$, with $[P_i, P_j] = 0$ -- to give the $\mathcal{N}$-extended supertranslation or supersymmetry algebra $\mathfrak{t}^{d,\mathcal{N}} \defeq V \oplus \mathfrak{p}_{{1}}^{d, \mathcal{N}},$ 
\begin{equation}
\xymatrix{
0 \ar[r] & V \ar[r] & \mathfrak{t}^{d, N} \ar[r] & \mathcal{S} \otimes U \ar[r] & 0,
}
\end{equation}
where its odd part $\mathfrak{p}_{ 1} \defeq \mathcal{S} \otimes U$ is given by the tensor product of a spin representation with an auxiliary vector space, that can be endowed with a symmetric or antisymmetric bilinear form, that enters the definition of the bracket on $\mathfrak{t}^{d, \mathcal{N}}.$ The number $N$, which is referred to as \emph{degree of extended supersymmetry}, is a multiple of the minimal possible dimension for $U$ - \emph{i.e}\ in the case $U$ is symplectic, $N = \dim U /2$. Basis elements in $\mathfrak{p}_{ 1}^{d, \mathcal{N}}$ are usually denoted with $\{ Q_i \}$ and are referred to as supercharges.\\
The $\mathcal{N}$-extended super Poincar\'e algebra $\mathfrak{p}^{d, \mathcal{N}} = \mathfrak{p}^{d, \mathcal{N}}_{ 0} \oplus \mathfrak{p}^{d, \mathcal{N}}_{ 1}$ also includes the infinitesimal automorphisms of the super translation algebra. As such, it is defined by the extension
\begin{equation}
\xymatrix{
0 \ar[r] & \mathfrak{t}^{d, \mathcal{N}} \ar[r] & \mathfrak{p}^{d, \mathcal{N}} \ar[r] & \mathfrak{aut} (\mathfrak{t}^{d, \mathcal{N}}) \ar[r] & 0,
}
\end{equation}
where, concretely, $\mathfrak{aut} (\mathfrak{t}^{d,\mathcal{N}}) = \mathfrak{so} (V)\oplus \mathfrak{d} \oplus \mathfrak{r}$. Here $\mathfrak{so} (V)$ are the Lorentz transformations, $\mathfrak{d}$ is an Abelian factor accounting for scale transformations, and $\mathfrak{r}$ is the $R$-symmetry algebra which accounts for the infinitesimal automorphisms of $U.$\\
It is convenient to lift the $\mathbb{Z}/2$-grading of the super Poincar\'e algebra to a $\mathbb{Z}$-grading, and look at super Poincar\'e algebra as a certain $\mathbb{Z}$-graded Lie algebra concentrated in degree 0, 1, 2 such that 
\begin{equation} \label{superP}
\mathfrak{p}^{d, \mathcal{N}} \defeq ( \underbrace{\mathfrak{so} (V)\oplus \mathfrak{d} \oplus \mathfrak{r}}_{\mbox{\tiny{deg 0}}} )\oplus ( \underbrace{\mathcal{S} \otimes U}_{\mbox{\tiny{deg 1}}} ) \oplus ( \underbrace{V}_{\mbox{\tiny{deg 2}}} ),
\end{equation}
together with its obvious brackets: the $\mathbb{Z}/2$-grading can be recovered by reducing modulo $2$ the above $\mathbb{Z}$-grading. Notice that, in this way, the supertranslation algebra is concentrated in positive degrees.

\subsection*{Nilpotence varieties in \texorpdfstring{$d=1$}{d=1} \& spinor bundles} \label{NVSPapp}

We now take the above super translations of super Poincar\'e algebras, which we will denote by abuse of notation by the same symbol, and we let  
\begin{equation}
R \defeq \mbox{Sym}^\bullet_{k} ( (\mathfrak{t}^{d,\mathcal{N}}_{ 1 })^\vee) 
\end{equation}
be the ring of polynomial functions on the odd part of the $d$-dimensional, $\mathcal{N}$-extended super translation algebra. Denoting with $\{ \lambda_i \}_{i = 1, \ldots, \dim \mathfrak{t}_{1}}$ the dual of the $\{ Q_i \}_{i = 1, \ldots, \dim \mathfrak{t}_{ 1}}$, we have that $R = k[\lambda_1, \ldots, \lambda_{\dim \mathfrak{t}_{ 1}}]$. If we let $Q \in \mathfrak{t}^{d,N}_{ 1 }$, the nilpotency equation $\{ Q, Q\} = 0$ defines a quadratic ideal $I$ in the polynomial ring $R$, and the quotient ring determines the ring of functions of an algebraic variety $Y$, called \emph{nilpotence variety} of $\mathfrak{p}^{d,\mathcal{N}}$. Concretely, expanding $Q = \sum_i \lambda_i Q_i$ and using the commutation relations characterizing the specific super translation algebra (these depends on $d$ and $\mathcal{N}$), the ideal is generated by 
\begin{equation} \label{psideal}
I = \sum_{i,j = 1}^{\dim \mathfrak{t}_{ 1}} \lambda_i \Gamma^\mu_{ij} \lambda_j, 
\end{equation}
if $\{ Q_i, Q_j\} = 2 \Gamma^\mu_{ij} P_\mu$. Posing $\mathbb{A}_{k}^{\dim \mathfrak{t}_{ 1}} \defeq \mbox{Spec}\, k[\lambda_1, \ldots, \lambda_{\dim \mathfrak{t}_{ 1}}]$, the following general definition can be given.
\begin{definition}[Nilpotence Variety $Y_{d,\mathcal{N}}$ of $\mathfrak{p}^{d,\mathcal{N}}$] The nilpotence variety of the $d$-dimensional, $\mathcal{N}$-extended super Poincar\'e algebra is the affine scheme $(Y_{d,\mathcal{N}}, \mathcal{O}_{Y_{d,\mathcal{N}}}) \subset \mathbb{A}^{\dim \mathfrak{t}_{1}}_k$ such that $Y_{d,N} \defeq \mbox{Spec} (R / I)$ and $\mathcal{O}_{Y_{d,\mathcal{N}}} \defeq R/I,$ with $I$ as in equation  \eqref{psideal}.
\end{definition}
\noindent Since the ideal $I$ is homogeneous, the above nilpotence variety descends to quadric hypersurface in $\mathbb{P}^{{\dim \mathfrak{t}_{ 1}} -1}_k \defeq \mbox{Proj} \, k[\lambda_1, \ldots, \lambda_{\dim \mathfrak{t}_{ 1}}]$, which we will denote again with $Y_{d,\mathcal{N}}$. \\
The case of interest for this paper is $d=1$ when $I = q_N$, the standard quadratic form $\sum_{i}^N \lambda^2_i$ -- notice that $ \dim \ft_1 = N$ in one dimension. Working over the complex numbers $\kk = \mathbb{C}$, for $N$ small, these quadric hypersurfaces $Y_{N} \defeq Y_{1, N} $ can be described in a very concrete manner.
\begin{enumerate}[leftmargin=*]
\item For $N=1$ the nilpotence variety $Y_1$ is a fat point at the origin of the affine line $\mathbb{A}^1_\mathbb{C}$, that is $Y_1 = \mbox{Spec} (\mathbb{C}[\lambda] / \langle \lambda \rangle)$.
\item For $N=2$ the nilpotence variety $Y_2 $ is given by two lines crossing at the origin in $\mathbb{A}^2_\mathbb{C}$ or two points $Y_2 = P_1 \sqcup P_2 \subset \mathbb{P}^1$ in the projective line $\mathbb{P}_{\mathbb{C}}^1$ - for example $P_1 = [1:0]$ and $P_2 = [0:1]$.  
\item For $N=3$ the projective nilpotence variety $Y_3$ is a smooth conic in the projective plane, and hence it is isomorphic to the projective line $Y_3 \cong \mathbb{P}_{\mathbb{C}}^1$ via Veronese embedding.
\item For $N=4$ the projective nilpotence variety $Y_4$ is a smooth quadric in $\mathbb{P}^3_{\mathbb{C}}$, and hence it is isomorphic to the $(1,1)$-Segre variety $\Sigma_{1,1} = \mathbb{P}^1_{\mathbb{C}} \times \mathbb{P}^1_{\mathbb{C}} \subset \mathbb{P}^3_{\mathbb{C}} $ via Segre embedding.
\item For $N=5$ the projective nilpotence variety $Y_5$ is is a smooth quadric in $\mathbb{P}^4_{\mathbb{C}}$, that can be identified with the 3-fold Lagrangian Grassmannian $LG(2,4)$ of 2-planes in a 4-dimensional complex symplectic vector space $(V^4, \omega)$ on which the $\omega |_{LG (2,4)} = 0.$ The identification can be seen as a consequence of the exceptional isomorphism $\mbox{SO}(5)\cong \mbox{Sp}(4) / \mathbb{Z}/2$.
\item For $N=6$ the projective nilpotence variety $Y_6 $ is a smooth quadric in $\mathbb{P}^5_{\mathbb{C}}$, and hence it is isomorphic to the 4-fold Grassmannian $G(2,4)$ of 2-planes in a 4-dimensional complex vector space via Pl\"ucker embedding.  
\item For $N=7$ the projective nilpotence variety $Y_7$ is a smooth quadric in $\mathbb{P}^6_{\mathbb{C}}$, which can be identified with the 5-fold orthogonal Grassmannian $OG(2,6)$, the variety of lines in a quadric 4-fold hypersurface.  The projective nilpotence variety also arises as the quotient of the exceptional Lie group $\mbox{G}_2$ by a parabolic subgroup $P_{\alpha_1}$ \cite{Ottaviani90}.
\item For $N=8$ the projective nilpotence variety $Y_8$ is a smooth quadric in $\mathbb{P}^7_{\mathbb{C}}$, which can be identified with the 6-fold orthogonal Grassmannian $OG(3,7)$, the variety of planes in quadric 5-fold hypersurface. Similarly, both components of the variety of 3-planes in a quadric 6-fold, are isomorphic to a quadric 6-folds by triality of $\mbox{SO}(8)$. 
\end{enumerate}
We remark that all nilpotence varieties $Y_N$ are smooth except for $N=1$ and $N=2$, and that, over the complex numbers, every smooth quadric hypersurface of fixed dimension is projectively equivalent: in other words, every quadric is the smooth quadric hypersurface $Y_N$ defined by $q_N$, up to a linear automorphism of $\mathbb{P}^{N-1}_{\mathbb{C}}.$ 
\vspace{.3cm}


Of great relevance to this paper are certain natural bundles that can be defined on a (smooth) quadric hypersurface, such as the nilpotence variety $Y_{1,N}$ for $N\geq3$: the spinor bundles. 
We now briefly introduce these bundles in a representation-theoretic fashion\footnote{For a more geometric definitions, we refer for example to \cite{Ottaviani88}.}.\\
Keep on working over the complex numbers, we let $Y_N \subset \mathbb{P}_{\mathbb{C}}^{N-1}$ be a smooth (projective) quadric of dimension $N-2$. The group $\mbox{SO}(N)$, and hence its double covering $\mbox{{Spin}}(N)$, acts on $Y_N$ and indeed, $Y_N$ can be seen to be a homogeneous space of $\mbox{Spin}(N).$ More precisely, fixing a maximal torus $\mathbb{T}_N \subset \mbox{Spin}(N)$ and setting $\dim Y_N = 2m+1$ or $\dim Y_N = 2m$, the group $\mbox{Spin} (N)$ has $m+1$ simple roots, that we call $\{ \alpha_1, \ldots, \alpha_{m+1} \}$, and the quadric $Y_N$ is a homogeneous space $Y_N = \mbox{Spin} (N) / P(\alpha_1)$, for $P(\alpha_1)$ maximal parabolic subgroup of $\mbox{Spin}(N).$ We denote $\mathpzc{w}_j$ the fundamental weights such that $2 \langle \mathpzc{w}_j, \alpha_i \rangle / \langle \alpha_i , \alpha_j \rangle = \delta_{ij}.$ Then, for $N = 2m$, there exists two spin representations of dimension $2^{m-1}$ of the Levi quotient of $P (\alpha_1)$ with highest weight $\lambda_{m+1}$ and $\lambda_m$ - we call them $\rho_{{\pm}} : \mbox{Spin} (N) \rightarrow GL (2^{m-1})$. For $N = 2m+1$ there is one spin representation of dimension $2^m$ of the Levi quotient $P (\alpha_1)$, with highest weight $\lambda_{m+1}$ - we call it $\rho : \mbox{Spin} (N) \rightarrow GL (2^m). $
\begin{definition}[Spinor Bundle on $Y_N$] Let $Y_N \subset \mathbb{P}^N_{\mathbb{C}}$ be a smooth $(N-2)$-dimensional quadric, and let $\mathcal{P} \defeq ( P (\alpha_1) \rightarrow \mbox{Spin} (N) \stackrel{\pi}{\rightarrow} Y_N)$ be the principal $P(\alpha_1)$-bundle over $Y_N$.
For $N$ odd, we define the spinor bundle $\mathcal{S}$ to be the rank $2^{m}$ vector bundle defined as the dual of associated bundle to $\mathcal{P}$ via the representation $\rho,$ \emph{i.e.} $\mathcal{S}^{\vee} \defeq \mathcal{P} \times_\rho \mathbb{C}^{2^{m}}$.
For $N$ even we define the spinor bundles $\mathcal{S}_{\pm}$ to be the rank $2^{m-1}$ vector bundles defined as the dual of the associated bundle $\mathcal{P}$ via the representations $\rho_{\pm}$, \emph{i.e} $\mathcal{S}^\vee \defeq \mathcal{P} \times_{\rho_{\pm}} \mathbb{C}^{2^{m-1}}$.
\end{definition}
\noindent It is easy to see that the for $N = 3$, when $Y_1 = \mathbb{P}^1_{\mathbb{C}}$, the spinor bundle is the tautological bundle, \emph{i.e.}\ $\mathcal{S} = \mathcal{O}_{\mathbb{P}^1} (-1)$ -- indeed it is unique theta characteristic of the Riemann sphere $\mathbb{P}^1$, the ``square root'' of the canonical bundle $K_{\mathbb{P}^1} = \mathcal{O}_{\mathbb{P}^1} (-2)$ -- and accordingly, for $Y_{4} = \mathbb{P}^1_{\mathbb{C}} \times \mathbb{P}^1_\mathbb{C}$ the spinor bundles are the pull-backs of the spinor bundles on the two copies of $\mathbb{P}^1_\mathbb{C}$ respectively, that is 
\begin{equation}
\mathcal{S}_+ = \mathcal{O}_{\mathbb{P}^1 \times \mathbb{P}^1} (-1,0), \qquad \mathcal{S}_- = \mathcal{O}_{\mathbb{P}^1 \times \mathbb{P}^1} (0,-1).
\end{equation}
where $\mathcal{O}_{\mathbb{P}^1 \times \mathbb{P}^1} (\ell_1, \ell_2) \defeq \mathcal{O}_{\mathbb{P}^1} \boxtimes \mathcal{O}_{\mathbb{P}^1} (\ell_2).$ 


\subsection*{Elements of the Pure Spinor Superfield Formalism} 

\label{PSapp}

\noindent The main idea of the pure spinor superfield formalism is that of obtaining multiplets\footnote{For a general definition, based of $L_{\infty}$-actions, we refer to the recent \cite{Eager22, ElliottDerived}} as byproducts of the cohomology of a certain complex, which is constructed out of the data of the nilpotence variety $Y_\mathfrak{p}$ of a super Poincar\'e algebra $\mathfrak{p}$ (in any dimension) and a ``canonical'' split supermanifold $\mathpzc{M}$, the superspacetime $\mathpzc{M}$, whose sheaf of functions reads 
\begin{equation}
\mathcal{C}^\infty (\mathpzc{M}) \defeq \mathcal{C}^\infty (\mathfrak{t}_{ 2}^\vee) \otimes_\mathbb{C} \bigwedge^\bullet \mathfrak{t}_{ 1}^\vee.
\end{equation} 
Usually, the spacetime is given local coordinates $x^\mu | \theta^\alpha$, for $\mu = 1, \ldots, \dim \mathfrak{t}_{ 2} $ and $\alpha = 1, \ldots, \dim \mathfrak{t}_{ 1},$ and elements of $\mathcal{C}^\infty(\mathpzc{M})$ are referred to as \emph{free} superfields. The superspacetime $\mathpzc{M}$ enjoys two commuting actions of $\mathfrak{t}$, we call them $(\ell, r) : \mathfrak{t} \rightarrow \mbox{End} (\mathpzc{M})$. In the coordinates $x^\mu | \theta^\alpha$, the supercharges $Q$'s are represented as follows
\begin{align} \label{landr}
&\ell (Q_\alpha) \equiv \mathcal{Q}_\alpha \defeq \frac{\partial}{\partial \theta^\alpha } - \Gamma^\mu_{\alpha \beta} \theta^\beta \frac{\partial}{\partial x^\mu}, \qquad  r (Q_\alpha) \equiv \mathcal{D}_\alpha \defeq \frac{\partial}{\partial \theta^\alpha } + \Gamma^\mu_{\alpha \beta} \theta^\beta \frac{\partial}{\partial x^\mu},
\end{align} 
where one has $\{Q_\alpha, Q_\beta\} = 2 \Gamma^\mu_{\alpha \beta} P_\mu$, for $\langle P_\mu \rangle = \mathfrak{t}_{0}$. Given a $\mathfrak{p}_{0}$-equivariant module $M$ on $Y_\mathfrak{p} $, the pure spinor complex (or $M$) is defined as follows 
\begin{equation}
( \mathcal{A}^\bullet (M) , \mathcal{D} ) \defeq ( M \otimes_{\mathbb{C}} {C}^\infty (\mathfrak{t}^\vee),\; \sum_{\alpha = 1}^{\dim \mathfrak{t}_{ 1}} \lambda_\alpha \otimes \mathcal{D}_\alpha ),
\end{equation} 
where the $\lambda$'s act via the $R/I$-module structure. This is easily seen to be a complex: indeed $\mathcal{D}^2 = 0$ upon using the fact that $I = \langle \lambda_\alpha \Gamma_{\alpha \beta}^\mu \lambda^\beta \rangle$.  Assigning $\mathbb{Z}\times \mathbb{Z}/2$-degrees as 
$
\deg (\lambda_\alpha) = (1, -),  \deg (x^\mu) = (0,+),  \deg (\theta^\alpha) = (0, -),
$
one can interpret the complex $\mathcal{A}^\bullet (M)$ as the global sections of a differentially graded super vector bundle over $\mathpzc{M}$, whose fiber over a point look 
\begin{equation}
\begin{tikzcd}
\big (\cdots \arrow[r, "\lambda_\alpha"] & M^{(k-1)} \arrow[r, "\lambda_\alpha"] & M^{(k)} \arrow[r, "\lambda_\alpha"]  & M^{(k+1)} \arrow[r, "\lambda_\alpha"] & \cdots \big ) \otimes \mathcal{C}^{\infty} (\mathpzc{M}) \arrow[loop right, "\mathcal{D}_\alpha"]
\end{tikzcd}
\end{equation}
Note that since $\{ \mathcal{D}_\alpha , \mathcal{Q}_\alpha \} = 0,$ the pure spinor complex $\mathcal{A}^\bullet (M)$ inherits a $\mathfrak{t}$-module structure, that can in turn be lifted to a full $\mathfrak{p}$-module, hence defining a \emph{multiplet} for the super Poincar\'e algebra $\mathfrak{p}.$\\
It is to be observed that the pure spinor complex $\mathcal{A}^\bullet (M)$ presents a multiplet as a cochain complex of super vector bundles over the superspacetime $\mathpzc{M}$, while usually the \emph{field content} of a multiplet is expressed in terms of sections of vector bundles over ordinary spacetime. This physical presentation can be obtained by filtering $\mathcal{A}^\bullet (M)$ so that the dependence on smooth functions factors out -- more precisely, the procedure of passing from the pure spinor multiplet $\mathcal{A}^\bullet (M)$ to the component field description amounts to finding a finite-rank ``minimal'' free resolution of the module $M$ over the polynomial ring $R$. Concretely, this can be achieved by computing the cohomology of the Koszul complex of $M$.

A natural question is whether the pure spinor superfield formalism, as presented, is capable of accounting for \emph{all} possible multiplets of a certain super Poincar\'e algebra $\mathfrak{p}$. Examples show that this is not the case - one such physically relevant example is given by the antifield multiplet in $d=4$ - and that this is related to the regularity properties of the modules (in particular the Cohen--Macaulay property). On the other hand, upon considering suitable \emph{derived} replacements in the previous construction, all the multiplets can be reached by the pure spinor formalism, as recently proved in \cite{ElliottDerived}. The idea is that of replacing $R/I$ with a complex, whose $0$-cohomology gives back $R/I$. This is nothing but the Chevalley--Eilenberg complex of the supertranslation algebra $\mathfrak{t}$, 
\begin{equation} \label{filtered}
\mbox{CE}^\bullet (\mathfrak{t}) \defeq ( \mbox{Sym}^\bullet (\mathfrak{t}^\vee[1]),\; \mbox{d}_{\mbox{\scriptsize{CE}}} ).
\end{equation}

Accordingly, instead of considering only $R/I$-modules one considers $\mbox{CE}^\bullet (\mathfrak{t})$-modules as input: geometrically, this amount to consider modules on the derived scheme $\mbox{Spec} (\mbox{CE}^\bullet (\mathfrak{t}))$ instead of modules on the nilpotence variety $Y = \mbox{Spec} (R/I).$ In particular, the following holds. 
\begin{theorem}[\cite{ElliottDerived}] \label{derivedPS} Given a super Poincar\'e algebra $\mathfrak{p} = \mathfrak{p}_0 \oplus \mathfrak{t}$, the pure spinor functor defines an equivalence of dg categories 
\begin{equation}
\mbox{\emph{CE}}^\bullet (\mathfrak{t})\mbox{\emph{-\sffamily{Mod}}}^{\mathfrak{p}_0} \cong \mbox{\emph{\sffamily{Mult}}}^{\mathpzc{strict}\mbox{\tiny{-}}\mathpzc{ob}}_{\mathfrak{g}}
\end{equation}
between the category of $\mathfrak{p}_0$-equivariant $\mbox{\emph{CE}}^\bullet (\mathfrak{t})$-modules and the full dg subcategory of the category of multiplets whose objects are strict multiplets.
\end{theorem}
\noindent The upshot of the theorem is twofold: on the one hand, from a physical point of view, it says that (up to quasi-isomorphism) all multiplets can be constructed via a suitable derived ``enhancement'' of the pure spinor formalism. On the other hand, mathematically, it relates a geometric category, to a representation-theoretic one, in a Koszul duality-like fashion.  

%% file: AppendixB.tex
\section{Elements of the Pure Spinor Superfield Formalism} \label{PSapp}

\noindent The main idea of the pure spinor superfield formalism is that of obtaining multiplets as byproducts of the cohomology of a certain complex. For us, the datum of a \emph{multiplet} is roughly that of a $\mathfrak{p}$-module structure on a collection of sections of vector bundles. More precisely, the following is adapted from \cite{Eager22, ElliottDerived} 
\begin{definition}[Multiplet] Let $\mathfrak{p}$ be any super Poincar\'e algebra of the flat spacetime $V$. Then a multiplet is a local $L_\infty$ $\mathfrak{p}$-module , \emph{i.e.} a triple $(E, D, \rho)$, where $(E, D)$ is a differentially graded super vector bundle on $V$ endowed with an action $\sigma : \mathfrak{Aff}(V)\otimes \mathcal{E} \rightarrow \mathcal{E}$ of the infinitesimal affine transformations $\mathfrak{Aff} (V) \subset \mathfrak{p}$ of $V$ on its sections $\mathcal{E} = \Gamma (V,E)$ and 
$
\rho : \mathfrak{p} \otimes \mathcal{E} \rightarrow \mathcal{E}
$ 
is a local $L_\infty$-action of $\mathfrak{p}$ on $\mathcal{E}$ such that $\rho |_{\mathfrak{Aff} (V)} = \varphi. $
\end{definition}
\noindent Multiplets can be arranged into a category: the full subcategory of the category of $\mathfrak{p}$-modules whose objects are multiplets in the sense of the previous definition. We will denote this dg category by $\mbox{\sffamily{{Mult}}} _\mathfrak{p}$. Moreover, for future use, we will denote with $\mbox{\sffamily{Mult}}^{\mathpzc{strict}\mbox{\tiny{-}}\mathpzc{ob}}_{\mathfrak{g}}$ the full dg subcategory of the category of multiplets whose objects are strict multiplets, \emph{i.e.}\ the action is a strict Lie algebra action, but morphisms can be generic $L_\infty$-maps.\\
Given these definitions, we will now review the main steps of the construction of this cochain complex and the computation of its cohomology, leaving the details to the literature.  
\begin{enumerate}[leftmargin=*]
\item The cochain complex is constructed out of the data of the nilpotence variety $Y_\mathfrak{p}$ of a super Poincar\'e algebra $\mathfrak{p}$ (in any dimension) and a ``canonical'' split supermanifold $\mathpzc{M}$, the superspacetime, which is attached to the supertranslation sub-algebra $\mathfrak{t} \defeq \mathfrak{t}_{\bar 0} \oplus \mathfrak{t}_{\bar 1} \subseteq \mathfrak{p}$, by taking 
\begin{equation}
\mathcal{C}^\infty (\mathpzc{M}) \defeq \mathcal{C}^\infty (\mathfrak{t}_{\bar 0}^\vee) \otimes_\mathbb{C} \bigwedge^\bullet \mathfrak{t}_{\bar 1}^\vee.
\end{equation}
Notice that since $\mathfrak{t}_{\bar 0}$ is generated by the spacetime translations, then $\mathcal{C}^\infty (\mathfrak{t}_{\bar 0}^\vee)$ can be interpreted as the ring of smooth functions on the spacetime on which the theory is defined - the reduced spacetime underlying $\mathpzc{M}.$ Usually, the spacetime is given local coordinates $x^\mu | \theta^\alpha$, for $\mu = 1, \ldots, \dim \mathfrak{t}_{\bar 0} $ and $\alpha = 1, \ldots, \dim \mathfrak{t}_{\bar 1},$ and elements of $\mathcal{C}^\infty(\mathpzc{M})$ are referred to as \emph{free} superfields.
\item The superspacetime $\mathpzc{M}$ enjoys two commuting actions of $\mathfrak{t}$, we call them $(\ell, r) : \mathfrak{t} \rightarrow \mbox{End} (\mathpzc{M})$. In the above coordinates, the supercharges $Q$'s in are represented as follows
\begin{align} \label{landr}
&\ell (Q_\alpha) \equiv \mathcal{Q}_\alpha \defeq \frac{\partial}{\partial \theta^\alpha } - \Gamma^\mu_{\alpha \beta} \theta^\beta \frac{\partial}{\partial x^\mu}, \\
& r (Q_\alpha) \equiv \mathcal{D}_\alpha \defeq \frac{\partial}{\partial \theta^\alpha } + \Gamma^\mu_{\alpha \beta} \theta^\beta \frac{\partial}{\partial x^\mu},
\end{align} 
where one has $\{Q_\alpha, Q_\beta\} = 2 \Gamma^\mu_{\alpha \beta} P_\mu$, for $\langle P_\mu \rangle = \mathfrak{t}_{\bar 0}$.
\item Let $M$ be a $\mathfrak{p}_{0}$-equivariant module on $Y_\mathfrak{p} $, where $\mathfrak{p}$ is seen as $\mathbb{Z}$-graded super algebra with $\mathfrak{p} = \mathfrak{p}_0 \oplus \mathfrak{t}$, \emph{i.e.}\ infinitesimal Lorentz transformations and $R$-symmetries are in degree zero, and consider the following pair
\begin{equation}
( \mathcal{A}^\bullet (M) , \mathcal{D} ) \defeq ( M \otimes_{\mathbb{C}} \mathcal{C}^\infty (\mathfrak{t}^\vee),\; \sum_{\alpha = 1}^{\dim \mathfrak{t}_{\bar 1}} \lambda_\alpha \otimes \mathcal{D}_\alpha ),
\end{equation} 
where the $\lambda$'s act via the $R/I$-module structure. This is easily seen to be a complex: indeed $\mathcal{D}^2 = 0$ upon using the fact that $I = \langle \lambda_\alpha \Gamma_{\alpha \beta}^\mu \lambda^\beta \rangle$.  
\item Assigning suitable $\mathbb{Z}\times \mathbb{Z}/2$-degrees, in particular taking 
\begin{equation}
\deg (\lambda_\alpha) = (1, -), \quad \deg (x^\mu) = (0,+) , \quad \deg (\theta^\alpha) = (0, -),
\end{equation}
one can interpret the complex $\mathcal{A}^\bullet (M)$ as the global sections of a differentially graded super vector bundle over $\mathpzc{M}$, whose fiber over a point look 
\begin{equation}
\begin{tikzcd}
\big (\cdots \arrow[r, "\lambda_\alpha"] & M^{(k-1)} \arrow[r, "\lambda_\alpha"] & M^{(k)} \arrow[r, "\lambda_\alpha"]  & M^{(k+1)} \arrow[r, "\lambda_\alpha"] & \cdots \big ) \otimes \mathcal{C}^{\infty} (\mathpzc{M}) \arrow[loop right, "\mathcal{D}_\alpha"]
\end{tikzcd}
\end{equation}
We call $\mathcal{A}^\bullet (M)$ the pure spinor superfield complex associated to $M$, or pure spinor complex of $M$ for short.
\item Since $\{ \mathcal{D}_\alpha , \mathcal{Q}_\alpha \} = 0,$ \emph{i.e.}\ the left and right actions of $\mathfrak{t}$ on $\mathcal{M}$ in \eqref{landr} are compatible, the pure spinor complex $\mathcal{A}^\bullet (M)$ inherits a $\mathfrak{t}$-module structure, that can in turn be lifted to a full $\mathfrak{p}$-module structure by verifying equivariance with respect to $\mathfrak{p}_0$. This defines a \emph{multiplet} for the super Poincar\'e algebra $\mathfrak{p}.$ 
\end{enumerate}
\noindent Notice that $\mathcal{A}^\bullet (M)$ presents a multiplet as a cochain complex of super vector bundles over the superspacetime $\mathcal{M}$, while usually the \emph{field content} of a multiplet is expressed in terms of sections of vector bundles over ordinary spacetime. This physical presentation can be obtained by filtering $\mathcal{A}^\bullet (M)$ so that the dependence on smooth functions factors out, and $\mathcal{A}^\bullet (M)$ becomes a (differentially graded) filtered vector bundle over the spacetime. This is achieved by assigning filtered weight 1 to the odd sections, \emph{i.e.} the $\theta$'s and the pure spinor variables $\lambda$'s, and zero filtered weight to the even sections $x$'s. The graded complex associated to this filtration is nothing but the Koszul complex of the module $M$ (twisted by smooth functions on the complexified spacetime $\mathpzc{M}_{\mathpzc{red}}$):
\begin{align}
\mbox{Gr}^\bullet_F ( \mathcal{A}^\bullet(M) ) = (\mathcal{K}^\bullet (M) \otimes_\mathbb{C} \mathcal{C}^\infty (\mathpzc{M}_{\mathpzc{red}}), \; \mathcal{D}_F) 
\end{align} 
where we have defined
\begin{align}
\mathcal{K}^\bullet (M) \defeq M \otimes_\mathbb{C} \bigwedge^\bullet [\theta_1, \ldots, \theta_N] ,\quad  \mathcal{D}_F \defeq \sum_{\alpha } \lambda_\alpha \otimes \frac{\partial}{\partial \theta_\alpha}. 
\end{align}
The first page of the spectral sequence of the filtered complex is hence the cohomology of the Koszul complex of $M$. The field content of the multiplet can be read off    
\begin{equation}
E_1 \defeq H^\bullet (\mathcal{K}^\bullet (M), \mathcal{D}_F) \otimes_\mathbb{C} \mathcal{C}^\infty (\mathpzc{M}_{\mathpzc{red}}), 
\end{equation}
and the page 1 differential is the BV or BRST differential. Likewise, the action of the $Q$'s descends to $E_1$ and determines the supersymmetry transformations of the physical fields. Notice that the procedure or passing from the pure spinor multiplet $\mathcal{A}^\bullet (M)$ to the component field description amounts to finding a finite-rank ``minimal'' free resolution of the module $M$ over the polynomial ring $R$.  \\
As observed early on, the nilpotence variety is cut out by homogeneous equations, and hence it descends to a projective variety. In order to exploit powerful methods from projective algebraic geometry, a \emph{projective version} of the pure spinor superfield formalism was recently developed by two of these authors in \cite{6Dmultiplets}. In this context, the input of the formalism is a vector bundle or a sheaf $\mathcal{E}$ defined on the projective nilpotence variety $\mathbb{P}Y_{\mathfrak{p}} \hookrightarrow \mathbb{P}^n$, that gets mapped to a $R$-module via the functor 
\begin{equation}
\mathcal{E} \longmapsto \Gamma_\ast (\mathcal{E}) \defeq \bigoplus_{\ell \in \mathbb{Z}} H^0 (Y, \mathcal{E} (\ell)),
\end{equation}
where $\mathcal{E} (\ell) = \mathcal{E} \otimes_{\mathcal{O}_{\mathbb{P}Y_{\mathfrak{p}}}} (\ell)$ for $\mathcal{O}_{\mathbb{P}Y_{\mathfrak{p}}} (\ell) = \iota^\ast \mathcal{O}_{\mathbb{P}^n} (\ell)$. Resolving the output $R$-module yields a multiplet as explained above. Notice, though, that the projective formalism misses some bundles: namely, multiplets that correspond to modules that are concentrated in finitely many degrees cannot be obtained from sheaves on the projective nilpotence variety. An important example of this is provided by the free superfield: this is constructed by skyscraper sheaf with value $\mathbb{C}$ at the cone point on the affine nilpotence variety, which is obviously trivial on the projective nilpotence variety. More generally, the projective pure spinor formalism loses information on sheaves having zero-dimensional support. \\

A good question is whether the pure spinor superfield formalism, as presented, is capable of accounting for \emph{all} possible multiplets of a certain super Poincar\'e algebra $\mathfrak{p}$. Examples show that this is not the case - one such physically relevant example is given by the antifield multiplet in $d=4$ - and that this is related to the regularity properties of the modules (in particular the Cohen-Macaulay property \cite{Eager22}) . On the other hand, upon considering suitable \emph{derived} replacements in the previous construction, all the multiplets can be reached by the pure spinor formalism, as recently proved in \cite{ElliottDerived}. The idea is that of replacing $R/I$ with a complex, whose $0$-cohomology gives back $R/I$. This is nothing but the Chevalley-Eilenberg complex of the supertranslation algebra $\mathfrak{t}$, 
\begin{equation} \label{filtered}
\mbox{CE}^\bullet (\mathfrak{t}) \defeq ( \mbox{Sym}^\bullet (\mathfrak{t}^\vee[1]),\; \mbox{d}_{\mbox{\scriptsize{CE}}} ).
\end{equation}
Totalizing the degree, the even generators sit in degree $-1$ and the odd generator sits in degree $0$, and the previous definition yields 
\begin{equation}
\mbox{CE}^{-p} (\mathfrak{t}) = \bigwedge^p \mathfrak{t}^\vee_{\bar 0} \otimes \mbox{Sym}^\bullet (\mathfrak{t}_{\bar 1}^\vee) = \bigwedge^p \mathfrak{t}^\vee_{\bar 0} \otimes R.
\end{equation} 
The action of the differential is given by the dual of the Lie bracket, and it can be read from the Maurer-Cartan equations of the related Lie supergroup. In particular, posing $V^\mu \defeq dx^\mu + \theta^\alpha \Gamma^\mu_{\alpha \beta} d\theta^\beta$ and $\lambda^\alpha \defeq d\theta^\alpha$, one see that
\begin{equation}
\left \{ \begin{array}{l}
\mbox{d}_{\mbox{\scriptsize{CE}}} \, V^\mu = \lambda^\alpha \Gamma^\mu_{\alpha \beta} \lambda^\beta, \\
\mbox{d}_{\mbox{\scriptsize{CE}}} \, \lambda^\alpha = 0,
\end{array}
\right.
\end{equation}
where the vielbeins $V$'s and the pure spinors $\lambda$'s are the (shifted) dual to the left-invariant vector fields $Q$'s and $P$'s generating $\mathfrak{t}.$ It is then immediate to see that 
\begin{equation}
H^0 (\mbox{CE}^\bullet (\mathfrak{t})) \cong R \big / \langle \lambda^\alpha \Gamma^\mu_{\alpha \beta} \lambda^\beta \rangle = R/I.
\end{equation}
Accordingly, instead of considering only $R/I$-modules one considers $\mbox{CE}^\bullet (\mathfrak{t})$-modules as input: geometrically, this amount to consider modules on the derived scheme $\mbox{Spec} (\mbox{CE}^\bullet (\mathfrak{t}))$ instead of modules on the nilpotence variety $Y = \mbox{Spec} (R/I).$ In particular, the following holds. 
\begin{theorem}[\cite{ElliottDerived}] \label{derivedPS} Given a super Poincar\'e algebra $\mathfrak{p} = \mathfrak{p}_0 \oplus \mathfrak{t}$, the pure spinor functor defines an equivalence of dg categories 
\begin{equation}
\mbox{\emph{CE}}^\bullet (\mathfrak{t})\mbox{\emph{-\sffamily{Mod}}}^{\mathfrak{p}_0} \cong \mbox{\emph{\sffamily{Mult}}}^{\mathpzc{strict}\mbox{\tiny{-}}\mathpzc{ob}}_{\mathfrak{g}}
\end{equation}
between the category of $\mathfrak{p}_0$-equivariant $\mbox{\emph{CE}}^\bullet (\mathfrak{t})$-modules and the full dg subcategory of the category of multiplets whose objects are strict multiplets.
\end{theorem}
\noindent The upshot of the theorem is twofold: on the one hand, from a physical point of view, it says that (up to quasi-isomorphism) all multiplets can be constructed via a suitable derived ``enhancement'' of the pure spinor formalism. On the other hand, mathematically, it relates a geometric category, to a representation-theoretic one, in a Koszul duality-like fashion.  

%% file: AppendixC.tex
\section{Complexes from Adinkras: Constructions and Proofs} \label{CplxAd}

In this appendix, we give the detailed constructions and proofs lying at the basis of the results exposed in section \ref{CplAdMain}, thus justifying definition \ref{defCpl}.
First off, we prove the that the pair $(C^\bullet (\mathpzc{A}), d)$ introduced in Construction \ref{constr} defines indeed a linear complex of $R$-modules. 
\begin{lemma} Given an Adinkra $\mathpzc{A}$, the pair $(C^\bullet(\mathpzc{A}), d )$ defines a
    linear complex of free $R$-modules. 
\end{lemma}
\begin{proof} We are left with proving that $d : C^{-i} (\mathpzc{A}) \rightarrow {C}^{-i+1} (\mathpzc{A})$ defined in equation \eqref{diffA} is a differential. One has 
\begin{equation}
d (dv) = \sum_{\mathpzc{S}(v)} ( \mathpzc{p} (e) \lambda_{\mathpzc{c} (e)} \left ( \sum_{\mathpzc{S (\mathpzc{t} (e))}} \mathpzc{p} (f) \lambda_{\mathpzc{c} (f)} \mathpzc{t} (f) \right ),
\end{equation} 
where $\mathpzc{S} (\mathpzc{t} (e)) = \{ f \in E(\mathpzc{A}) : \mathpzc{s}(f)
    = \mathpzc{t} (e) \}$. Assume that both sums are non-empty and fix an
    $e \in \mathpzc{S}(v)$ and an $f \in S(\mathpzc{t} (e))$ in the double sum
    above. The definition of Adinkras implies that $\mathpzc{c}(e) \neq \mathpzc{c} (f)$
    and hence $e$ and $f$ will be part of a certain 2-colored 4-cycle. This
    means there exists $e^\prime \in \mathpzc{S} (v)$ with $\mathpzc{c}
    (e^\prime) = \mathpzc{c} (f)$ and $f^\prime \in \mathpzc{S}(\mathpzc{t}(e^\prime))$ with $\mathpzc{t} (f^\prime) = \mathpzc{t} (f)$ and $\mathpzc{c} (f^\prime) = \mathpzc{c} (e)$ such that one gets a contribution to $d^2 (v)$ of the form
\begin{equation} \label{contrib}
d^2 (v) \owns (\mathpzc{p} (e) \mathpzc{p} (f) + \mathpzc{p}(e^\prime) \mathpzc{c}(f^\prime) ) \lambda_{\mathpzc{c} (e)} \lambda_{\mathpzc{c} (f)} \mathpzc{t} (f).
\end{equation}
Now, since $\langle e, e^\prime, f, f^\prime \rangle$ generates a 2-colored
    4-cycle by construction, it follows from axiom 4 in the definition of
    Adinkras that the sum $\mathpzc{p}(e) + \mathpzc{p}(f) + \mathpzc{p}
    (e^\prime) + \mathpzc{p} (f^\prime)$ will be odd, and hence if
    $\mathpzc{p}(e) \mathpzc{p}(f) = \pm 1$ then $\mathpzc{p}(e^\prime)
    \mathpzc{p} (f^\prime) = \mp 1,$ hence the contribution of equation
    \eqref{contrib} vanishes. By extending to all 2-colored 4-cycles containing 
    a vertex $v$, one gets that $d^2=0$ concluding the proof.
\end{proof}
As remarked in the main text, the complex $C^\bullet (\pzA)$ is defined so that it belongs to $\mbox{\sffamily{D}}^\flat (R\mbox{-\sffamily{Mod}})$. In the rest of this appendix, we will show that it can be seen to belong to the full subcategory $\mbox{\sffamily{D}}^\flat (R/ \langle q_N\rangle \mbox{-\sffamily{Mod}})$ of $\mbox{\sffamily{D}}^\flat (R\mbox{-\sffamily{Mod}})$. This says that this $C^\bullet (\pzA)$ can be seen as (coming from) a complex defined over the nilpotence variety $Y_N$. As stressed, a crucial step in this direction is to prove the existence of a Laplacian operator, that behaves on the free modules of the complex as the multiplication by the quadratic form $q_N$. Starting from definitions \ref{adj} and \ref{Lapl}, we provide the proof of lemma \ref{LaplacianLemma} in the main text.
\begin{lemma} Let $(C^\bullet (\mathpzc{A}), d)$ be the complex associated to
$\pzA$. Let $q_N \defeq \sum_{i = 1}^N \lambda^2_i$ be the standard quadratic form on $R$. Then the
Laplacian acts via multiplication by $q_N $ in $C^\bullet (\mathpzc{A})$, \emph{i.e.}\
\begin{equation}
    \Delta  = q_N \cdot \mbox{id}_{C^\bullet (\mathpzc{A})}. 
    \end{equation}
\end{lemma}
\begin{proof} The intuition behind this proof is simple. Indeed $d + d^\dagger$ maps a certain
vertex $v \in V(\mathpzc{A})$ to all the vertices connected to $v$. Applying $d+d^\dagger$ again one
reaches all the vertices $v^\prime$ of the same $\mathbb{Z}$-degree of $v$, whose distance from $v$
is either 0 or 2 edges. Now, if this distance is 0, then the vertex $v^\prime$ is $v$ itself, and
mapping two times back and forth along the same edge picks up a factor $\lambda^2_i$: since by
definition of Adinkra $v$ has exactly one edge of each color attached, this gives overall a factor
$q_N $ by $N$-regularity. Let now the distance be 2: then the vertices are connected by two edges of
different colors. For each such pair, there is a second pair that completes a 2-colored 4-cycle,
and these pairs are the same up to a relative sign so that they cancel
pairwise and they do not contribute to $\Delta v$, so that $\Delta v = q_N v.$ In formulas, one has
\begin{equation}
\Delta v =  \sum_{\mathpzc{S}(\mathpzc{s}(f))} \sum_{\mathpzc{T}(v)} 
\mathpzc{p}(e) \mathpzc{p}(f) \lambda_{\mathpzc{c}(e)}\lambda_{\mathpzc{c}(f)} \mathpzc{t}(e) +
\sum_{\mathpzc{T}(\mathpzc{t}(e^\prime))} \sum_{\mathpzc{S}(v)} 
\mathpzc{p}(e^\prime) \mathpzc{p}(f^\prime) 
\lambda_{\mathpzc{c}(e^\prime)}\lambda_{\mathpzc{c}(f^\prime)} \mathpzc{s}(e^\prime)  
\end{equation}
If $v^\prime \neq v$, then the projection to $v^\prime \subset \Delta v$ is the sum over all the
2-colored 4-cycles containing $v$ and $v^\prime, $ and hence by the axiom 4 in the definition of
Adinkra these elements vanish. On the other hand, for $v^\prime = v$, the above becomes
\begin{equation}
\Delta v = \left ( \sum_{\mathpzc{T}(v)} \mathpzc{p}(f)^2 \lambda_{\mathpzc{c}(f)}^2 + \sum_{\mathpzc{S}(v)} p(f^\prime)^2 \lambda_{c(f^\prime)}^2 \right ) v,
\end{equation} 
which equals $q_N$ by $N$-regularity and the coloring axiom 3 in the definition of Adinkras.
\end{proof}
\noindent The previous lemma leads to the following immediate corollary. 
\begin{corollary} In the previous setting, one has 
  \begin{equation} \label{LaplacianQ}
    (d+d^\dagger)^2 = d\circ d^\dagger + d^\dagger\circ d =  \Delta = q_N.
  \end{equation}
\end{corollary} 

Relying on this result, we now show that the complex of $\mathpzc{A}$ as defined above is quasi-isomorphic to
a complex of $R / I$-modules, for $I = \langle q_N \rangle $ the coordinate ring of the quadric
hypersurface cut out by the equation $q_N = \sum_{i = 1}^N \lambda^2_i = 0.$ \\
For this, we start introducing the \emph{extension of scalar complex} $C^\bullet_{R / I} (\mathpzc{A}) \defeq
C^\bullet (\mathpzc{A}) {\otimes_R} R / I$ \footnote{We remark that the above tensor product
coincides with the derived tensor product, since $C^\bullet (\pzA)$ consists of free modules.}.
Since every summand of $C^\bullet (\mathpzc{A})$ is a free $R$-module, we obtain a complex $(C_{R /
I}^\bullet (\mathpzc{A}), d_{R / I}) $ with  
\begin{equation}
    C_{R / I}^\bullet (\mathpzc{A}) =  \left ( \bigoplus_{i \geq 0} R^{\oplus
    n_i} \right ) {\otimes_R}\,  R / I \cong \bigoplus_{i \geq
    0} R / I^{\oplus n_i} ,
\end{equation}
and differential given by $d_{R / I} \defeq d \otimes 1$. Using this, we introduce a new complex 
\begin{equation}
    \widehat C^\bullet_{R / I} (\mathpzc{A}) \defeq \bigoplus_{\ell \geq 0}
    \widehat C^{-\ell}_{R / I} (\mathpzc{A}), 
\end{equation}
by considering isomorphic copies of the complex $C^\bullet_{R / I, j} \cong C^\bullet_{R / I}
(\mathpzc{A})$ for every $j\geq 0$ and posing
\begin{equation}
\widehat C^{-\ell}_{R / I} (\mathpzc{A}) \defeq \bigoplus_{\ell = i-2j}
    C^{-i}_{R / I, j}[2j].
\end{equation}
The differential is defined as 
\begin{equation} \label{dhat}
\xymatrix@R=1.5pt{
    \widehat d : \widehat C^{-\ell}_{R / I } (\mathpzc{A}) \ar[r] &
    C^{-\ell+1}_{R / I} (\mathpzc{A}) \\
    x \ar@{|->}[r] & \widehat d x \defeq (d_{R / I} + (1-\delta_{0j})d^\dagger_{R / I} ) x,
}
\end{equation}
where here $d_{R / I} : C^{-i}_{R / I,j} (\mathpzc{A}) \rightarrow C^{-i+1}_{R / I, j}$, and similarly  
\begin{equation}\xymatrix@R=1.5pt{
d^\dagger_{R/I} : C^{-i}_{R / I, j } (\mathpzc{A}) \ar[r] & C^{-i-1}_{R / I, j-1} (\mathpzc{A}) \\
x \ar@{|->}[r] & (d^\dagger \otimes 1) x,
}
\end{equation}
where we stress that for $j = 0$ we let $d^{\dagger}_{R / I}$ act as $0$ in definition \eqref{dhat}.
Keeping the subscript $R / I$ understood for the sake of notation, one has the following picture: 
\begin{equation}
\xymatrix{ 
& & & & \iddots & \iddots & \iddots & \iddots &  \iddots &\\
& & \ldots \ar[r] & C^{-4}_{j-1}\ar@{}[d]|\oplus \ar[ur] \ar[r] & C^{-3}_{j-1} \ar@{}[d]|\oplus \ar[ur]  \ar[r]^d & C^{-2}_{j-1} \ar@{}[d]|\oplus \ar[ur]  \ar[r]^d & C^{-1}_{j-1} \ar[r] \ar[ur]  & C^{0}_{j-i} \ar[ur] \ar[r] & 0 \\
& \ldots \ar[r] & C^{-3}_j \ar[ur] \ar[r] &  C^{-2}_j \ar@{}[d]|\oplus \ar[ur]^{d^\dagger}  \ar[r]^d & C^{-1}_j \ar[ur]^{d^\dagger} \ar@{}[d]|\oplus \ar[r]^d & C^{0}_j \ar[ur] \ar[r] & 0 \\
 \ldots \ar[r] & C^{-2}_{j+1} \ar[ur]  \ar[r]^d & C^{-1}_{j+1} \ar[ur]^{d^\dagger} \ar[r]^d & C^{0}_{j+1} \ar[ur]^{d^\dagger} \ar[r] & 0. \\
\iddots \ar[ur] & \iddots \ar[ur] & \iddots \ar[ur]  
}
\end{equation}
Fixing $\ell$, the $R / I$-module $\widehat C^{-\ell}_R (\mathpzc{A})$ is given by the direct sums
of the modules $C^{-i}_{j} (\mathpzc{A})$ on the same column and, leaving the $j$-index understood
with abuse of notation, we will denote an element $x \in \widehat C^{-\ell}_{R / I} (\mathpzc{A})$ 
\begin{equation} \label{vectorell}
x \defeq (x_\ell, x_{\ell -2}, \ldots, x_{\ell-4}, x_{|\ell|}) \in 
    \widehat C^{-\ell}_{R / I} (\mathpzc{A}), 
\end{equation}
where we have set a lexicographic order \emph{top-to-bottom} in the above diagram and where the last
index $|\ell |$ can be either $0$ or $1$. Note that the action of the differential $d$ on a generic
element $x\in C^{-\ell}_{R} (\mathpzc{A})$ reads 
\begin{equation} \label{diffhat}
\widehat d x = (d x_\ell + d^\dagger x_{\ell -2}, d x_{\ell
        - 2} + d^\dagger x_{\ell - 4}, \dots, dx_{|\ell| + 2} + d^\dagger x_{|\ell|}, d
    x_{|\ell|}).
\end{equation}
The following lemma makes sure that the above definitions are well-given.
\begin{lemma} The pair $\left ( \widehat{C}^\bullet_{R / I} (\mathpzc{A}) ,
    \widehat{d} \; \right )$ defines a complex of $R/I$-modules.
\end{lemma}
\begin{proof} Since $d_{R / I}^2 = 0 = (d^{\dagger}_{R / I})^{2}$, it is enough to
    observe that $\{ d_{ R / I} , d^{\dagger}_{R / I} \} = \Delta \otimes 1 = q_N
    (\lambda ) \cdot 1 \otimes 1$, and hence
    $\widehat{d}\circ \widehat d \equiv 0 \, \mbox{mod} \, \langle q_N \rangle$.
\end{proof}
\noindent In the following lemma, we characterize a feature of the cohomology of the above complex.
\begin{lemma}[Unrolling] \label{unrolling} For every element $x \in \ker
    \widehat d \cap C^{-\ell}_{R / I} (\mathpzc{A})$ of the form
    \eqref{vectorell} there exists an element $z \in \widehat C^{-\ell-1}_{R /
    I} (\mathpzc{A})$ such that $\widehat d z = (d^\dagger z_{\ell - 1}, x_{\ell
    - 2}, \ldots, x_{|\ell|})$. In particular, every cocycle of the form
    $(x_{\ell}, x_{\ell-2}, \ldots, x_{|\ell |})$ satisfying $d^{\dagger}
    x_{\ell} = 0$ is a coboundary.
\end{lemma}
\begin{proof} Let $x \in \ker \widehat d$ and let $z \in \widehat C^{-\ell-1}_{R / I} (\mathpzc{A})$
of the form $z = (0, z_{\ell - 1}, z_{\ell - 3}, \dots, z_{|\ell- 1|})$. We will now construct an
element such that  
\begin{equation}
\widehat{d} z = (d^\dagger x_{\ell}, x_{\ell - 2}, x_{\ell - 4}, \dots, x_{|\ell|}). 
\end{equation}

To this end, let first assume $|\ell| = 0$, and consider (by right-exactness of restriction of
scalars) two lifts of $x_0 \in C^{0}(\mathpzc{A})$ and $x_2 \in C^{-2}(\mathpzc{A})$. Then there
exists $z_1 \in C^{-1}(\pzA)$ such that 
\begin{equation}
d x_2 + d^\dagger x_0 = q_N z_1.
\end{equation}
This implies the relations $dd^\dagger x_0 = q_N dz_1 =  dd^\dagger d z_1$, which in turn imply
\begin{equation}
(d + d^\dagger)^2 (dz_1 - x_0) = 0,
\end{equation}
and hence, by injectivity of $(d + d^\dagger)^2$ over $R$ one has that $x_0 = dz_1$.

Next, let assume instead that $|\ell| = 1$, and consider a lift of $x_1 \in C^{-1}(\pzA)$ and $x_3
\in C^{-3}(\pzA)$ over $R$ as above. Then, there exists $z_0 \in C^{0} (\mathpzc{A})$ and $z_2 \in
C^{-2} (\mathpzc{A})$ such that 
\begin{equation}
dx_1 = q_N z_0, \qquad d^\dagger x_1 + dx_3 = q_N z_2.
\end{equation}
These lead to the relations $d^\dagger d x_1 = q_N d^\dagger z_0 = d^\dagger d d^\dagger z_0$ and
$dd^\dagger x_1 = q_N dz_2 = d d^\dagger d x_2$, which in turn imply that 
\begin{equation}
(d+d^\dagger)^2 (d^\dagger z_0 + d z_2 - x_1) = 0,
\end{equation}
and hence $x_1 = d z_2 + d^\dagger z_1$. For generic $\ell = j$, if the lift $x_j \in
C^{-j}(\mathpzc{A})$ is such that  
\begin{equation}
dx_j + d^{\dagger} x_{j - 2} = q_N z_{j - 1}, \qquad d x_{j + 2} + d^{\dagger} x_j = q_N z_{j +1},
\end{equation}
for some $z_{j-1} \in C^{-j+1} (\mathpzc{A})$ and $z_{j+1} \in {C}^{-j -1} (\mathpzc{A})$, then,
proceeding as above, one finds that 
\begin{equation}
x_j = dz_{j+1} + d^\dagger z_{j-1}.
\end{equation}
Keep going up in $\ell$ this way, it is then easy to see that only the last top term cannot be made
exact, whilst all the terms of the form $(x_{\ell}, x_{\ell -2}, \ldots, x_{|\ell|}) \in C^{-\ell}_R
(\mathpzc{A})$ with $d^{\dagger}x_l = 0$ are $\widehat d$-exact.
\end{proof}

Now, observe that there is a map $\widehat \iota : C^\bullet (\mathpzc{A}) \rightarrow C^\bullet_{R
/ I, 0} (\mathpzc{A}) \hookrightarrow \widehat{C}^\bullet_{R / I} (\mathpzc{A})$, given by the
extension of scalars followed by the embedding in $\widehat{C}^\bullet_R (\mathpzc{A})$. On elements
of the form \eqref{vectorell} it reads 
\begin{equation} 
\xymatrix@R=1.5pt{
\widehat \iota : C^{-\ell} (\mathpzc{A}) \ar[r]  & \widehat C^{-\ell}_{R / I} (\mathpzc{A})\\
x \ar@{|->}[r] & (x \, \mbox{mod}\, q_N, 0, \ldots, 0).
}
\end{equation}
Denoting by $\widehat C^\bullet_R (\mathpzc{A})$ the restriction of scalars (pullback) to $R$ of
the complex $\widehat C^{\bullet}_{R / I} (\mathpzc{A})$, we have a map $\iota : {C}^{-\ell}
(\mathpzc{A}) \rightarrow \widehat C^{-\ell}_R (\mathpzc{A})$. The restriction of
scalars defines a fully faithful embedding 
$
\mbox{\sffamily{D}}^\flat(R/I\mbox{-\sffamily{Mod}}) \hookrightarrow \mbox{\sffamily{D}}^\flat(R\mbox{-\sffamily{Mod}}).
$ and the adjunction between restriction of scalars and extension of scalars restricts to an equivalence
of categories on the essential image of the embedding. More precisely, one proves the the following theorem, which shows that the complex
$C^{\bullet}(\pzA)$ can in fact be viewed as a complex over $R/I$ as claimed above.

\begin{theorem}[Quasi-Isomorphism] The map $\iota : C^\bullet (\mathpzc{A}) \rightarrow \widehat
C^\bullet_R (\mathpzc{A})$ is a cochain map, \emph{i.e.}\ it descends to a well-defined map in
cohomology. It is a quasi-isomorphism of complexes of $R$-modules. In particular, $C^{\bullet}(\pzA)$ is in the
essential image of the embedding $\mbox{\sffamily{\emph{D}}}^\flat(R/I\mbox{-\sffamily{\emph{Mod}}})
\hookrightarrow \mbox{\sffamily{\emph{D}}}^\flat(R\mbox{-\sffamily{\emph{Mod}}})$.
\end{theorem}
\begin{proof} To start, let $x \in C^{-i} (\mathpzc{A})$. Then one has 
\begin{equation}
(\iota \circ d x - \widehat d \circ \iota ) (x) = (dx \,\mbox{mod}\, q_N - d ( x \, \mbox{mod} \, q_N ) , 0 , \ldots, 0) = 0,
\end{equation}
since $\widehat{d} \circ \iota (x)  = d_{R / I} \circ \iota(x)$, by degree reasons and $d ( x \,
\mbox{mod} \, q_N ) = dx \,\mbox{mod}\, q_N$ by definition of $d$. 

For the second part, thanks to lemma \ref{unrolling}, without loss of generality, one can restrict
to consider classes of the form $(x_l, 0, \ldots, 0) \in C^{-\ell}_{R / I} (\mathpzc{A}) \cap \ker
\widehat d$. Let us first prove that $(x_l, 0, \ldots, 0) \in \iota(\ker d)$, i.e.\ the class $x_l$
has a representative in the image of $\ker d$ under the map $C^\bullet (\mathpzc{A}) \hookrightarrow
\widehat C^\bullet_{R / I} (\mathpzc{A})$. Take any lift $x$ of $x_l$. Let $z \defeq (0, z_{\ell -
1}, 0, \dots, 0)$ where $d x = q_N z_{\ell - 1}$. Then, $\widehat{d} z = (d^\dagger z_{\ell - 1}, 0,
\dots, 0)$ since $d z_{\ell - 1} = 0$ and $q_N z_{\ell - 1} = d d^\dagger z_{\ell - 1}$. If we
consider $\tilde{x} = x - d^{\dagger} z_l$, then $d \tilde{x}= 0$ in $C^{-l}(\pzA)$. On the other
hand, $\iota(\tilde{x}) = \iota(x) - \widehat{d} (0, z_{l - 1}, 0, \ldots, 0)$. This proves that
$\ker \widehat d \subseteq \iota(\ker d) + \mbox{im}\, \widehat{d}.$ This shows the surjectivity of
the map on cohomology.

On the other hand, let $x \in \mbox{im}\, \widehat d$, with $x = \widehat d z.$ By the previous
part, we assume that $x = (x_\ell, 0, \ldots, 0)$ and we can bring $z$ into the form $(z_{\ell + 1},
z_{\ell - 1}, 0, \dots, 0)$ as before and $d z_{\ell - 1} = 0$ and $x_\ell = d z_{\ell + 1} +
d^\dagger z_{\ell - 1} + q_N f$, for any choice of lifts and some $f$. Of course, we can modify
$x_\ell$ by $(d + d^\dagger) f$ so we can assume without loss of generality that $f = 0$. Now, since
$d x_\ell = 0$, we have $d d^\dagger z_{\ell - 1} = q_N z_{\ell -1} = 0$. Thus we have $x_\ell = d
z_{\ell + 1}$, and hence $x_\ell \in \mbox{im}\, d$, proving that $\mbox{im}\, \widehat d \subseteq
\iota(\mbox{im} \, d)$. 

Assume $x \in \mbox{im}\, d$ with $x = dz$, then $x - q_N f \in \mbox{im}\, d$, for any $f \in \ker
d$. Clearly, $d (z - d^{\dagger} f) = x - q_N f$. Now consider, $x \in \ker d$ such that $\iota(x)
\in \mbox{im}\, \widehat d$. Since $\mbox{im}\, \widehat d \subseteq \iota(\mbox{im} \, d)$, there
is an element $f$, such that $x + q_N f \in \mbox{im}\, d$. Since $dx = 0$, and $d(x + q_N f) = 0$,
clearly, $d f = 0$ and hence $x = x + q_N f - q_N f \in \mbox{im}\, d$ so that $x$ is already
contained in $\mbox{im} \, d$. This shows injectivity of the map on cohomology, thus concluding the
proof.
\end{proof}

\begin{remark}
The above double complex with the differential given in equation~\ref{CEdifferential} plays
a leading role in Kapranov's proof \cite{Kapranov88,Kapranov88b} of theorem \ref{KapranovQuadrics} on the semi-orthogonal decomposition of the derived category of coherent sheaves on the quadric.  The geometric idea of ``resolving the diagonal'' appears algebraically as resolving $\kk$ as an $A\text{-}A^{!}$ bimodule.
This is the {\it problem of linear algebra} from \cite{BGG} alluded to in the abstract.
\end{remark}

%% file: AppendixD.tex
\section{Maximal Cohen--Macaulay Modules and Clifford Modules}
\label{AppendixC}

In this appendix, we construct an equivalence between the category of even Clifford modules and the subcategory of graded maximal Cohen--Macaulay modules without free summands, admitting a \emph{linear} graded free presentation 
\begin{equation}
\xymatrix{
R^n (-1) \ar[r]^{\; \varphi} & R^n \ar[r] & M \ar[r] & 0,
}
\end{equation} 
where $\varphi$ is a linear map -- these are the kind of modules we are dealing with in the present paper, associated with adinkras. We call this subcategory $\mbox{\sffamily{MCM}}^{\mbox{\scriptsize{\sffamily{lin}}}}_{\mbox{\scriptsize{\sffamily{gr}}}} (R).$
More precisely, we prove the following result.
\begin{theorem}[$C\ell (q)$-modules and MCM modules] \label{Cl_MCM} Let $q $ be a non-singular quadratic form on a $k$-vector space $V$ and let $R$ be the homogeneous coordinate ring of the quadric hypersurface defined by $q$. Then the functor 
\begin{equation}
\xymatrix@R=1.5pt{
\mathpzc{T} : \mbox{\sffamily{\emph{MCM}}}^{\mbox{\scriptsize{\sffamily{\emph{lin}}}}}_{\mbox{\scriptsize{\sffamily{\emph{gr}}}}} (R) \ar[r] & C\ell_0 (q)\mbox{-\sffamily{\emph{mod}}} \\
M \ar@{|->}[r] & \mathpzc{T} (M) \defeq k \otimes_R M 
 }
\end{equation} 
is an equivalence of categories. Under this equivalence, the first syzygy module $\mbox{Syz}_1 (M)$ of $M$ is mapped to 
\begin{equation}
\mathpzc{T} (\mbox{\emph{Syz}}_1 (M)) = C\ell_1 (q) \otimes_{C\ell_0 (q)} \mathpzc{T} (M),  
\end{equation}
and the $\mathbb{Z}/2$-graded module $\mathpzc{T} (M) \oplus \mathpzc{T} (\mbox{\emph{Syz}}_1 (M))$ is a $C\ell (q)$-module. 
\end{theorem}
\begin{proof}[Proof (a detailed sketch)] Before we start, we observe that the category $C\ell_{\bar{0}} (q)\mbox{-\sffamily{{Mod}}}$ is equivalent to the category of $\mathbb{Z}/2$-graded $C\ell (q)\mbox{-\sffamily{{Mod}}}$, provided that $q$ is non-singular. Indeed, the projection functor $M_0 \oplus M_1 \rightarrow M_0$ has an inverse given by 
\begin{equation}
M_0 \longmapsto C\ell  (q) \otimes_{C\ell_{\bar{0}} (q)} M_0 \cong M_0 \oplus ( C\ell_{\bar{1}} (q) \otimes_{C\ell_{\bar{0}} (q)} M_0).
\end{equation} 
Relying on this, let first $M = M_0 \oplus M_1$ be a $C\ell (q)$-module obtained as above from a $C\ell_{\bar 0} (q)$-module $M_0$, so that $\mbox{rank}\, M_0 = \mbox{rank}\, M_1 = m$. By the universal property of Clifford algebra, the inclusion $\iota: V \hookrightarrow C\ell (q_N)$ is a Clifford map $\iota^2 (v ) \defeq \lambda^2 = q_N (v) \cdot \mbox{id}_{C\ell (q)}$. Then the $C\ell (q)$-module structure on $M$ is a lift of a $V$-action compatible with the Clifford condition, that is if 
\begin{equation}
\xymatrix@R=1.5pt{
f : C\ell (q) \otimes M \ar[r] & M\\
v\otimes m \ar@{|->}[r] & f (v, m) \defeq v \cdot m
}
\end{equation}
we have a family of maps $V \owns v \longmapsto f_v \in \mbox{End}(V)$,
such that $f^2_v \defeq f_v \circ f_v = q_N (v) \cdot \mbox{id}_M $, since $v \cdot (v \cdot m) = (v^2) \cdot m$.  If $M$ is also $\mathbb{Z}/2$-graded, then the map $f_v$ splits into two components
\begin{equation}
f_v = (\psi_v : M_0 \longrightarrow M_1 , \; \varphi_v : M_1 \longrightarrow M_0).
\end{equation}    
Notice that $V \subset C\ell_{\bar 1} (q_N)$, so that the $|V \cdot m | = |m| + 1$, for $m$ homogeneous in $M$, which justifies the parity of the pair $(\psi_v, \varphi_v)$. 
In block forms, one has that 
\begin{equation}
f_\lambda = \left ( 
\begin{array}{c|c}
0 & \varphi_v \\
\hline
\psi_v & 0
\end{array}
\right ).
\end{equation}
Accordingly, the condition $f^2_v = q_N (v) \cdot \mbox{id}_M $ is easily seen to become 
\begin{equation}
f^2_\lambda = \left ( 
\begin{array}{c|c}
\varphi_v \circ \psi_v & 0 \\
\hline
0 & \psi_v \circ \varphi_v 
\end{array}
\right ) = q_N \cdot (\lambda) \left ( 
\begin{array}{c|c}
1 & 0 \\
\hline 
0 & 1
\end{array}
\right ).
\end{equation}
That is 
\begin{equation}
\varphi_v \circ \psi_v = q_N (v) \cdot \mbox{id}_{M_0}, \qquad \psi_v \circ \varphi_v = q_N (v) \cdot \mbox{id}_{M_1}.
\end{equation}
Finally, notice that letting $\lambda$ free, by hom-tensor adjunction one has that the map $f \in \mbox{Hom} (V \otimes M, M) \cong \mbox{Hom} (M , \mbox{Hom} (V, M) $ get decomposed as 
\begin{equation}
f = (\psi \in \mbox{Hom} (M_0, \mbox{Hom} (V, M_1)), \; \varphi \in \mbox{Hom} (M_1, \mbox{Hom} (V, M_1))\; ).
\end{equation}
In other words, $f$ can be represented as a block matrix whose non-zero entries are the off-diagonal blocks, which are both linear forms in $V^\vee$, 
\begin{equation}
f \in \left ( \begin{array}{c|c}
0 &  (\varphi)_{m\times m}[V^\vee] \in \mbox{Hom} (M_1, M_0 \otimes V^\vee) \\
\hline
(\psi)_{m\times m}[V^\vee] \in \mbox{Hom} (M_0, M_1 \otimes V^\vee) & 0
\end{array}
\right ).
\end{equation}
This shows that the non-zero odd diagonal entries, the pair of square matrices $(\psi , \varphi)$, defines a linear matrix factorization of the quadratic form $q_N$ over $\kk[V^\vee]$, 
\begin{equation}
\varphi \psi = \psi \varphi = q_N \cdot \mbox{id}_{m\times m}.
\end{equation} 
In other words any $\mathbb{Z}/2$-graded $C\ell (q_N)$-module corresponds to a linear matrix factorization of $q_N,$ and hence to a maximal Cohen--Macaulay $R$-module via the bijection $(\psi , \varphi ) \mapsto M_{(\psi, \varphi)} = \coker (\varphi)$ in \eqref{coker}. \\
To show the other direction, let $M$ be a maximal Cohen--Macaulay module $R$-module. Then its free resolution $\mathbf{F}_{(\psi, \varphi)}$ is determined by a matrix factorization $(\psi, \varphi)$ of $q$ and it is periodic of period 2, by theorem \ref{Eisenbud} -- notice that $M \cong \coker (\varphi)$. The periodicity isomorphism of the resolution is given by a map 
\begin{equation}
\xymatrix@R=1.5pt{
\mathpzc{P}: \mbox{Ext}^2_R (k, k) \otimes_R \mbox{Tor}^R_\bullet (k, M) \ar[r] & \mbox{Tor}^R_{\bullet + 2} (k, M) \\
H \otimes F^M_{\bullet} \ar@{|->}[r] & \mathpzc{P} (H \otimes F^M_{\bullet} ) \defeq H \cdot F^M_\bullet.
}
\end{equation}
where $H \in \mbox{Ext}^2_R (k,k)$ is an element of the Yoneda algebra $\mbox{Ext}^\bullet_R (k,k)$. This implies that the direct sum $M_0 \oplus M_1 \defeq \mbox{Tor}^R_0 (k, M) \oplus \mbox{Tor}^R_1 (k, M)$ has an action of the quotient algebra $\mbox{Ext}^\bullet_R (k, k) / \langle H - 1 \rangle$. On the other hand, we have seen that setting $H = 1$ in the Yoneda algebra of a quadric -- which is isomorphic to the universal enveloping algebra of the homotopy algebra induced by the quadrics -- defines a $\mathbb{Z}/2$-graded algebra isomorphism $\mbox{Ext}^\bullet_R (k, k) / \langle H -1 \rangle \cong C\ell (q)$, and hence $M_0 \oplus M_1$ defines a $\mathbb{Z}/2$-graded $C\ell (q)$-module. Its even part $M_0 = k \otimes_R M \cong M / \mathfrak{m}\cdot M$, where $\mathfrak{m}$ is the maximal ideal generated by $V^\vee$ is a $C\ell_0 (q)$-module. Its odd part $M_1$ is isomorphic to $k \otimes_R \mbox{Syz}_1 (M)$, where $\mbox{Syz}_1 (M)$ is the first syzygy module of $M$.   
\end{proof}
\noindent It turns out that if the quadratic form is regular, all linear MCM modules are linearizable in the sense of \cite{Buchweitz06}, and hence linear MCM modules are really arbitrary MCM modules - see theorem 2.2 in cite \cite{Buchweitz06}.